\documentclass{sig-alternate-05-2015}
\usepackage{amsmath}
\usepackage{amssymb}
\usepackage{amscd}
\usepackage{amsfonts}
\usepackage{graphicx}%
\usepackage{algorithm}
\usepackage[noend]{algpseudocode}
\usepackage{cite}
\usepackage{etoolbox}
\usepackage{setspace}
\usepackage{enumitem}
\usepackage{filecontents}
\usepackage{color}
\usepackage{subfig}
\usepackage[colorinlistoftodos]{todonotes}
\usepackage[
  pass,
]{geometry}

\makeatletter
\def\@copyrightspace{\relax}
\makeatother

\newtheorem{theorem}{Theorem}[section]
\newtheorem{corollary}[theorem]{Corollary}

\newtheorem{lemma}[theorem]{Lemma}
\newtheorem{proposition}[theorem]{Proposition}
\newtheorem{definition}[theorem]{Definition}

\title{On the Capacity Regions of\\Single-Channel and Multi-Channel Full-Duplex Links}
\author{\alignauthor 
Jelena Mara\v{s}evi\'{c} and Gil Zussman\\
      \affaddr{Department of Electrical Engineering,}\\ 
      \affaddr{Columbia University}\\
      \affaddr{New York, NY, 10027, USA}\\
      \email{\{jelena, gil\}@ee.columbia.edu}
} 

\graphicspath{{images/}}

\newif\iffullpaper
\fullpapertrue 

\iffullpaper
\allowdisplaybreaks
\fi

\begin{document}
\maketitle

\begin{abstract}
We study the achievable \emph{capacity regions of full-duplex links} in the \emph{single- and multi-channel cases} (in the latter case, the channels are assumed to be orthogonal -- e.g., OFDM). We present analytical results that characterize the uplink and downlink capacity region and efficient algorithms for computing rate pairs at the region's boundary. We also provide near-optimal and heuristic algorithms that ``convexify'' the capacity region when it is not convex. The convexified region corresponds to a combination of a few full-duplex rates (i.e., to time sharing between different operation modes). The algorithms can be used for theoretical characterization of the capacity region as well as for resource (time, power, and channel) allocation with the objective of maximizing the sum of the rates when one of them (uplink or downlink) must be guaranteed (e.g., due to QoS considerations). We numerically illustrate the capacity regions and the rate gains (compared to time division duplex) for various channel and cancellation scenarios. The analytical results provide insights into the properties of the full-duplex capacity region and are essential for future development of scheduling, channel allocation, and power control algorithms.
\end{abstract}

\section{Introduction}

Existing wireless systems are Half-Duplex (HD), where separating the  transmitted and received signal in either frequency or time causes inefficient utilization of the wireless resources. An emerging technology that can substantially improve spectrum efficiency is Full-Duplex (FD) wireless, namely, simultaneous transmission and reception on the same frequency channel~\cite{FullDuplex_RiceU_JSCAinvited14}.
The main challenge in implementing FD devices is the high Self-Interference (SI) caused by signal leakage from the transmitter into the receiver. The SI signal is usually many orders of magnitude higher than the desired signal at the receiver's input, requiring over 100dB of Self-Interference Cancellation (SIC).\footnote{The SI signal power has to be reduced by $10^{10}$ times.} 
Recently, several groups demonstrated that combining techniques in the analog and digital domains can provide SIC that can support practical applications (e.g., \cite{choi2010achieving, jain2011practical,  khojastepour2011case, aryafar2012midu,Zhou_WBSIC_ISSCC15,bharadia2013full, duarte2012characterization}).

\begin{figure}[t!]
\centering
\hspace{-15pt}\subfloat[]{\label{fig:asymmetric-req}\includegraphics[height = .9in]{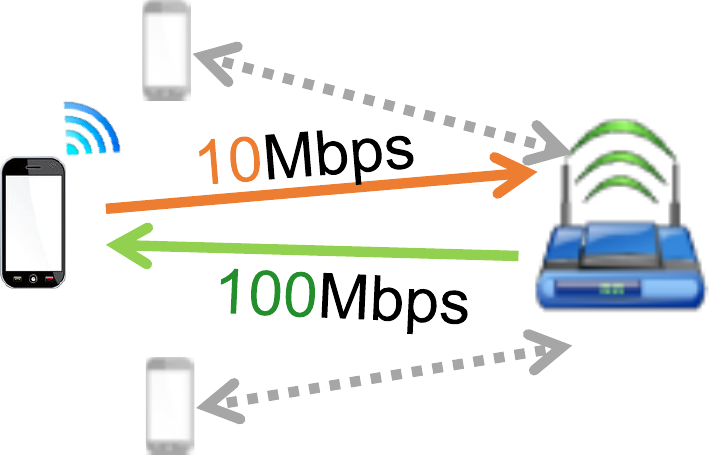}}
\hspace{20pt}\subfloat[]{\label{fig:PA}\includegraphics[height = 1in]{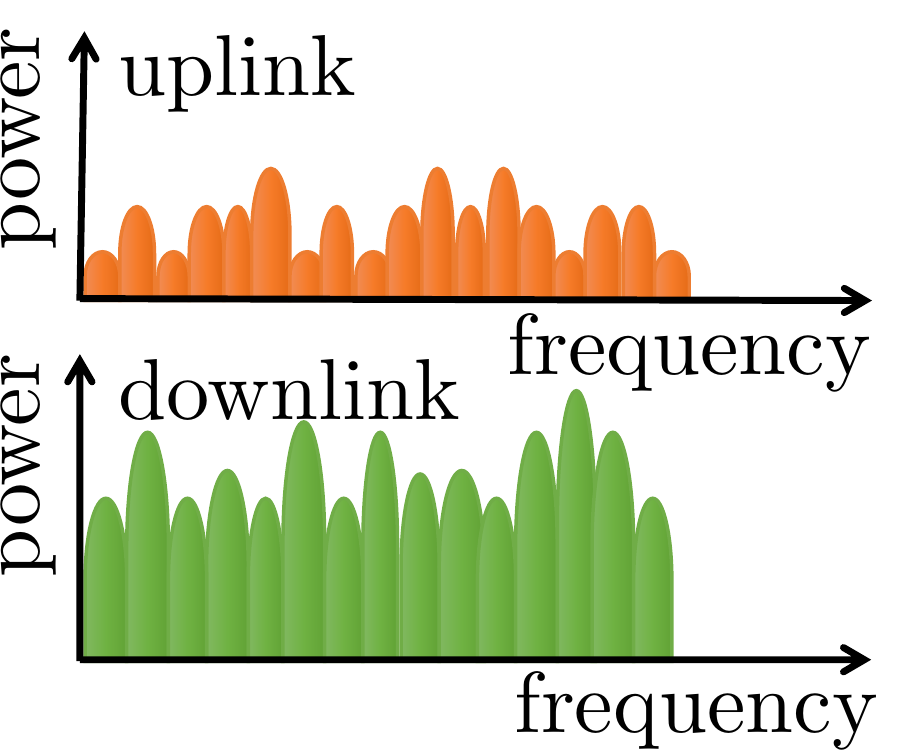}}\hspace{\fill}\\\vspace{-10pt}
\hspace{\fill}\subfloat[]{\label{fig:channels-off}\includegraphics[height = 1in]{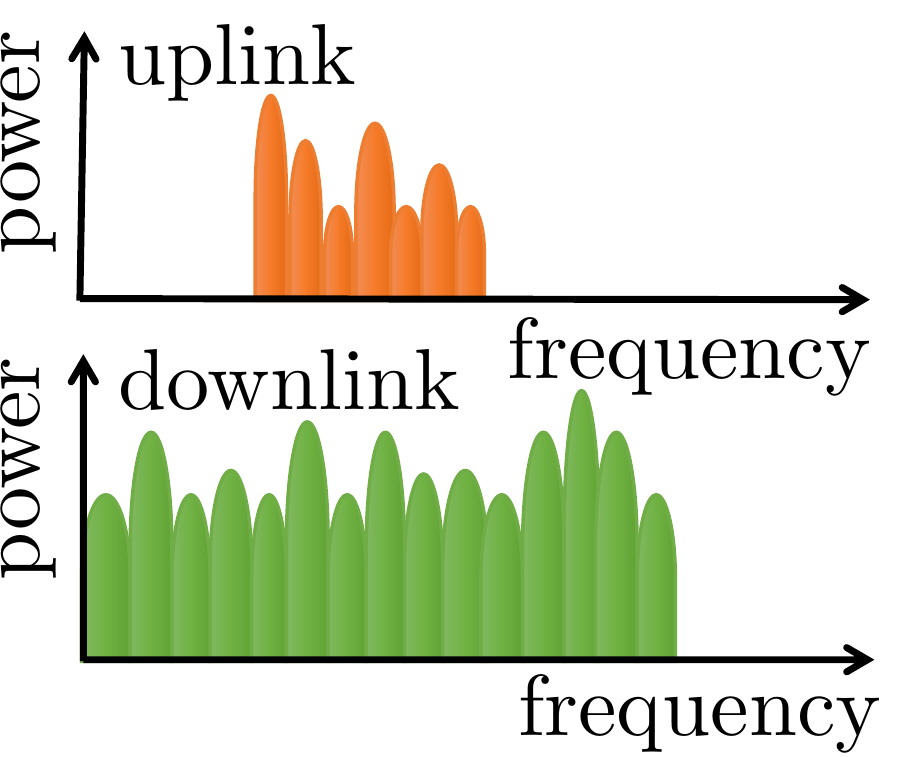}}\hspace{\fill}
\subfloat[]{\label{fig:scheduling}\includegraphics[height = 1in]{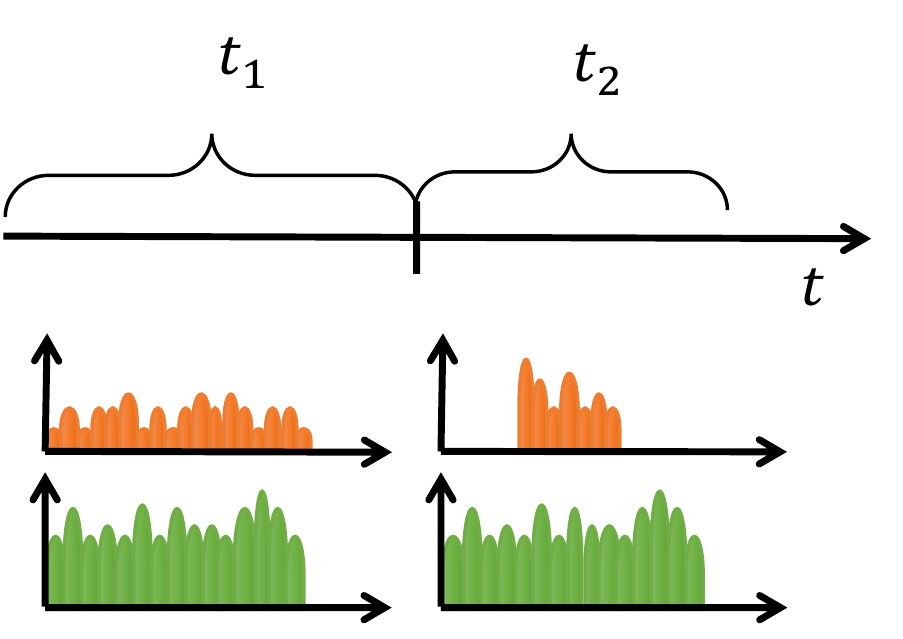}}\hspace{\fill}\vspace{-5pt}
\caption{\protect\subref{fig:asymmetric-req} An example of different rate requirements on a full-duplex link and possible policies to meet the requirements: \protect\subref{fig:PA} reduction of the the power levels on the UL channels, \protect\subref{fig:channels-off} allocation of a subset of the channels to the UL, and \protect\subref{fig:scheduling} time-sharing between two FD rate pairs (TDFD).}
\label{fig:motivation}\vspace{-10pt}
\end{figure}

The first implementations of FD receivers optimistically envisioned $2\times$ data rate improvement (e.g., \cite{jain2011practical, bharadia2013full}). However, such a rate increase requires perfect SIC, which is extremely challenging to achieve. 
While a few recent papers considered non-negligible SI and the resulting rate gains~\cite{ahmed2013rate, li2014rate, cheng2013optimal,full-duplex-sigmetrics}, \emph{there is still no explicit characterization of the FD capacity region for a given profile of residual SI over frequency\footnote{In compact FD radio implementations (e.g.,~\cite{Zhou_WBSIC_ISSCC15}), the residual SI can vary wildly with the frequency.} and parameters of the wireless signal}. 
Most recent research has focused on maximizing the total throughput without considering Quality of Service (QoS) requirements. Namely, there has been very limited work on asymmetric traffic requirements on the uplink (UL) and downlink (DL)~\cite{full-duplex-sigmetrics,li2014rate,bi2015rate,korpi2015achievable}. 

While in Time Division Duplex (TDD) systems asymmetric traffic can be supported via time-sharing between the UL and DL, in FD the dependence of the bi-directional rates on the transmission power levels and Signal-to-Noise Ratio (SNR) levels is much more complex. As shown in Fig.~\ref{fig:motivation}, any (combination) of the following policies can be used: (i) FD with reduced transmission power at one of the stations, (ii) FD with fewer channels allocated to one of the stations, and (iii) time sharing between a few types of FD transmissions. 

We study asymmetric link traffic and analytically characterize the capacity region (i.e., all possible combinations of UL and DL rates) under non-negligible SI. Such characterization has theoretical importance, since it provides insights into the achievable gains from FD, thereby allowing to quantify the benefits in relation to the costs (in hardware and algorithmic complexity, power consumption, etc.). It also has practical importance, since it supports the development of algorithms for rate allocation under different UL and DL requirements. Such algorithms will determine the required combinations of the policies illustrated in Fig.~\ref{fig:motivation}. 

We first consider the case where both stations transmit on a \emph{single channel} and the remaining SI is a constant fraction of the transmitted power \cite{full-duplex-sigmetrics,bi2015rate}. We study the structural properties of the FD capacity region  and derive necessary and sufficient conditions for its convexity. Based on the properties, we present a simple and fast algorithm to ``convexify'' the region.\footnote{A convex region is desirable, since most resource allocation and scheduling algorithms rely on convexity and providing performance guarantees for a non-convex region is hard.} The convexified region combines (via time sharing) different FD rate pairs (see Fig.~\ref{fig:motivation}\subref{fig:scheduling}) and we refer to it as the Time Division Full-Duplex (TDFD) region. 
The algorithm finds the points at the region's boundary, given a constraint on one of the (UL or DL) rates.

We then consider the the \emph{multi-channel} case in which channels are orthogonal, as in Orthogonal Frequency Division Multiplexing (OFDM). We assume that the \emph{shape of the power allocation is fixed} but the total transmission power can be varied. Namely, the ratios between power levels at different channels are given. 
For each channel, the remaining SI is some fraction of the transmitted power \cite{full-duplex-sigmetrics,cheng2013optimal,zheng2015joint}. We characterize the FD capacity region and analytically show that any point on the region can be computed with a low-complexity binary search. 
We also focus on determining the TDFD capacity region, which due to the lack of structure cannot in general be obtained via binary search. However, we argue that for any practical input, the TDFD capacity region can be determined in real time.

Finally, we consider the TDFD capacity region in the \emph{multi-channel} case under a \emph{general power allocation}, (i.e., the power level at each channel is a decision variable). In this case, maximizing one of the rates when the other rate is given is a non-convex problem which is hard to solve. However, we develop an algorithm that under certain mild restrictions converges to a stationary point that in practice is a global maximum. Although for most practical cases, the algorithm is near-optimal and runs in polynomial time, its running time is not suitable for a real-time implementation. Hence, we develop a simple heuristic and show numerically that in most cases it has similar performance. 

For all the cases mentioned above, we present extensive numerical results that illustrate the capacity regions and the rate gains (compared to TDD) as a function of the receivers' SNR levels and SIC levels. We also highlight the intuition behind the performance of the different algorithms.

To summarize, the main contributions of the paper are two-fold: (i) it provides a fundamental characterization and structural understanding of the FD capacity regions, and (ii) the rate maximization algorithms, designed for asymmetrical traffic requirements, can serve as resource allocation building blocks for future FD MAC protocols.

The rest of the paper is organized as follows. Sections \ref{sec:related-work} and \ref{sec:model} review related work and outline the model. Section \ref{sec:single} studies the single channel case. Sections \ref{sec:multi} and \ref{sec:multi-power} study the multi-channel cases with fixed and general power allocations. We conclude in Section \ref{sec:conclusion}. 

\section{Related Work}\label{sec:related-work}

Various challenges related to FD wireless recently attracted significant attention. These include FD radio/system design \cite{choi2010achieving, jain2011practical,  khojastepour2011case, aryafar2012midu,Zhou_WBSIC_ISSCC15,bharadia2013full} as well as rate gain evaluation and resource allocation \cite{bai2013distributed, ahmed2013rate, li2014rate, xie2014does, cheng2013optimal,yang2015scheduling,zheng2015joint,bi2015rate, full-duplex-sigmetrics}. A large body of (analytical) work \cite{bai2013distributed, xie2014does, yang2015scheduling} focuses on \emph{perfect SIC} while we  focus on the more realistic model of imperfect SIC.

Rate gains and power allocation under \emph{imperfect SIC} were studied in \cite{ahmed2013rate, cheng2013optimal, li2014rate, full-duplex-sigmetrics,bi2015rate,zheng2015joint}. For the single channel case, \cite{ahmed2013rate} derives a sufficient condition for FD to outperform TDD in terms of sum UL and DL rates. 
However, \cite{ahmed2013rate} does not quantify the rate gains nor consider the multi-channel case. 

Power allocation for maximizing the sum of the UL and DL rates for the single- and multi-channel cases was studied in \cite{full-duplex-sigmetrics,cheng2013optimal}. The  maximization only determines a single point on the capacity region and does not imply anything about the rest of the region, which is our focus. 
While \cite{full-duplex-sigmetrics} (implicitly) constructs the FD capacity region in the single channel case (restated here as Proposition \ref{prop:FD-cap-region}), it does not derive any structural properties of the region, nor does it consider the multi-channel case or a combination of FD and TDD.

The capacity region for an FD MIMO two-way relay channel was studied in \cite{zheng2015joint} as a joint problem of beamforming and power allocation. For a fixed beamforming, the problem reduces to determining a single channel FD capacity region. Yet, the joint problem is significantly different from the problems considered here. 
The FD capacity region for multiple channels was considered in \cite{li2014rate}. 
While \cite{li2014rate} considers both fixed and general power allocation for determining an FD capacity region, the analytical results are obtained only for the fixed power case and the non-convex problem of general power allocation was addressed heuristically. Specifically, for the fixed power case, our proof of Lemma \ref{lemma:strictly-fd-ofdm-equal}  is more accurate than the proof of Theorem 3 in \cite{li2014rate} (see \iffullpaper the proof of Lemma \ref{lemma:strictly-fd-ofdm-equal}\fi).

The TDFD capacity region was studied in \cite{korpi2015achievable} only via simulation and in \cite{bi2015rate} analytically but mainly for the single-channel case. 
The ``convexification'' of the FD region in \cite{bi2015rate} is performed over a discrete set of rate pairs, which requires linear computation in the set size, assuming that the points are sorted (e.g., Ch.\ 33 in \cite{cormen2009introduction}). 
Our results for a single channel rely on the structural properties of the FD capacity region and do not require the set of FD rate pairs to be discrete. Moreover, the computation for determining the convexified region is logarithmic (see Section \ref{sec:single-algo}).

To the best of our knowledge, this is the first thorough 
study of the capacity region and rate gains of FD and TDFD.

\begin{figure*}[t]
\centering
\subfloat[]{\label{fig:gamma_bb}\includegraphics[scale = 0.22 ]{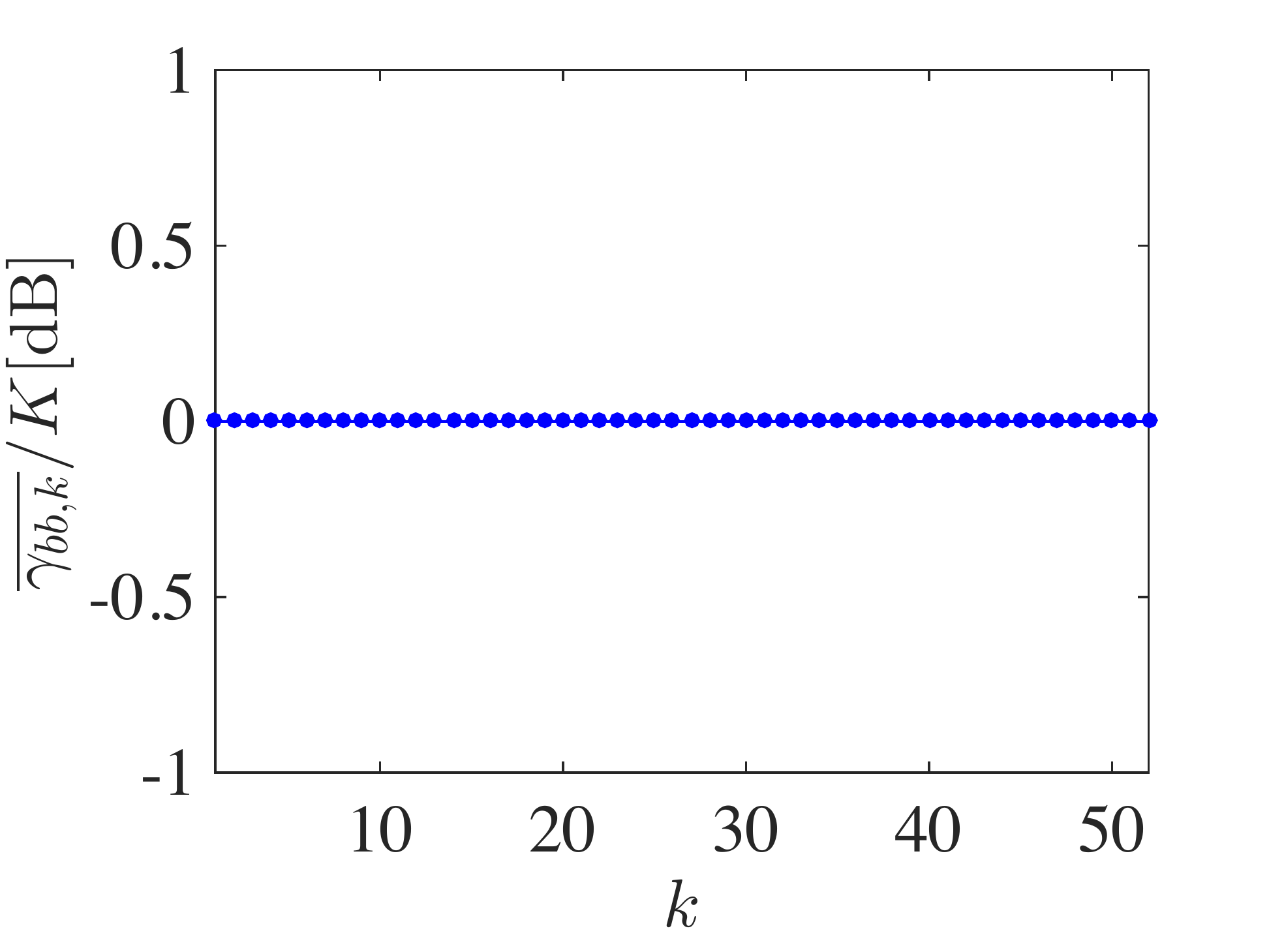}}\hspace{\fill}
\subfloat[]{\label{fig:gamma_mm_conv}\includegraphics[scale = 0.22]{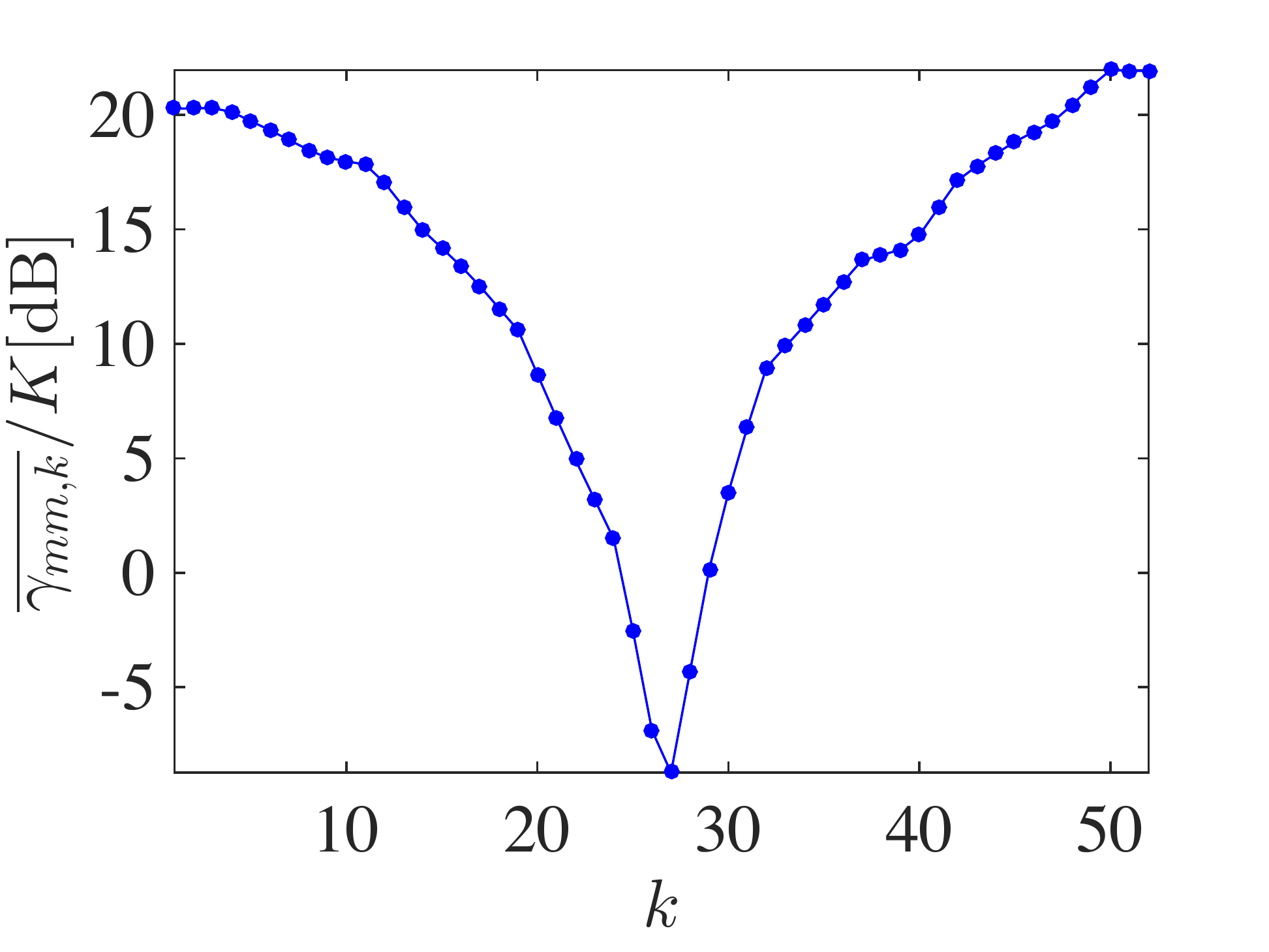}}\hspace{\fill}
\subfloat[]{\label{fig:gamma_mm_FDE_1}\includegraphics[scale = 0.22]{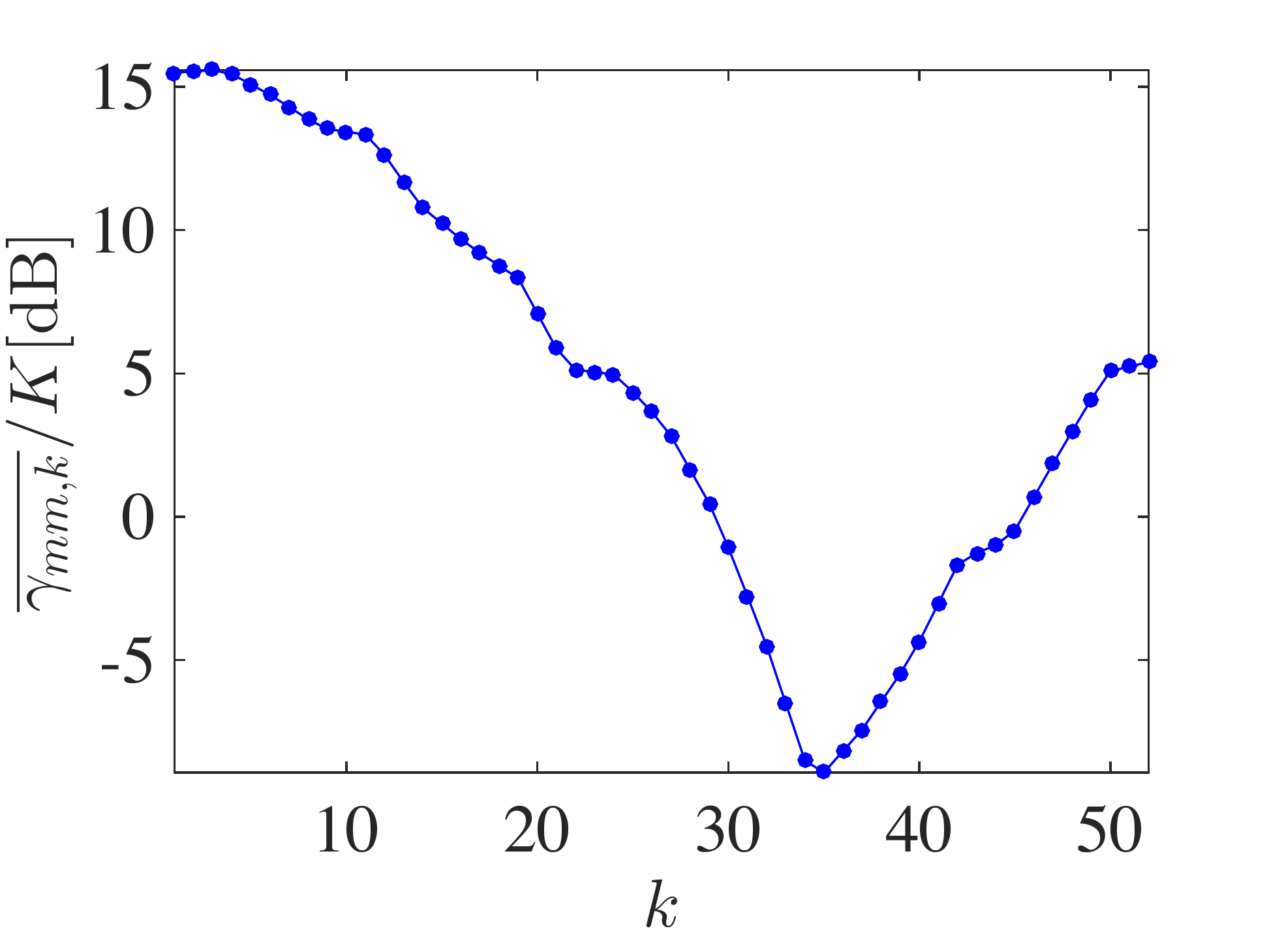}}\hspace{\fill}
\subfloat[]{\label{fig:gamma_mm_FDE_2}\includegraphics[scale = 0.22]{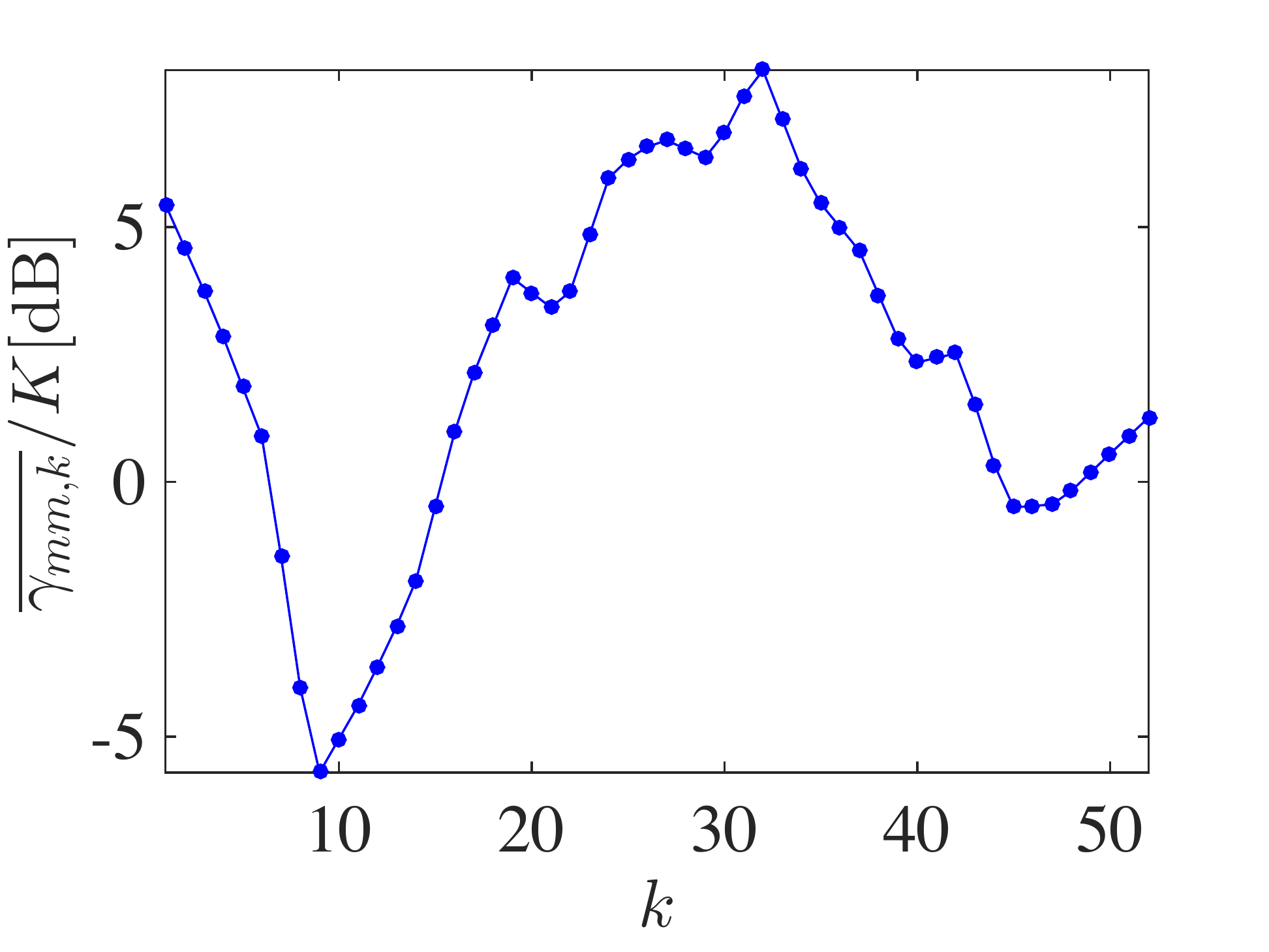}}
\vspace{-5pt}\caption{Considered cancellation profiles for the FD receiver  \protect \subref{fig:gamma_bb} at the BS \cite{bharadia2013full} and \protect\subref{fig:gamma_mm_conv}, \protect\subref{fig:gamma_mm_FDE_1}, \protect\subref{fig:gamma_mm_FDE_2} at the MS \cite{Zhou_WBSIC_ISSCC15}.}
\label{fig:cancellation-profiles}
\vspace{-10pt}
\end{figure*}

\section{Model and Notation}\label{sec:model}
We focus on the problem of determining the capacity region of an FD bidirectional link between two stations. For brevity, we refer to them as a Mobile Station (MS) and a Base Station (BS) and to the corresponding links as uplink (UL) and downlink (DL). 
For the number of channels $K$, we consider: (i) the single-channel case ($K=1$), and (ii) the multi-channel case ($K>1$), where we assume that the channels are orthogonal to each other. In the numerical evaluations, when $K>1$ we adopt $K=52$. We use $k$ to denote the channel index. When $K=1$, we omit the indices.  

$P_{u, k}$ denotes the transmission power level at station $u\in\{m, b\}$ on channel $k$ and $\overline{P_u}$ denotes the maximum sum of transmission power levels at station $u$: $\sum_{k=1}^K P_{u, k}\leq \overline{P_u}$, where $u\in \{m, b\}$. For simplicity, we introduce notation for the normalized transmission power levels: $\alpha_{b, k} = P_{b, k}/\overline{P_b}$, $\alpha_{m, k} = P_{m, k}/\overline{P_m}$. The constraints for the sum of transmission power levels are then: $\sum_k \alpha_{b, k}\leq 1$ and $\sum_{k}\alpha_{m, k}\leq 1$.

$\overline{\gamma_{bm, k}}$ and $\overline{\gamma_{mb, k}}$ denote the SNR of the signal from the BS to the MS and from the MS to the BS, respectively, on channel $k$, when the transmission power level on channel $k$ is set to its maximum value ($\overline{P_b}, \overline{P_m}$, respectively). $\overline{\gamma_{bm}}\equiv \frac{1}{K}\sum_{k}\overline{\gamma_{bm, k}}/K$ and $\overline{\gamma_{mb}}\equiv\frac{1}{K} \sum_k \overline{\gamma_{mb, k}}/K$ denote the average SNR when the power levels are equally allocated over channels (i.e., when $\alpha_{b, 1} =...=\alpha_{b, K}=1/K$ and $\alpha_{m, 1} =...=\alpha_{m, K}=1/K$). In the numerical evaluations, we adopt $\overline{\gamma_{bm, k}} = K \overline{\gamma_{bm}}$ and $\overline{\gamma_{mb, k}} = K \overline{\gamma_{mb}}$, $\forall k$, to focus on the effects caused by FD operation. Our results, however, hold for general values of $\overline{\gamma_{bm, k}}$ and $\overline{\gamma_{mb, k}}$ over channels $k$.

Similarly to \cite{full-duplex-sigmetrics,li2014rate,cheng2013optimal}, we model the remaining SI on channel $k$ as a constant fraction of the transmission power level on channel $k$. The Self-Interference-to-Noise-Ratio (XINR) at the BS on channel $k$ when $\alpha_{b, k} = 1$ is denoted by $\overline{\gamma_{bb, k}}$. The XINR at the MS on channel $k$ when $\alpha_{m, k} = 1$ is denoted by $\overline{\gamma_{mm, k}}$. 
In the numerical evaluations of the multi-channel case, we use $\overline{\gamma_{bb, k}}/K = 1 = 0$dB, as shown in Fig.~\ref{fig:cancellation-profiles}\subref{fig:gamma_bb}, which is motivated by \cite{bharadia2013full}. For $\overline{\gamma_{mm, k}}$, we consider three FD RFIC designs from \cite{Zhou_WBSIC_ISSCC15}, shown in Fig.~\ref{fig:cancellation-profiles}\subref{fig:gamma_mm_conv}--\subref{fig:gamma_mm_FDE_2}. For the FD RFICs from \cite{Zhou_WBSIC_ISSCC15}, we assume additional 50dB of cancellation in the digital domain and 110dB difference between the maximum transmission signal and the noise. 

For the DL rate on channel $k$, $r_{b, k}$, and for the UL rate on channel $k$, $r_{m, k}$, we use the Shannon capacity formula: 
$
r_{b, k} = \log\big(1+\frac{\alpha_{b, k}\overline{\gamma_{bm, k}}}{1+\alpha_{m, k}\overline{\gamma_{mm, k}}}\big)$, 
$r_{m, k} = \log\big(1+\frac{\alpha_{m, k}\overline{\gamma_{mb, k}}}{1+\alpha_{b, k}\overline{\gamma_{bb, k}}}\big), 
$
where $\log$ denotes the base-2 logarithm. $r_b = \sum_k r_{b, k}$ denotes the sum of DL rates over channels $k$, $r_m = \sum_k r_{m, k}$ denotes the sum of UL rates over channels $k$, and $r = r_m + r_b$ denotes the sum of all UL and DL rates over channels $k$ (in the following, we refer to $r$ as the sum rate). 

We denote by $\overline{r_b} = \max \{r_b(\{\alpha_{b, k}\}, \{\alpha_{m, k}\}): \sum_k \alpha_{b, k}\leq 1$, $\sum_{k}\alpha_{m, k}\leq 1\}$ the maximum DL rate. Observe that when $r_b$ is maximized, we have $\sum_{k}\alpha_{b, k} = 1, \alpha_{m, k} = 0, \forall k$, i.e., $\overline{r_b}$ is equal to the maximum HD rate on the DL. Similarly, $\overline{r_m}$ denotes the maximum UL rate. 

\begin{figure}[t]
\centering
\subfloat[]{\label{fig:FD_cap_region}\includegraphics[scale = 0.21]{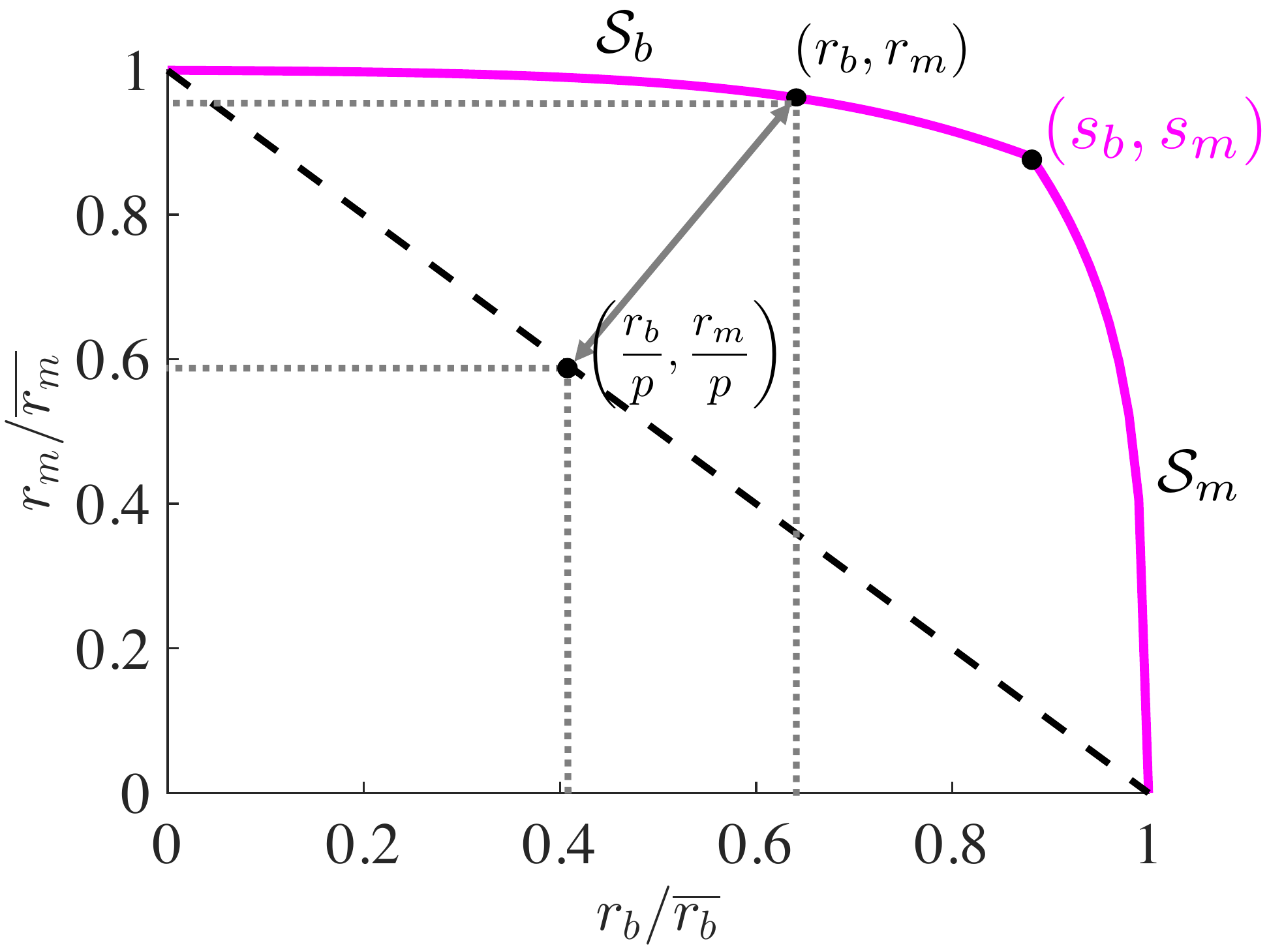}}\hspace{\fill}
\subfloat[]{\label{fig:FD_cap_region_ncvx}\includegraphics[scale = 0.21]{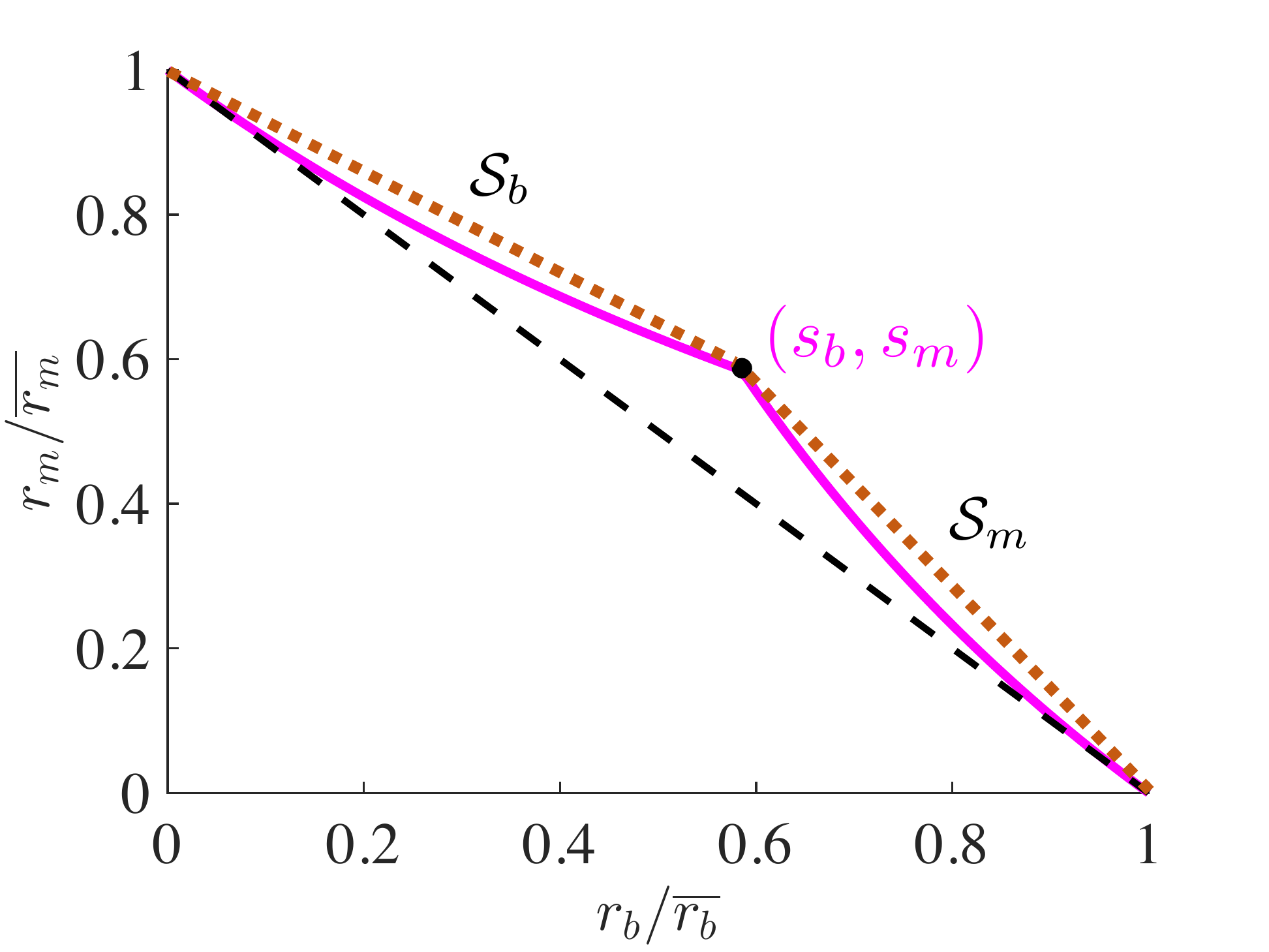}}\vspace{-5pt}
\caption{\protect\subref{fig:FD_cap_region} Convex and \protect\subref{fig:FD_cap_region_ncvx} non-convex FD capacity regions. A dashed line delimits the corresponding TDD region. An FD region is convex, if and only if segments $\mathcal{S}_b$ (connecting $(0, \overline{r_m})$ and $(s_b, s_m)$) and $\mathcal{S}_m$ (connecting $(s_b, s_m)$ and $(\overline{r_b}, 0)$) can be represented by a concave function $r_m(r_b)$.}\label{fig:cap-regions}\vspace{-10pt}
\end{figure}

A capacity region of an FD link is the set of all achievable UL-DL FD rate pairs. Examples of FD regions appear in Fig.~\ref{fig:cap-regions}, where a full line represents the FD  region boundary, and a dashed line represents the TDD region boundary. The problem of determining the FD capacity region is the problem of maximizing one of the rates (e.g., $r_m$) when the other rate ($r_b$) is fixed, subject to the sum power constraints.  

An FD capacity region is not necessarily convex. In such cases, we also consider a \emph{convexified} or \emph{TDFD} capacity region, namely, the convex hull of the FD capacity region. In practice, the TDFD region would correspond to time sharing between different FD rate pairs. Fig.~\ref{fig:cap-regions}\subref{fig:FD_cap_region_ncvx} illustrates a non-convex FD capacity region, with the dotted line representing the boundary of the TDFD capacity region.

To compare an FD or a TDFD capacity region to its corresponding TDD region, we use the following definition (a similar definition appears in \cite{full-duplex-sigmetrics}, see Fig.~\ref{fig:cap-regions}\subref{fig:FD_cap_region} for intuition): 
\begin{definition}\label{def:rate-improvement}
For a given rate pair $(r_b, r_m)$ from an FD or TDFD capacity region, the rate improvement $p$ is defined as the largest (positive) number such that $\Big(\frac{r_b}{p}, \frac{r_m}{p}\Big)$ is at the boundary of the corresponding TDD capacity region. 
\end{definition}
Using simple geometry, $p$ can be computed as follows \cite{full-duplex-sigmetrics}:
\begin{proposition}
$p(r_b, r_m) = {r_b}/{\overline{r_b}}+{r_m}/{\overline{r_m}}$.
\end{proposition}

\section{Single Channel}\label{sec:single}
We now study the structural properties of the FD and TDFD capacity regions for a single FD channel and devise an algorithm that determines the points at the boundary of the TDFD capacity region. First, we provide structural results that characterize FD capacity regions. We prove that the FD region boundary, which can be described by a function $r_m(r_b)$, can only have up to four either convex or concave pieces that can only appear in certain specific arrangements. We also provide necessary and sufficient conditions for the region's boundary to take one of the possible shapes. As a corollary, we derive necessary and sufficient conditions for the FD region to be convex as a function of $\overline{\gamma_{bm}}, \overline{\gamma_{mb}}, \overline{\gamma_{mm}}$, and $\overline{\gamma_{bb}}$.  
Based on the structural results, we present a simple and fast algorithm that can determine any point at the boundary of the TDFD capacity region. For a given rate $r_b^*$, to find the maximum rate $r_m$ subject to $r_b = r_b^*$, the algorithm determines the shape of the capacity region as a function of $\overline{\gamma_{bm}}, \overline{\gamma_{mb}}, \overline{\gamma_{mm}}$, and $\overline{\gamma_{bb}}$, and either directly computes $r_m$ or performs a binary search to find it.

\subsection{Capacity Region Structural Results}\label{section:SC-struct-results}
We start by characterizing the power allocation at the boundary of an FD capacity region, given by the following simple proposition (used implicitly in \cite{full-duplex-sigmetrics}). In the rest of the section, $s_b = r_b(1, 1)$, $s_m = r_m(1, 1)$.
\begin{proposition}\label{prop:FD-cap-region}
If $r_b = r_b^*\leq s_b$, then $r_m$ is maximized for $\alpha_m = 1$ and $\alpha_b$ that solves $r_b(\alpha_b, 1) = r_b^*$. \iffullpaper Similarly, if $r_m = r_m^*\leq s_m$, then $r_b$ is maximized for $\alpha_b = 1$ and $\alpha_m$ that solves $r_m(1, \alpha_m) = r_m^*$.\else\fi
\end{proposition}
\iffullpaper\begin{proof}
We prove the proposition for $r_b = r_b^*$ ($\leq s_b$), while the other part follows by symmetric arguments.

Let $\alpha_b$ be such that $r_b(\alpha_b, 1) = r_b^*$. Fix any $(\alpha_b', \alpha_m')\in [0, 1]^2$ such that $r_b(\alpha_b' , \alpha_m') = r_b^*$, and suppose that:
\begin{align}
r_m(\alpha_b', \alpha_m') \geq r_m(\alpha_b, 1). \label{eq:FD-cap-region-fpoc}
\end{align}
Then, after simple transformations of (\ref{eq:FD-cap-region-fpoc}), we have that:
\begin{align*}
\frac{\alpha_m'\overline{\gamma_{mb}}}{1 + \alpha_b' \overline{\gamma_{bb}}}\geq \frac{\overline{\gamma_{mb}}}{1 + \alpha_b\overline{\gamma_{bb}}}.
\end{align*}
Finally, using that (from $r_b(\alpha_b, 1) = r_b^*$, $r_b(\alpha_b' , \alpha_m') = r_b^*$): $\alpha_b = (2^{r_b^*}-1)\cdot(1+\overline{\gamma_{mm}})/\overline{\gamma_{bm}}$, $\alpha_b' = (2^{r_b^*}-1)\cdot(1+\alpha_m'\overline{\gamma_{mm}})/\overline{\gamma_{bm}}$, it follows that:
\begin{align}
&\alpha_m'\Big(1+ (2^{r_b^*}-1)\cdot\frac{\overline{\gamma_{bb}}}{\overline{\gamma_{bm}}}\Big) \geq 1 + (2^{r_b^*}-1)\cdot\frac{\overline{\gamma_{bb}}}{\overline{\gamma_{bm}}},\notag
\end{align}
and, therefore, $\alpha_m' \geq 1$. As $\alpha_m'\leq 1$, it follows that $\alpha_m' = 1$ and $\alpha_b' = \alpha_b$, thus completing the proof.
\end{proof}
\fi

Proposition \ref{prop:FD-cap-region} implies that to determine any point $(r_b, r_m)$ at the boundary of the capacity region, where $r_b, r_m > 0$, for $r_b \leq s_b$ (resp.\ $r_m \leq s_m$), it suffices to find $\alpha_b$ (resp.\ $\alpha_m$) that satisfies $r_b = r_b(\alpha_b, 1)$ (resp.\ $r_m = r_m(1, \alpha_m)$). The capacity region is convex, if and only if (i) $r_b(r_m)$ is concave for $r_m \in (0, s_m]$ and $r_b$ at the boundary of the capacity region, (ii) $r_m(r_b)$ is concave for $r_b \in (0, s_b]$ and $r_m$ at the boundary of the capacity region, and (iii) the functions $r_m(r_b)$ and $r_b(r_m)$ intersect at $(s_b, s_m)$ under an angle smaller than $\pi$.

If the FD capacity region is convex (Fig.~\ref{fig:cap-regions}\subref{fig:FD_cap_region}), then to maximize $r_m$ subject to $r_b = r_b^*$, it is always optimal to use FD and allocate the power levels according to Proposition \ref{prop:FD-cap-region}. This is not necessarily true, if the capacity region is not convex; in that case, it may be optimal to use a time-sharing scheme between two FD rate pairs (TDFD), since a convex combination of e.g., $(s_b, s_m)$ and $(\overline{r_b}, 0)$ may lie above the FD capacity region boundary (e.g., Fig.~\ref{fig:cap-regions}\subref{fig:FD_cap_region_ncvx}).

The following lemma characterizes the FD capacity region boundary.
\begin{lemma}\label{lemma:convexity-of-cap-region}
Given positive $\overline{\gamma_{mb}}, \overline{\gamma_{bm}}, \overline{\gamma_{bb}}, \overline{\gamma_{mm}}$, let $r_m(r_b)$ describe the boundary of the FD capacity region for $r_b\in [0, s_b]$, and $r_b(r_m)$ describe the boundary of the FD capacity region for $r_m\in [0, s_m]$. Then $r_m(r_b)$ ($r_b\in [0, s_b]$)
 and $r_b(r_m)$ ($r_m\in [0, s_m]$) can only be described by one of the following three function types: (i) concave, (ii) convex, and (iii) concave for $r_b \in [0, r_b^+]$ for some $r_b^+<s_b$ in the case of $r_m(r_b)$,  concave for $r_m \in [0, r_m^+]$ for some $r_m^+<s_m$ in the case of $r_b(r_m)$, and convex on the rest of the domain. 
\end{lemma}

\begin{figure}[t!]
\center
\subfloat[$\overline{\gamma_{mm}} = 0$dB]{\label{fig:cap_ext_00}\includegraphics[scale = 0.19]{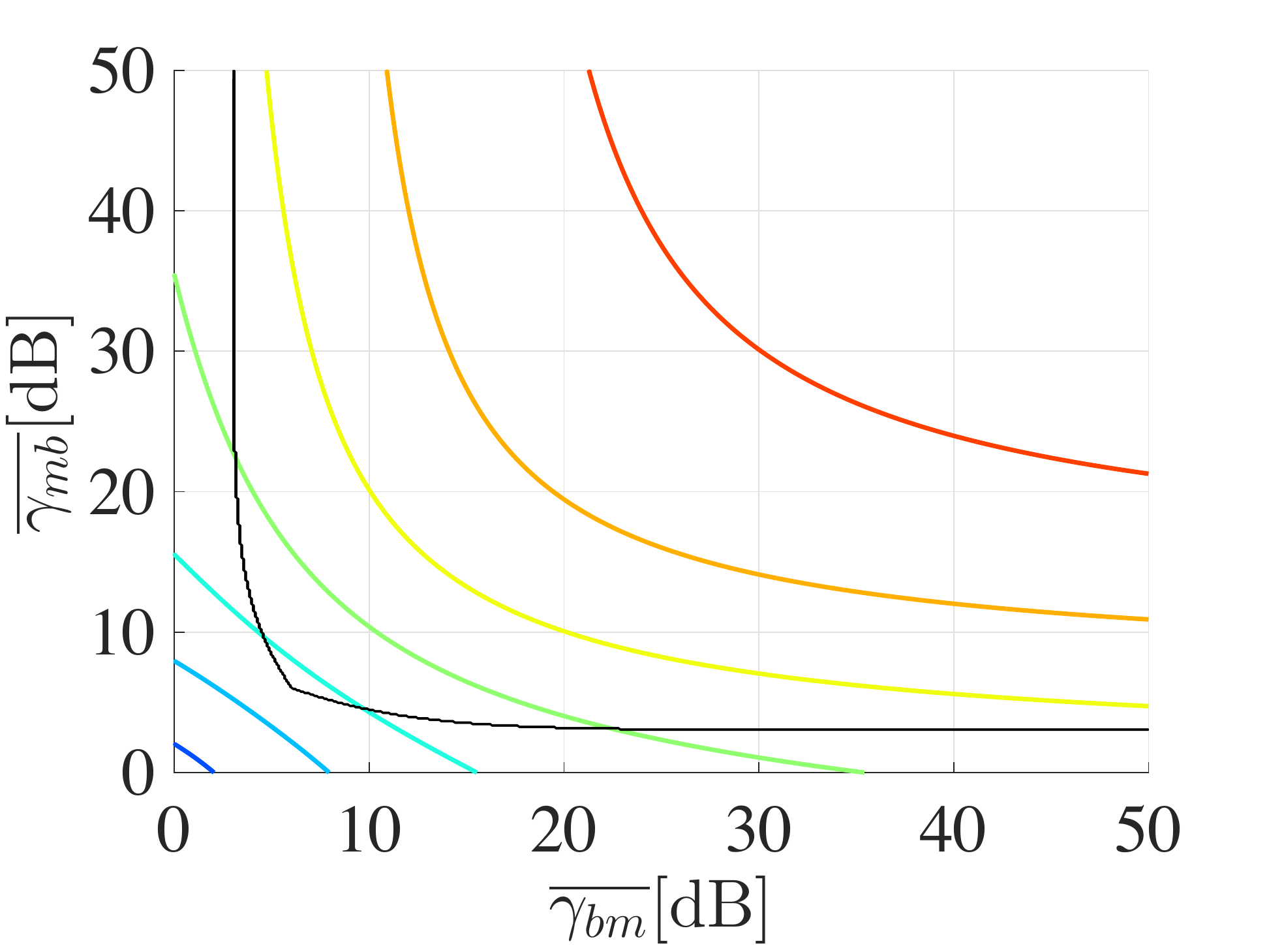}}\hspace{\fill}
\subfloat[$\overline{\gamma_{mm}} = 10$dB]{\label{fig:cap_ext_010}\includegraphics[scale = 0.19]{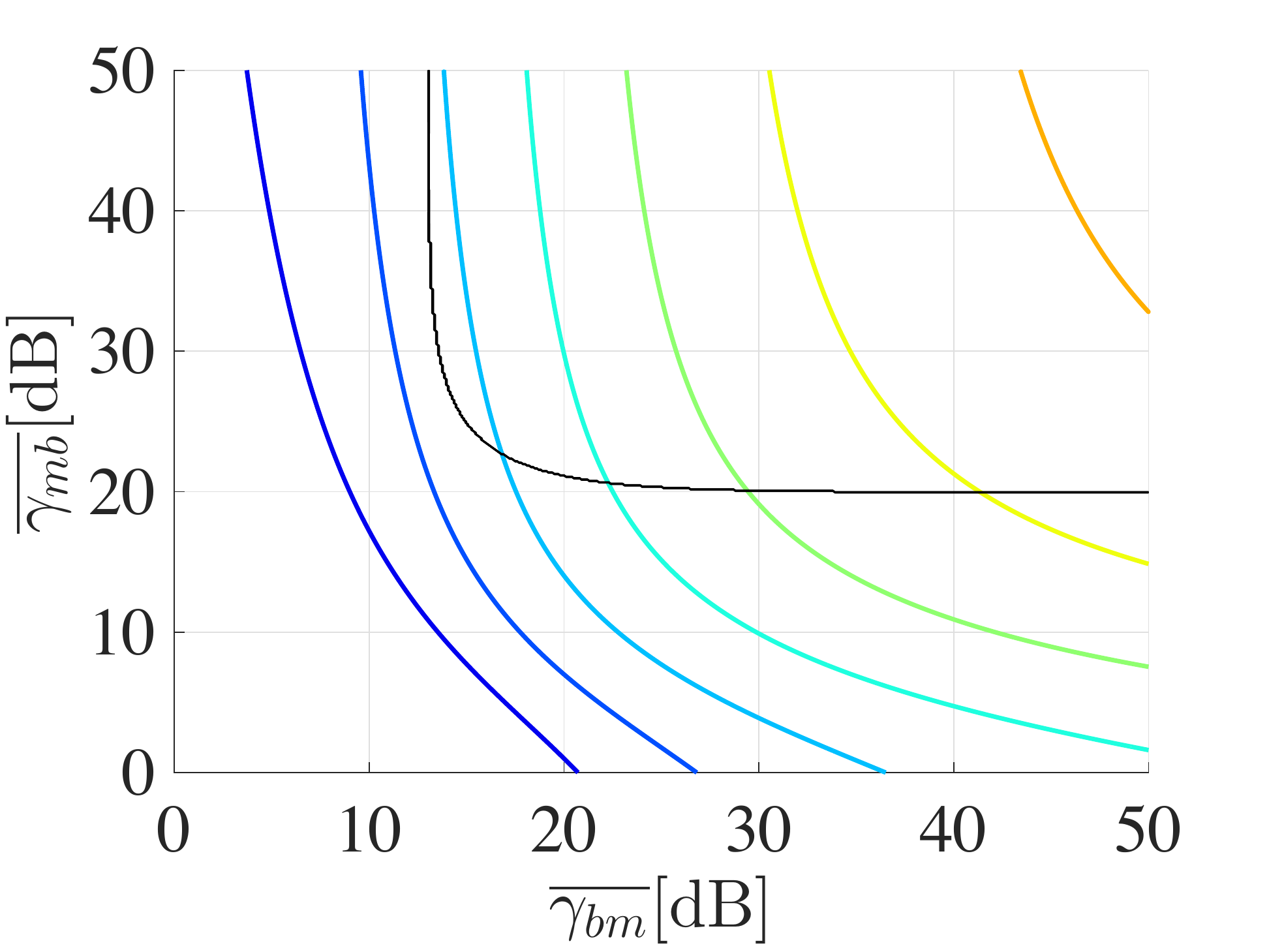}}\hspace{\fill}
\subfloat{\label{fig:colorbar}\includegraphics[scale = 0.19 ]{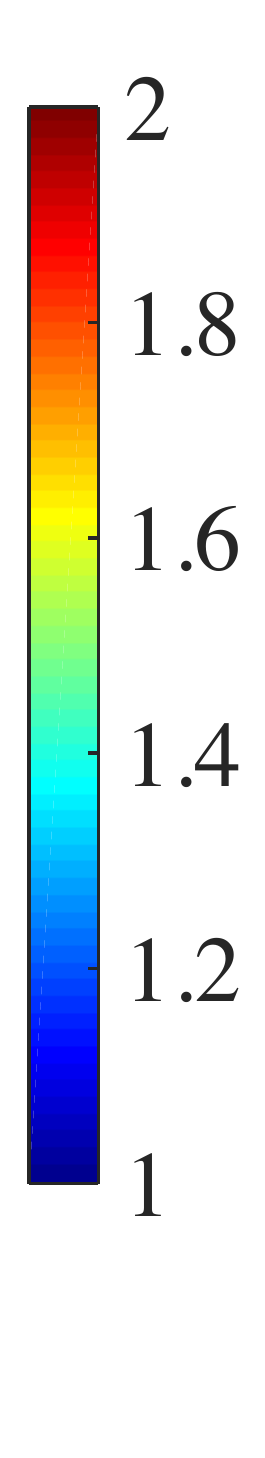}}
\caption{Convexity of the capacity region vs. rate improvement for $\overline{\gamma_{bb}} = 0$dB and: \protect\subref{fig:cap_ext_00} $\overline{\gamma_{mm}} = 0$dB and  
\protect\subref{fig:cap_ext_010} $\overline{\gamma_{mm}} = 10$dB. The capacity region is convex for UL and DL SNRs north and east from the black curve.}
\label{fig:cap-region-convexity}\vspace{-10pt}
\end{figure}

\iffullpaper
\begin{proof}
From Prop. \ref{prop:FD-cap-region}, segment $\mathcal{S}_b$ is described by $r_m(r_b)$, where $r_b\leq s_b$, $\alpha_b \in [0, 1]$, and:
\begin{equation}\label{eq:rb-rm-alpha-b}
r_b = \log\Big(1 + \frac{\alpha_b\overline{\gamma_{bm}}}{1 + \overline{\gamma_{mm}}}\Big)\text{ and } 
r_m = \log\Big(1 + \frac{\overline{\gamma_{mb}}}{1 + \alpha_b \overline{\gamma_{bb}}}\Big).
\end{equation}
Similarly, segment $\mathcal{S}_m$ is described by $r_b(r_m)$, where $r_m\leq s_m$, $\alpha_m\in[0, 1]$, and:
\begin{equation}\label{eq:rb-rm-alpha-m}
r_b = \log\Big(1 + \frac{\overline{\gamma_{bm}}}{1 + \alpha_m\overline{\gamma_{mm}}}\Big), \, 
r_m = \log\Big(1 + \frac{\alpha_m \overline{\gamma_{mb}}}{1 + \overline{\gamma_{bb}}}\Big).
\end{equation}
We prove the lemma only for segment $\mathcal{S}_b$, while the proof for segment $\mathcal{S}_m$ follows by symmetry.

Since, from (\ref{eq:rb-rm-alpha-b}), $r_m(r_b)$ is a continuous and twice differentiable function for $r_b\in[0, s_b]$ (equivalently, $\alpha_b\in[0, 1]$), $r_m(r_b)$ is concave for $r_b\in[0, s_b]$ if and only if $\frac{d^2 r_m}{d {r_b}^2}\leq 0$. Observe that we can write:
\begin{equation}\label{eq:drm-drb}
\frac{d r_m}{d r_b} = \frac{d r_m}{d \alpha_b}\cdot \frac{d \alpha_b}{d r_b}
\end{equation}
and
\begin{equation}\label{eq:d2rm-drb2}
\frac{d^2 r_m}{d {r_b}^2} = \frac{d^2 r_m}{d {\alpha_b}^2}\cdot\left(\frac{d \alpha_b}{d r_b}\right)^2 + \frac{d r_m}{d \alpha_b}\cdot \frac{d^2 \alpha_b}{d {r_b}^2}.
\end{equation}
From the left equality in (\ref{eq:rb-rm-alpha-b}):
\begin{gather}
\alpha_b = (2^{r_b}-1)\cdot\frac{1+\overline{\gamma_{mm}}}{\overline{\gamma_{bm}}},\notag\\
\frac{d \alpha_b}{d r_b} = \ln(2)\cdot 2^{r_b} \cdot\frac{1+\overline{\gamma_{mm}}}{\overline{\gamma_{bm}}}, \text{ and}\label{eq:dalphab-drb}\\
\frac{d^2 \alpha_b}{d {r_b}^2} = \ln^2(2)\cdot 2^{r_b} \cdot\frac{1+\overline{\gamma_{mm}}}{\overline{\gamma_{bm}}}. \label{eq:d2alphab-drb2}
\end{gather}
From the right equality in (\ref{eq:rb-rm-alpha-b}):
\begin{gather}
\frac{d r_m}{d \alpha_b} = - \frac{\overline{\gamma_{bb}}}{\ln(2)}\cdot \Big(\frac{1}{1 + \alpha_b \overline{\gamma_{bb}}} -  \frac{1}{1 + \alpha_b \overline{\gamma_{bb}} + \overline{\gamma_{mb}}}\Big),\label{eq:drm-dalpham}\\
\frac{d^2 r_m}{d {\alpha_b}^2} = \frac{(\overline{\gamma_{bb}})^2}{\ln(2)}\cdot \Big(\frac{1}{1 + \alpha_b \overline{\gamma_{bb}}} -  \frac{1}{1 + \alpha_b \overline{\gamma_{bb}} + \overline{\gamma_{mb}}}\Big)\notag\\
\cdot\Big(\frac{1}{1 + \alpha_b \overline{\gamma_{bb}}} +  \frac{1}{1 + \alpha_b \overline{\gamma_{bb}} + \overline{\gamma_{mb}}}\Big).\label{eq:d2rm-dalpham2}
\end{gather}
Plugging (\ref{eq:dalphab-drb})--(\ref{eq:d2rm-dalpham2}) back into (\ref{eq:d2rm-drb2}), we have that the sign of $\frac{d^2 r_m}{d {r_b}^2}\leq 0$ is equivalent to the sign of:
\begin{equation}\label{eq:rm-concave-cond}
\overline{\gamma_{bb}}\Big(\frac{1}{1 + \alpha_b \overline{\gamma_{bb}}} +  \frac{1}{1 + \alpha_b \overline{\gamma_{bb}} + \overline{\gamma_{mb}}}\Big) \frac{2^{r_b}(1+\overline{\gamma_{mm}})}{\overline{\gamma_{bm}}} - 1.
\end{equation}
Recalling (from (\ref{eq:rb-rm-alpha-b})) that $2^{r_b} = 1 + \frac{\alpha_b \overline{\gamma_{bm}}}{1+\overline{\gamma_{mm}}}$ and using simple algebraic transformations, (\ref{eq:rm-concave-cond}) is equivalent to:
\begin{equation}\label{eq:quad-ineq-alphab}
{\alpha_b}^2 + \alpha_b\cdot\frac{2(1+\overline{\gamma_{mm}})}{\overline{\gamma_{bm}}} + \frac{(2+\overline{\gamma_{mb}})(1+\overline{\gamma_{mm}})}{\overline{\gamma_{bb}}\overline{\gamma_{bm}}} - \frac{1+\overline{\gamma_{mb}}}{(\overline{\gamma_{bb}})^2}.
\end{equation}
(\ref{eq:quad-ineq-alphab}) is a quadratic function whose smaller root is negative. If the discriminant of (\ref{eq:quad-ineq-alphab}) is negative or the larger root is at most 0, (\ref{eq:quad-ineq-alphab}) is non-positive for all $\alpha_b\in [0, 1]$, and therefore $r_m(r_b)$ is convex for all $r_b\in [0, s_b]$. If the discriminant of (\ref{eq:quad-ineq-alphab}) is positive and the larger root is at least 1, (\ref{eq:quad-ineq-alphab}) is non-negative for all $\alpha_b\in [0, 1]$, and therefore $r_m(r_b)$ is concave for all $r_b\in [0, s_b]$. Finally, if the discriminant of (\ref{eq:quad-ineq-alphab}) is positive and the larger root takes value $\alpha_b^+ < 1$, $r_m(r_b)$ is concave for $r_b\in[0, r_b^+]$ and convex for $r_b\in [r_b^+, s_b]$, where $r_b^+ = r_b(\alpha_b^+ \overline{\gamma_{bm}}, \overline{\gamma_{mb}})$.
\end{proof}
\fi
The following corollary of the proof of Lemma \ref{lemma:convexity-of-cap-region} gives necessary and sufficient conditions for $r_m(r_b)$ to be concave for $r_b\in[0, s_b]$\iffullpaper, and, similarly, for $r_b(r_m)$ to be concave for $r_m\in[0, s_m]$\fi .
\begin{corollary}
For given positive $\overline{\gamma_{mb}}, \overline{\gamma_{bm}}, \overline{\gamma_{bb}}$, and $\overline{\gamma_{mm}}$, $r_m(r_b)$ is concave for $r_b\in[0, s_b]$ if and only if:
\begin{align}
\overline{\gamma_{bm}} > \max\Big\{&(\overline{\gamma_{mm}})^2-1,\; \overline{\gamma_{bb}}(1+\overline{\gamma_{mm}})\frac{2+\overline{\gamma_{mb}}}{1+\overline{\gamma_{mb}}},\notag \\ 
 &(1+\overline{\gamma_{mm}})\frac{2 + (2+\overline{\gamma_{mb}})/\overline{\gamma_{bb}}}{(1+\overline{\gamma_{mb}})/(\overline{\gamma_{bb}})^2 - 1}\Big\}.\label{eq:concave-rm}
\end{align}
\iffullpaper Similarly, $r_b(r_m)$ is concave for $r_m\in[0, s_m]$ if and only if:
\begin{align}
\overline{\gamma_{mb}}>\max\Big\{&(\overline{\gamma_{bb}})^2-1,\; \overline{\gamma_{mm}}(1+\overline{\gamma_{bb}})\frac{2+\overline{\gamma_{bm}}}{1+\overline{\gamma_{bm}}}, \notag\\
& (1+\overline{\gamma_{bb}})\frac{2 + (2+\overline{\gamma_{bm}})/\overline{\gamma_{mm}}}{(1+\overline{\gamma_{bm}})/(\overline{\gamma_{mm}})^2 - 1}\Big\}.\label{eq:concave-rb}
\end{align}\else\fi
\end{corollary}
\iffullpaper
\begin{proof}
From the proof of Lemma \ref{lemma:convexity-of-cap-region}, for $r_m(r_b)$ to be concave in all $r_b\in [0, s_b]$, the quadratic function (\ref{eq:quad-ineq-alphab}) in $\alpha_b$ needs to be non-positive for all $\alpha_b\in[0, 1]$. It follows that the discriminant of (\ref{eq:quad-ineq-alphab}) must be positive and the larger of the roots, $\alpha_b^+$, must be greater than or equal to 1 (the smaller root is negative). Finding the larger root of (\ref{eq:quad-ineq-alphab}) gives:
\begin{align}\label{eq:alphab+}
\alpha_b^+ &= \frac{1+\overline{\gamma_{mm}}}{\overline{\gamma_{bm}}}\bigg(-1 \notag\\
&+ \sqrt{1 + \frac{(\overline{\gamma_{bm}})^2}{(\overline{\gamma_{bb}})^2}
\cdot\frac{1+\overline{\gamma_{mb}}}{(1+\overline{\gamma_{mm}})^2}-\frac{\overline{\gamma_{bm}}}{\overline{\gamma_{bb}}}\cdot\frac{2+\overline{\gamma_{mb}}}{1+\overline{\gamma_{mm}}}}\bigg) \geq 1.
\end{align}
From (\ref{eq:alphab+}), as $\alpha_b^+ > 0$, it must also be:
\begin{align}
&\frac{(\overline{\gamma_{bm}})^2}{(\overline{\gamma_{bb}})^2}\cdot\frac{1+\overline{\gamma_{mb}}}{(1+\overline{\gamma_{mm}})^2}-\frac{\overline{\gamma_{bm}}}{\overline{\gamma_{bb}}}\cdot\frac{2+\overline{\gamma_{mb}}}{1+\overline{\gamma_{mm}}} > 0\notag \\
\Rightarrow \quad & \overline{\gamma_{bm}} > \overline{\gamma_{bb}}(1+\overline{\gamma_{mm}})\cdot \frac{2+\overline{\gamma_{mb}}}{1+\overline{\gamma_{mb}}}. \label{eq:gammabm-cond-1}
\end{align}
Note that (\ref{eq:gammabm-cond-1}) implies that the discriminant of (\ref{eq:quad-ineq-alphab}) is greater than 1 and therefore positive. 

Further, solving (\ref{eq:alphab+}) for $\overline{\gamma_{bm}}$, we get:
\begin{align}
&\frac{1+\overline{\gamma_{mm}}}{\overline{\gamma_{bm}}}\bigg(-1 \notag\\
&+ \sqrt{1 + \frac{(\overline{\gamma_{bm}})^2}{(\overline{\gamma_{bb}})^2}\cdot\frac{1+\overline{\gamma_{mb}}}{(1+\overline{\gamma_{mm}})^2}-\frac{\overline{\gamma_{bm}}}{\overline{\gamma_{bb}}}\cdot\frac{2+\overline{\gamma_{mb}}}{1+\overline{\gamma_{mm}}}}\bigg) \geq 1\notag\\
\Leftrightarrow \; &1 + \frac{(\overline{\gamma_{bm}})^2}{(\overline{\gamma_{bb}})^2}\frac{1+\overline{\gamma_{mb}}}{(1+\overline{\gamma_{mm}})^2}-\frac{\overline{\gamma_{bm}}}{\overline{\gamma_{bb}}}\frac{2+\overline{\gamma_{mb}}}{1+\overline{\gamma_{mm}}} \geq \Big(\frac{\overline{\gamma_{bm}}}{1+\overline{\gamma_{mm}}} + 1\Big)^2 
\notag\\
\Leftrightarrow\; &\frac{\overline{\gamma_{bm}}}{1+\overline{\gamma_{mm}}}\left(\frac{1+\overline{\gamma_{mb}}}{(\overline{\gamma_{bb}})^2}-1\right) - \frac{2+\overline{\gamma_{mb}}}{\overline{\gamma_{bb}}} - 2 \geq 0. 
\label{eq:gamma-bm-intermidiate-ineq}
\end{align}
Now, for (\ref{eq:gamma-bm-intermidiate-ineq}) to be possible to satisfy, as $\overline{\gamma_{bb}}, \overline{\gamma_{mm}}, \overline{\gamma_{bm}}, \overline{\gamma_{mb}}$ are all strictly positive, it must be:
\begin{align}
&\frac{1+\overline{\gamma_{mb}}}{(\overline{\gamma_{bb}})^2}-1 > 0\notag\\
\Rightarrow \quad & \overline{\gamma_{mb}} > (\overline{\gamma_{bb}})^2 - 1. \label{eq:gammabm-cond-2}
\end{align}
Finally, solving (\ref{eq:gamma-bm-intermidiate-ineq}) (given that (\ref{eq:gammabm-cond-2}) holds), we get:
\begin{equation}
\overline{\gamma_{bm}} \geq (1+\overline{\gamma_{mm}})\frac{2 + \frac{2+\overline{\gamma_{mb}}}{\overline{\gamma_{bb}}}}{\frac{1+\overline{\gamma_{mb}}}{(\overline{\gamma_{bb}})^2}-1}. \label{eq:gammabm-cond-3}
\end{equation}
Inequalities (\ref{eq:gammabm-cond-1}), (\ref{eq:gammabm-cond-2}), and (\ref{eq:gammabm-cond-3}) and their counterparts obtained when $r_b(r_m)$ is concave give (\ref{eq:concave-rm})--(\ref{eq:concave-rb}) from the statement of the lemma.
\end{proof}
\fi
Finally, we show that whenever both $r_m(r_b)$ is concave for all $r_b\in[0, s_b]$ and $r_b(r_m)$ is concave for all $r_m\in[0, s_m]$, the FD region is convex.
\begin{proposition}
If both $r_m(r_b)$ is concave for all $r_b\in[0, s_b]$ and $r_b(r_m)$ is concave for all $r_m\in[0, s_m]$, then the FD capacity region is convex.
\end{proposition}
\iffullpaper
\begin{proof}
Showing that the FD capacity region is convex is equivalent to showing that whenever (\ref{eq:concave-rm})--(\ref{eq:concave-rb}) hold, $r_m(r_b)$ and $r_b(r_m)$ intersect over an angle that is smaller than $\pi$ at the point $(s_b, s_m)$. (That is to say, the tangents of $r_m(r_b)$ and $r_b(r_m)$ at $(s_b, s_m)$ form an angle that is smaller than $\pi$.)

Observe the derivative of $r_m(r_b)$ with respect to $r_b$ at $r_b = s_b$ (equivalently $\alpha_b = 1$). From (\ref{eq:drm-drb}), (\ref{eq:dalphab-drb}), and (\ref{eq:drm-dalpham}):
\begin{align}
\left.\frac{d r_m}{d r_b}\right|_{r_b = s_b} &= \left.\left(\frac{d r_m}{d \alpha_b}\cdot \frac{d \alpha_b}{d r_b}\right)\right|_{\alpha_b = 1}\notag\\ 
&= - \frac{1+\overline{\gamma_{mm}}+\overline{\gamma_{bm}}}{1+\overline{\gamma_{bb}} + \overline{\gamma_{mb}}}\cdot \frac{\overline{\gamma_{bb}}}{1+\overline{\gamma_{bb}}}\cdot \frac{1}{\overline{\gamma_{bm}}}.\label{eq:drm-drb-alphab=1}
\end{align}
Symmetrically:
\begin{equation}
\left.\frac{d r_b}{d r_m}\right|_{r_m = s_m} = - \frac{1+\overline{\gamma_{bb}}+\overline{\gamma_{mb}}}{1+\overline{\gamma_{mm}}+\overline{\gamma_{mb}}}\cdot \frac{\overline{\gamma_{mm}}}{1+\overline{\gamma_{mm}}}\cdot \frac{1}{\overline{\gamma_{mb}}}.\label{eq:drb-drm-alpham=1}
\end{equation}
Observe that both $\left.\frac{d r_m}{d r_b}\right|_{r_b = s_b}< 0$ and $\left.\frac{d r_b}{d r_m}\right|_{r_m = s_m}<0$. Whenever $r_m(r_b)$ is concave and $r_b(r_m)$ is concave, for the capacity region to be convex it is necessary and sufficient that (see Fig.~\ref{fig:derivatives}):
\begin{align}
&\left(-\left.\frac{d r_b}{d r_m}\right|_{r_m = s_m}\right)^{-1} \geq - \left.\frac{d r_m}{d r_b}\right|_{r_b = s_b}\notag\\
\Leftrightarrow \quad & {\overline{\gamma_{mb}}}\cdot \frac{1+\overline{\gamma_{mm}}}{\overline{\gamma_{mm}}}\geq \frac{\overline{\gamma_{bb}}}{1+\overline{\gamma_{bb}}}\cdot \frac{1}{\overline{\gamma_{bm}}}\notag\\
\Leftrightarrow \quad & \overline{\gamma_{bm}}\overline{\gamma_{mb}}\geq \frac{\overline{\gamma_{mm}}\overline{\gamma_{bb}}}{(1+\overline{\gamma_{mm}})(1+\overline{\gamma_{bb}})}.\label{eq:gammabm-gammamb}
\end{align}
Recall (from (\ref{eq:gammabm-cond-1})) that for $r_m(r_b)$ to be concave, it must be:
\begin{equation}
\overline{\gamma_{bm}} > \overline{\gamma_{bb}}(1+\overline{\gamma_{mm}})\frac{2+\overline{\gamma_{mb}}}{1+\overline{\gamma_{mb}}} >  \overline{\gamma_{bb}}(1+\overline{\gamma_{mm}}) \geq \frac{\overline{\gamma_{bb}}}{1+\overline{\gamma_{mm}}}. \label{eq:first-half-gammabm-gammamb}
\end{equation}
Symmetrically:
\begin{equation}
\overline{\gamma_{bm}} >  \frac{\overline{\gamma_{mm}}}{1+\overline{\gamma_{bb}}}. \label{eq:second-half-gammabm-gammamb}
\end{equation}
Combining (\ref{eq:first-half-gammabm-gammamb}) and (\ref{eq:second-half-gammabm-gammamb}) gives (\ref{eq:gammabm-gammamb}), and therefore, the capacity region is convex whenever $r_m(r_b)$ and $r_b(r_m)$ are both concave (which, in turn, is equivalent to (\ref{eq:concave-rm})--(\ref{eq:concave-rb}) both being true).
\begin{figure}[t]
\center
\subfloat[]{\label{fig:derivatives-equal}\includegraphics[height = .8in]{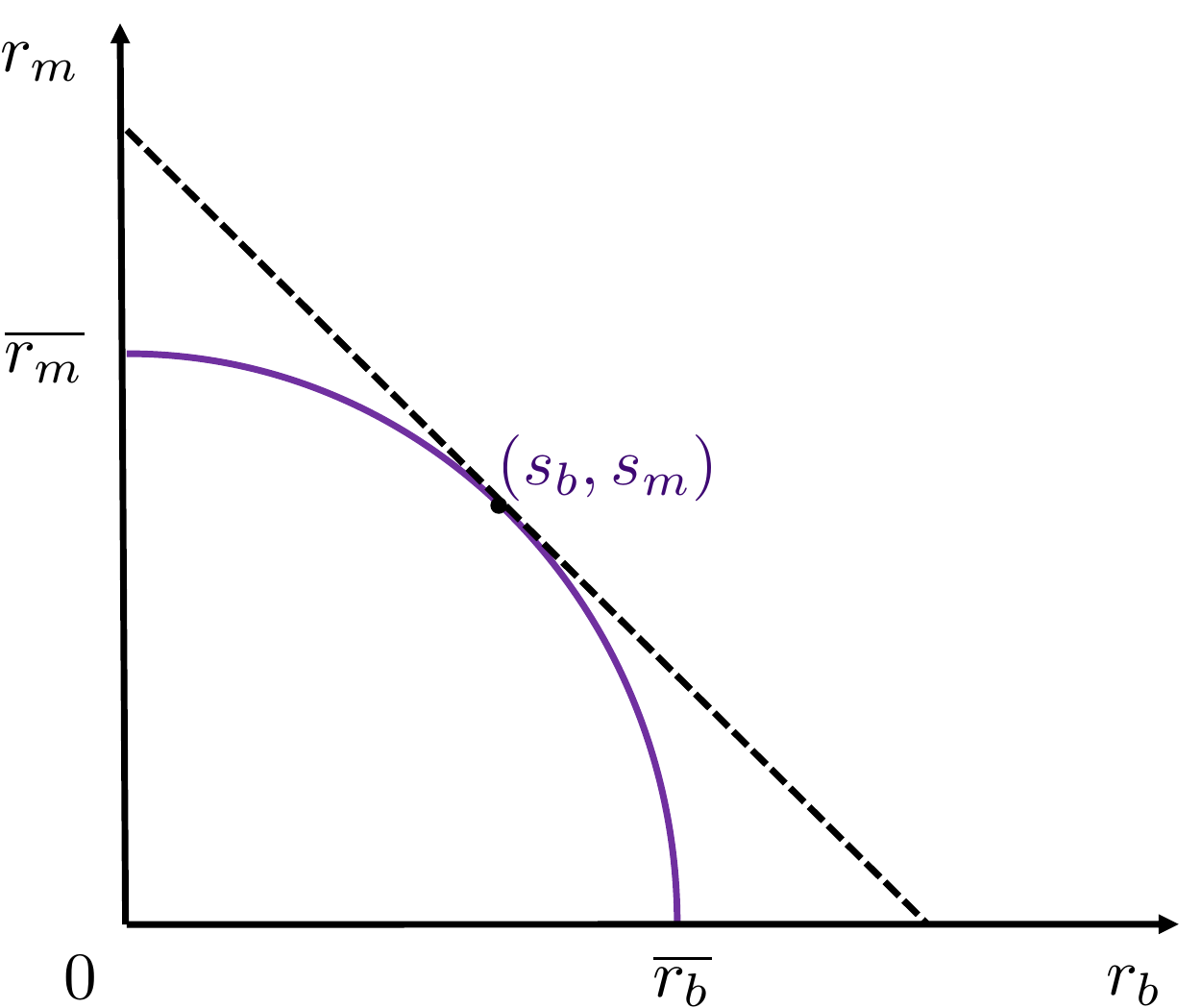}}\hspace{\fill}
\subfloat[]{\label{fig:derivatives-less}\includegraphics[height = .8in]{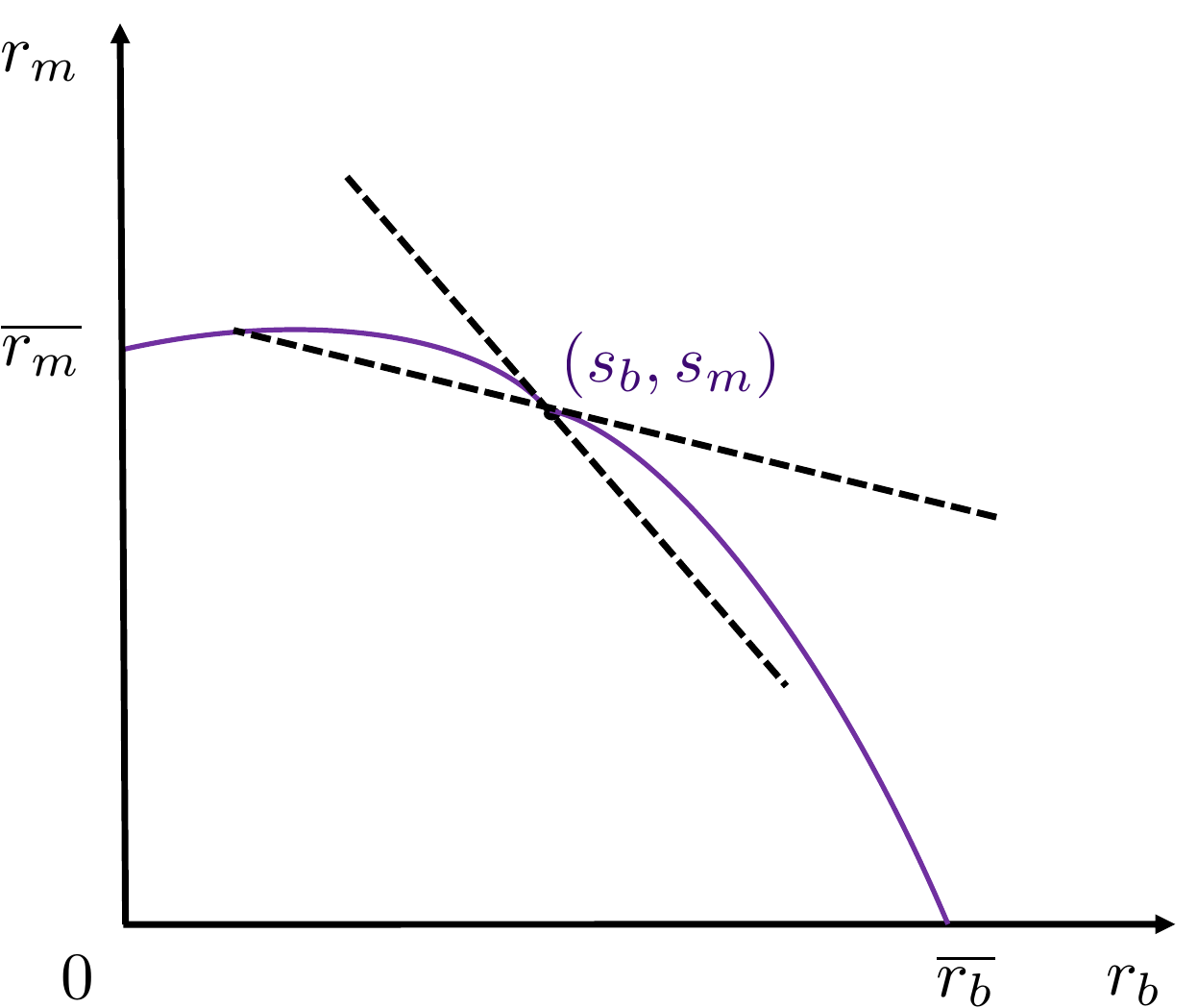}}\hspace{\fill}
\subfloat[]{\label{fig:derivatives-greater}\includegraphics[height = .8in]{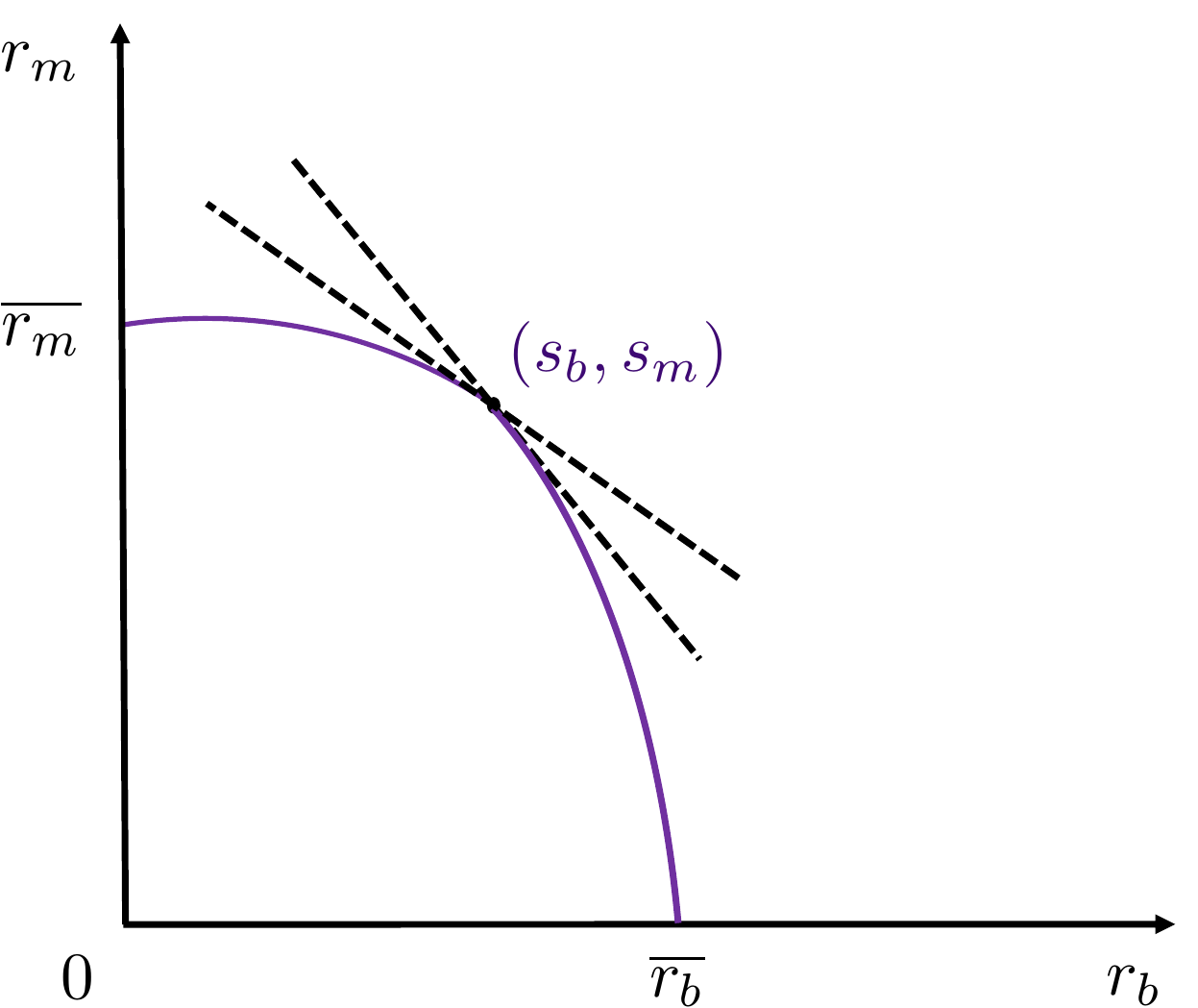}}
\caption{Possible intersections of $r_m(r_b)$ and $r_b(r_m)$ at $(s_b, s_m)$: \protect\subref{fig:derivatives-equal} $-\left(\left.\frac{d r_b}{d r_m}\right|_{r_m = s_m}\right)^{-1} = -\left.\frac{d r_m}{d r_b}\right|_{r_b = s_b}$, \protect\subref{fig:derivatives-less} $-\left(\left.\frac{d r_b}{d r_m}\right|_{r_m = s_m}\right)^{-1} < -\left.\frac{d r_m}{d r_b}\right|_{r_b = s_b}$, and \protect\subref{fig:derivatives-greater} $-\left(\left.\frac{d r_b}{d r_m}\right|_{r_m = s_m}\right)^{-1} > -\left.\frac{d r_m}{d r_b}\right|_{r_b = s_b}$}.
\label{fig:derivatives}
\end{figure}
\end{proof}
\fi

Fig.~\ref{fig:cap-region-convexity} illustrates the regions of (maximum) SNR values $\overline{\gamma_{bm}}$ and $\overline{\gamma_{mb}}$ for which the FD capacity region is convex, for different values of $\overline{\gamma_{mm}}$ and $\overline{\gamma_{bb}}$, compared to the maximum achievable rate improvements. The black line 
delimits the region of $\overline{\gamma_{bm}}$ and $\overline{\gamma_{mb}}$ for which the FD region is convex: north and east from it, the  region is convex, while south and west from it, the region is not convex. As Fig.~\ref{fig:cap-region-convexity} suggests, high (over $1.6\times$) rate improvements are mainly achievable in the area where the FD region is convex, unless one of the SNR values $\overline{\gamma_{bm}}$ and $\overline{\gamma_{mb}}$ is much higher than the other.

\begin{figure*}[t!]
\center
\subfloat[$\overline{\gamma_{bb}}=0$dB, $\overline{\gamma_{mm}}=0$dB]{\label{fig:cap_region_00}\includegraphics[scale = 0.22 ]{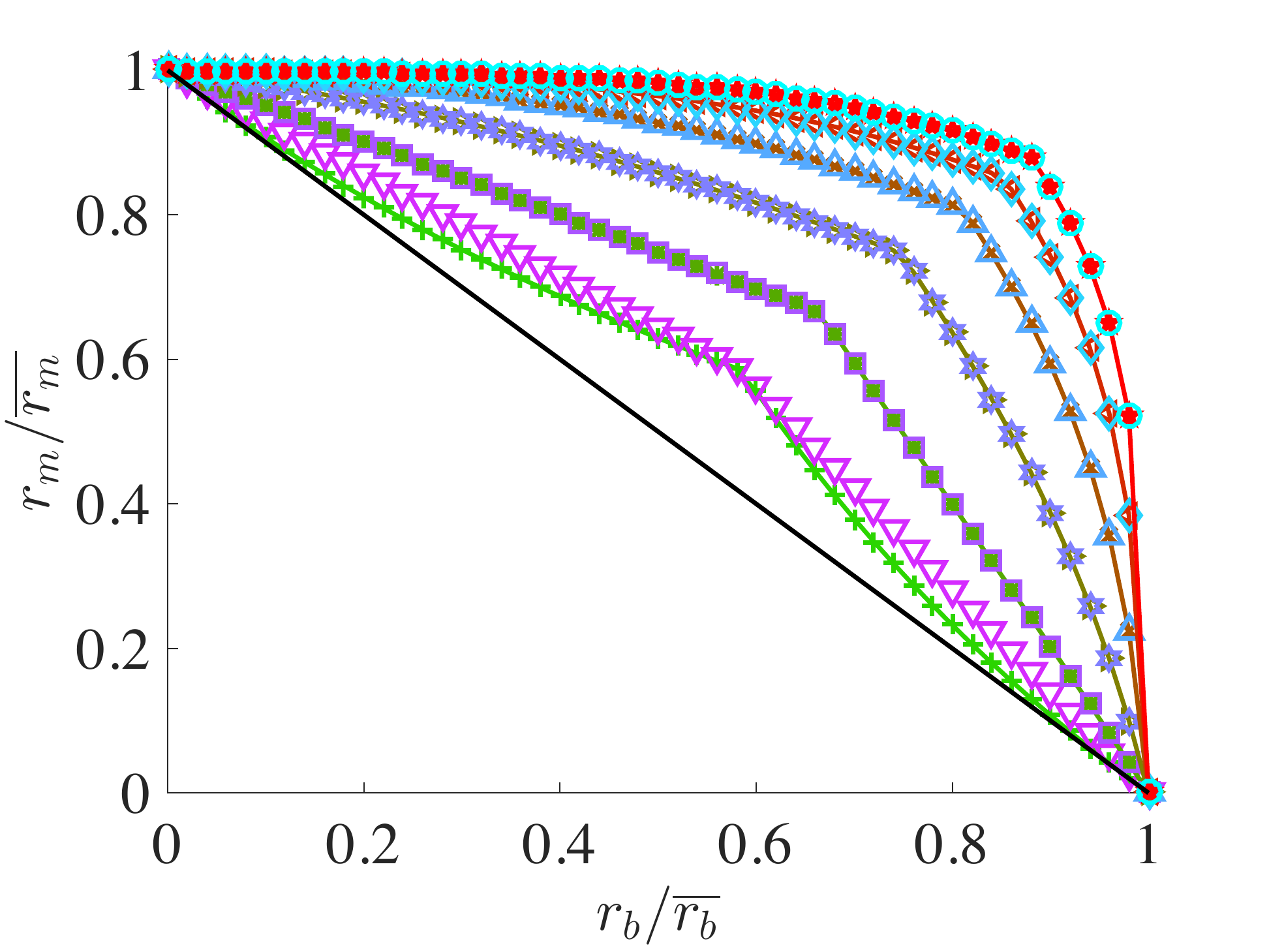}}\hspace{\fill}
\subfloat[$\overline{\gamma_{bb}}=0$dB, $\overline{\gamma_{mm}}=5$dB]{\label{fig:cap_region_05}\includegraphics[scale = 0.22]{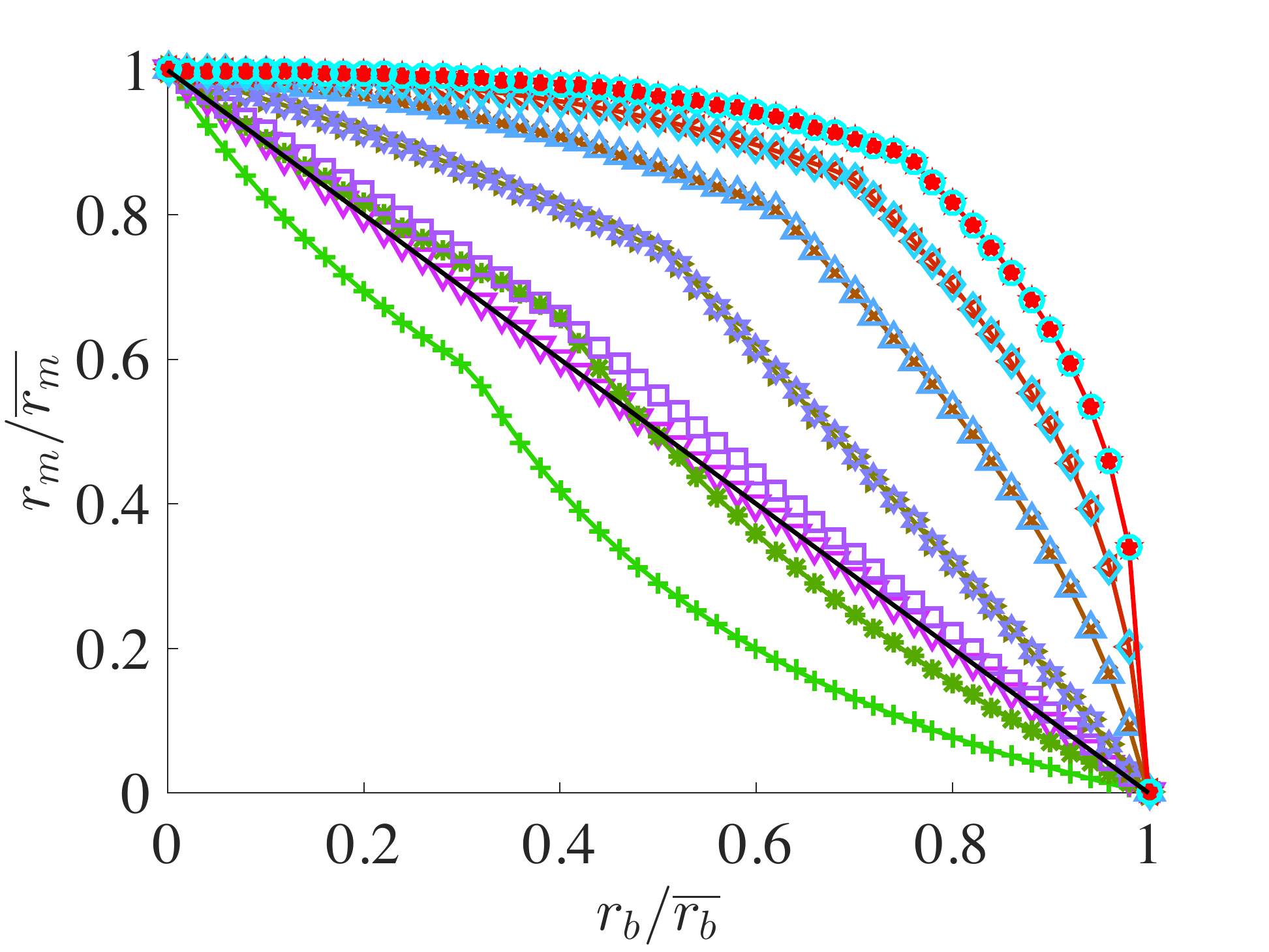}}\hspace{\fill}
\subfloat[$\overline{\gamma_{bb}}=0$dB, $\overline{\gamma_{mm}}=10$dB]{\label{fig:cap_region_010}\includegraphics[scale = 0.22]{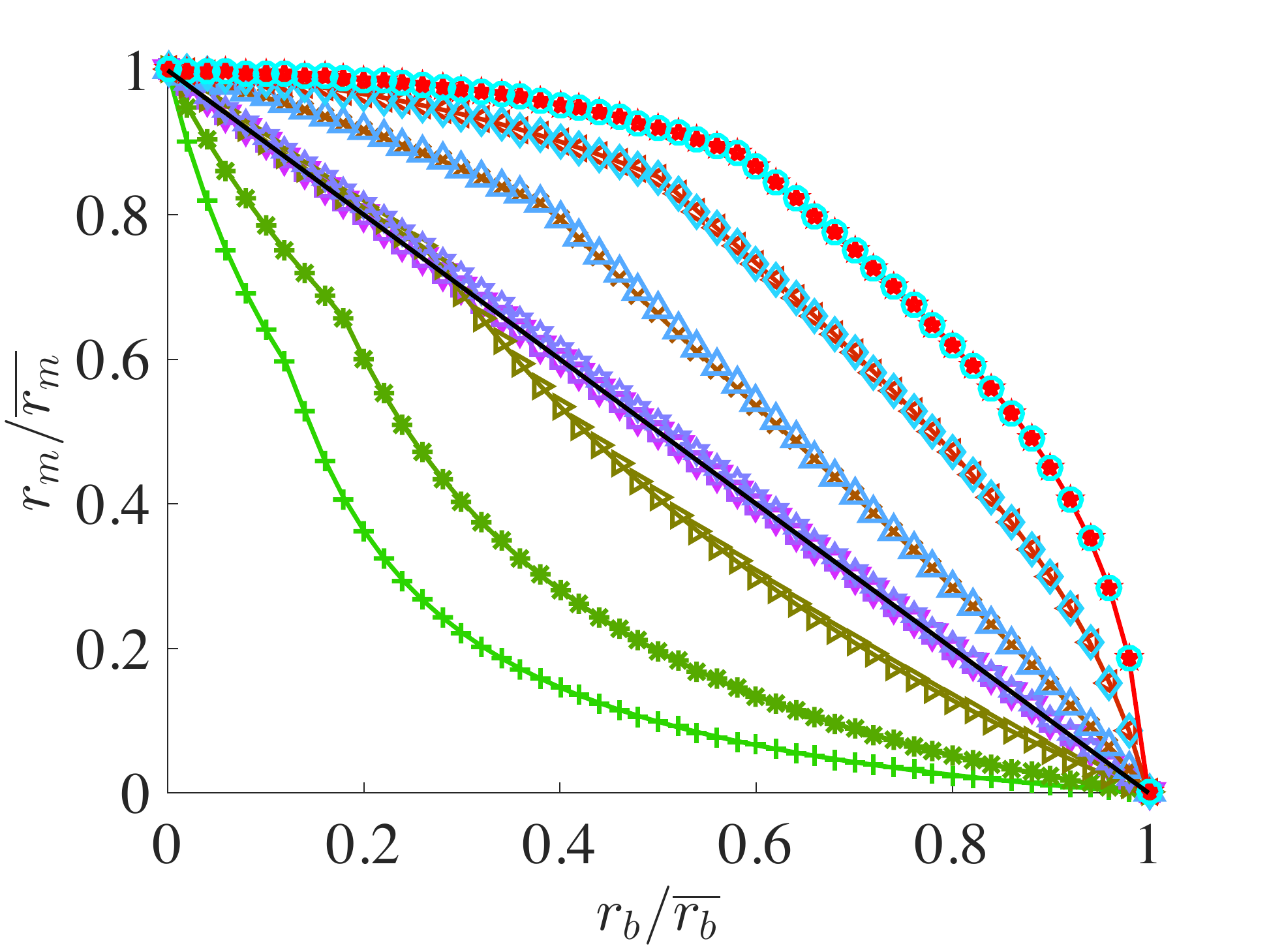}}\hspace{\fill}
\subfloat{\label{fig:cap_region_legend}\includegraphics[scale = 0.25]{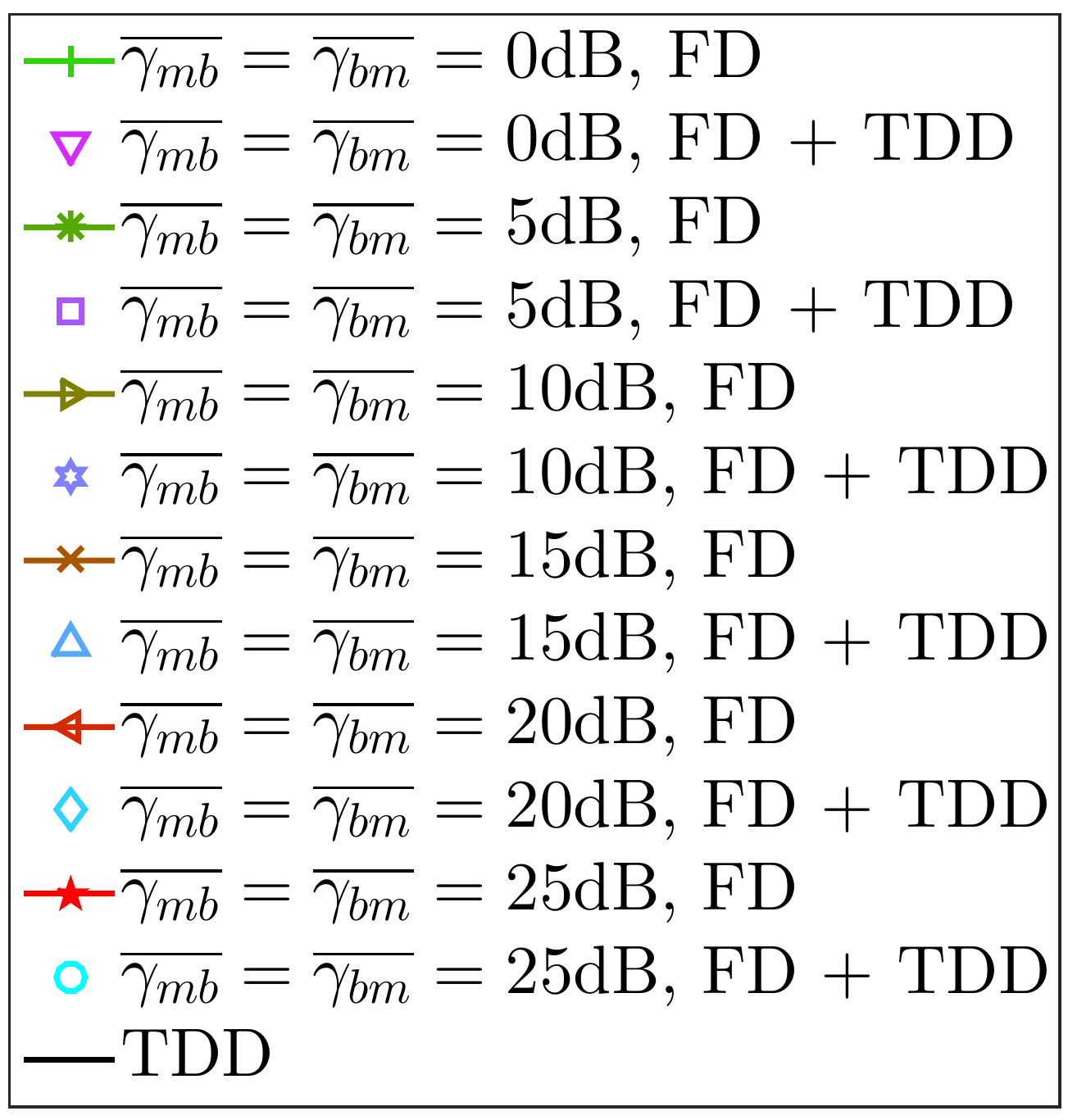}}\\\vspace{-10pt}
\setcounter{subfigure}{3}
\subfloat[$\overline{\gamma_{bb}}=0$dB, $\overline{\gamma_{mm}}=0$dB]{\label{fig:asym_cap_region_00}\includegraphics[scale = 0.22 ]{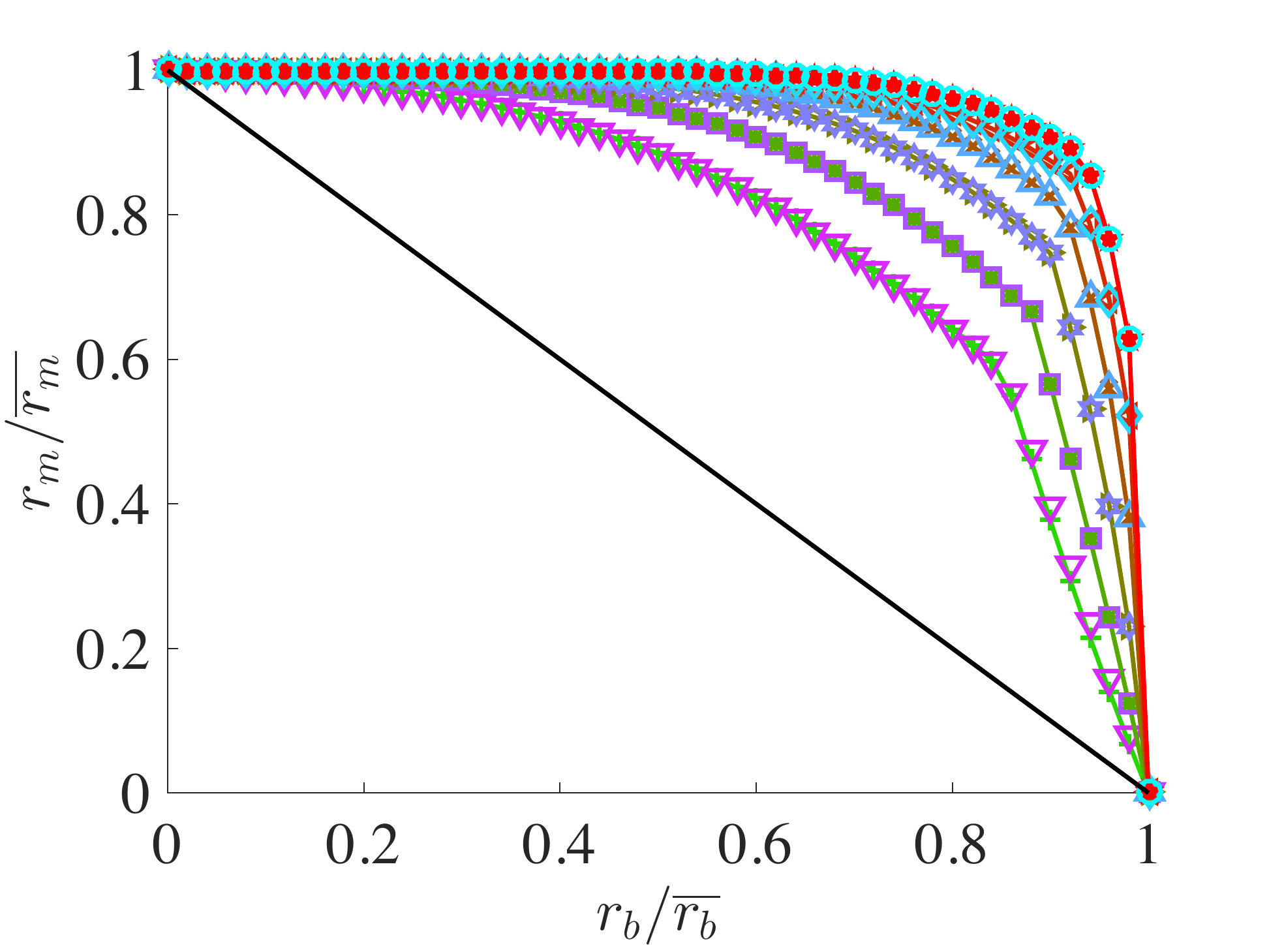}}\hspace{\fill}
\subfloat[$\overline{\gamma_{bb}}=0$dB, $\overline{\gamma_{mm}}=5$dB]{\label{fig:asym_cap_region_05}\includegraphics[scale = 0.22]{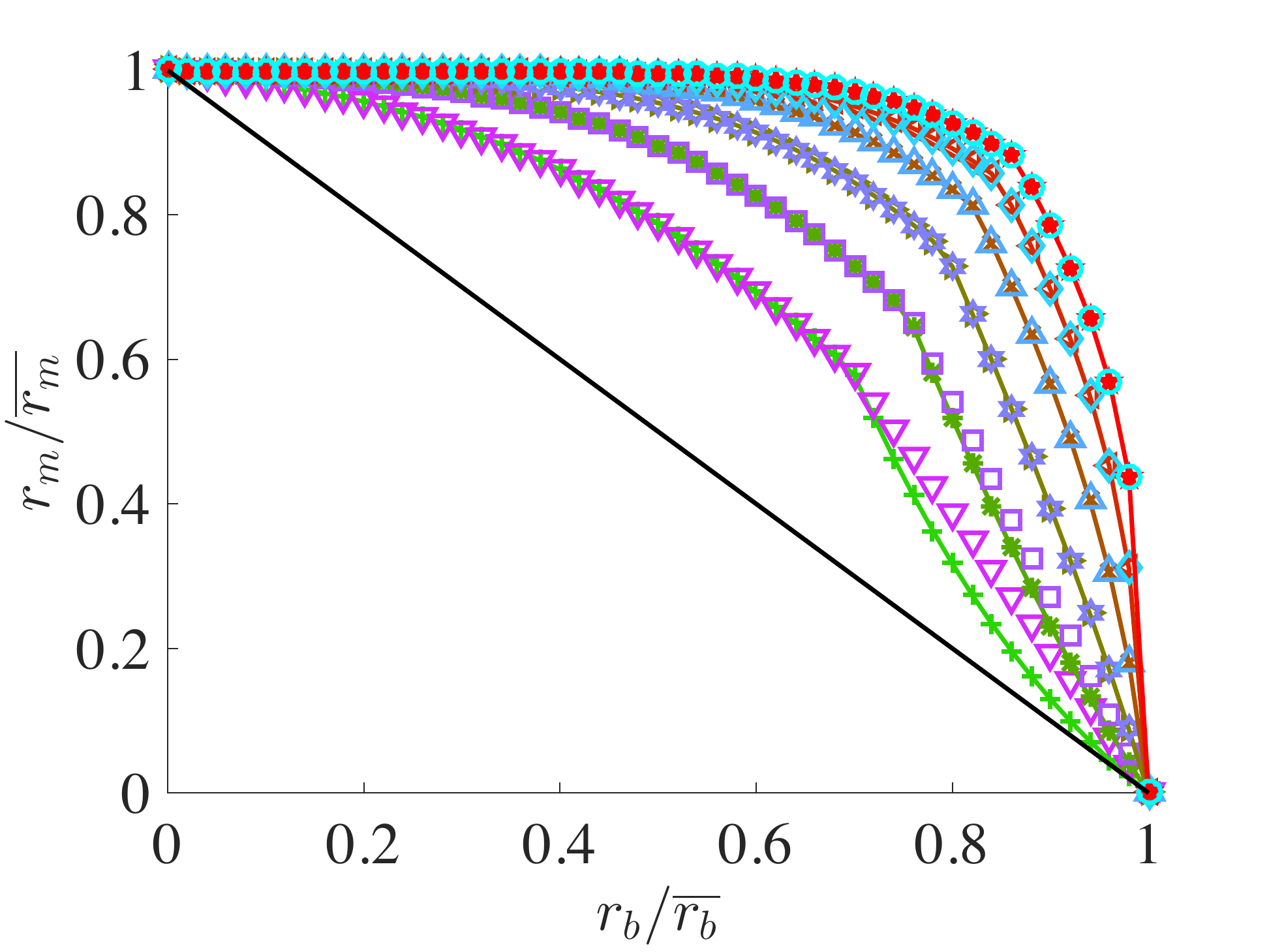}}\hspace{\fill}
\subfloat[$\overline{\gamma_{bb}}=0$dB, $\overline{\gamma_{mm}}=10$dB]{\label{fig:asym_cap_region_010}\includegraphics[scale = 0.22]{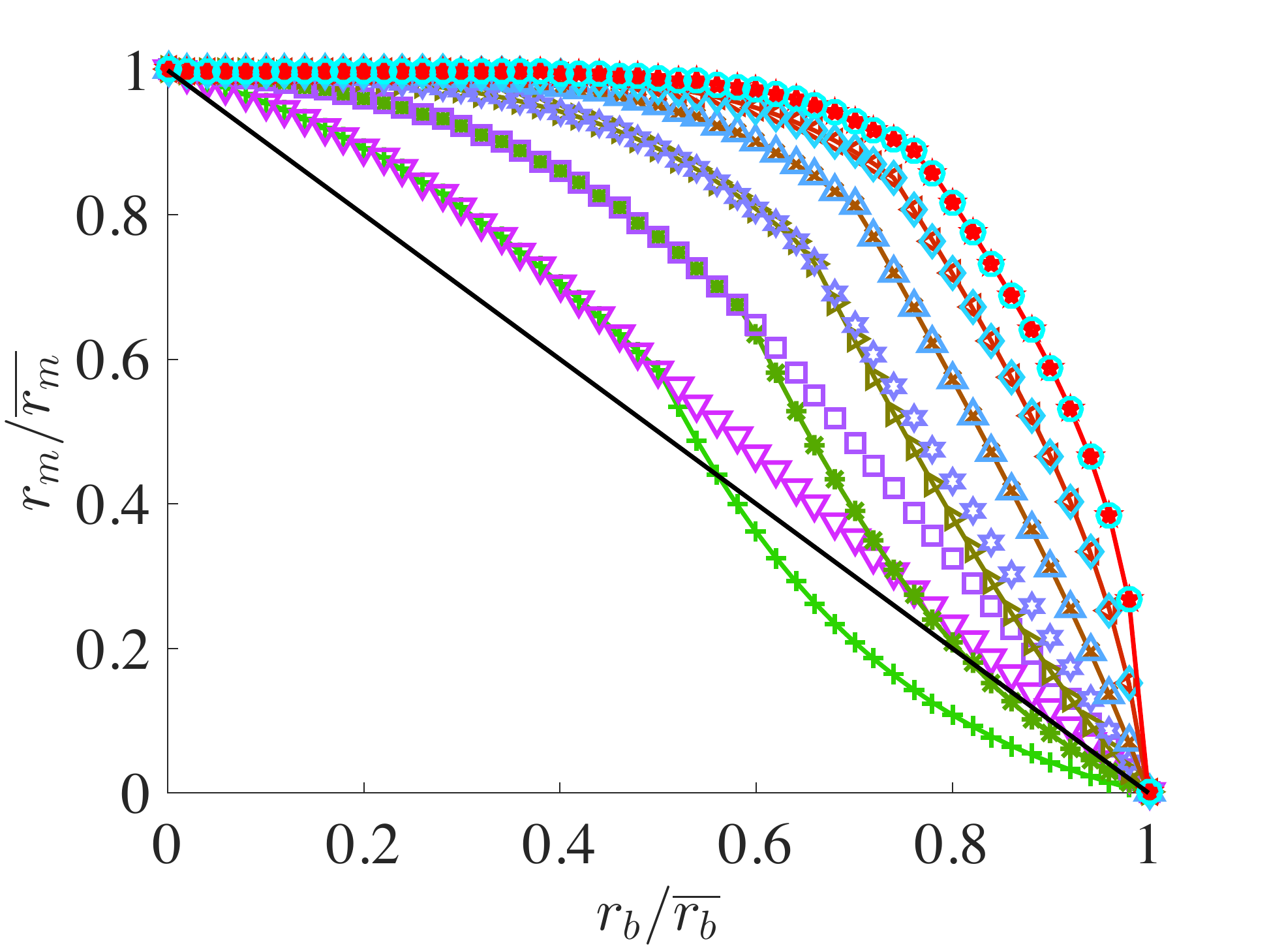}}\hspace{\fill}
\subfloat{\label{fig:asym_cap_region_legend}\includegraphics[scale = 0.25]{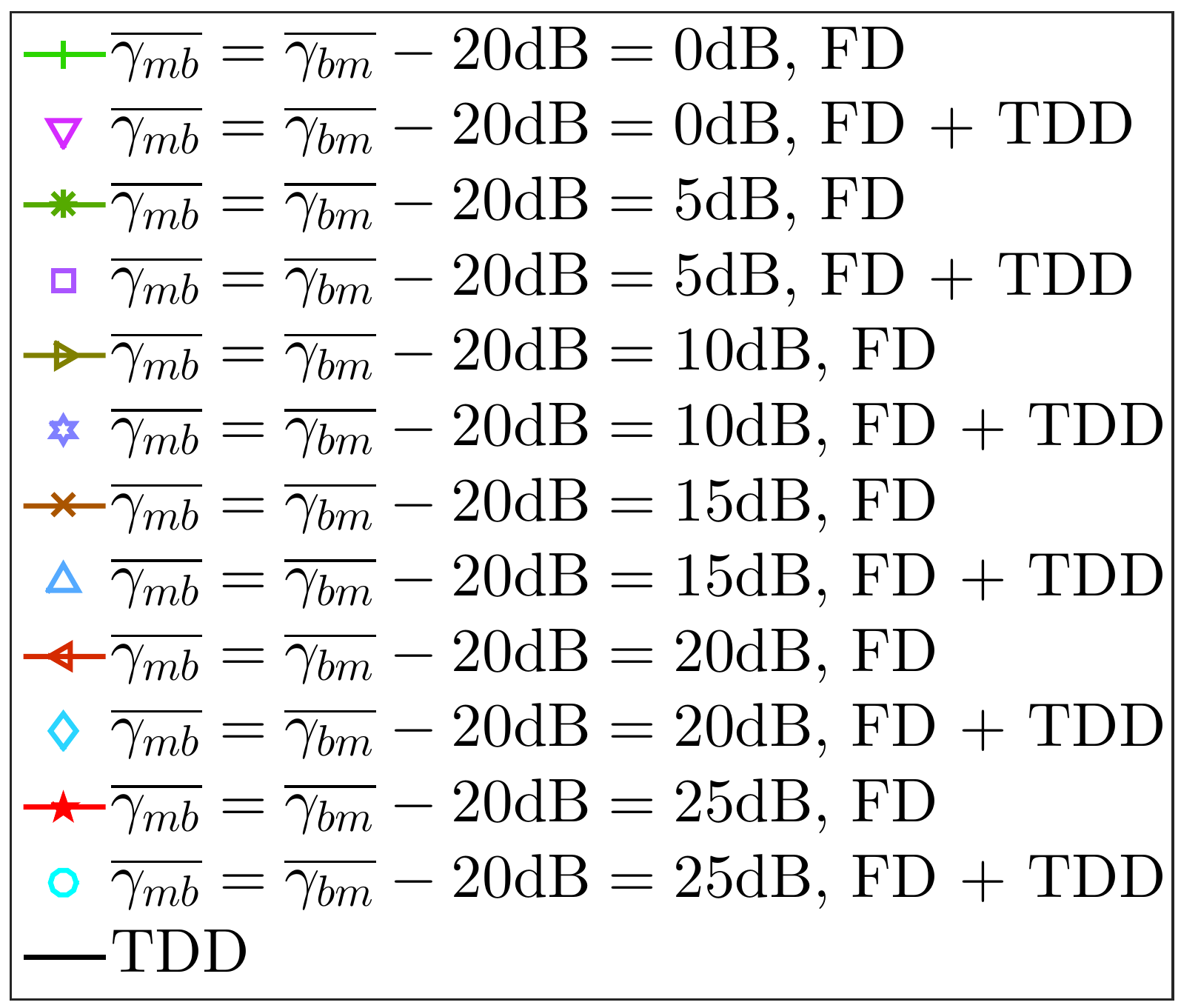}}
\caption{Capacity regions for \protect\subref{fig:cap_region_00}--\protect\subref{fig:cap_region_010} $\overline{\gamma_{bm}} = \overline{\gamma_{mb}}$ and \protect\subref{fig:asym_cap_region_00}--\protect\subref{fig:asym_cap_region_010} $\overline{\gamma_{bm}} > \overline{\gamma_{mb}}$.}\vspace{-10pt}
\label{fig:cap-region-fd-tdd} 
\end{figure*}
\iffullpaper
\begin{figure*}[t!]
\center
\subfloat[$\overline{\gamma_{bb}}=0$dB, $\overline{\gamma_{mm}}=0$dB]{\label{fig:rate_improve_00}\includegraphics[scale = 0.21 ]{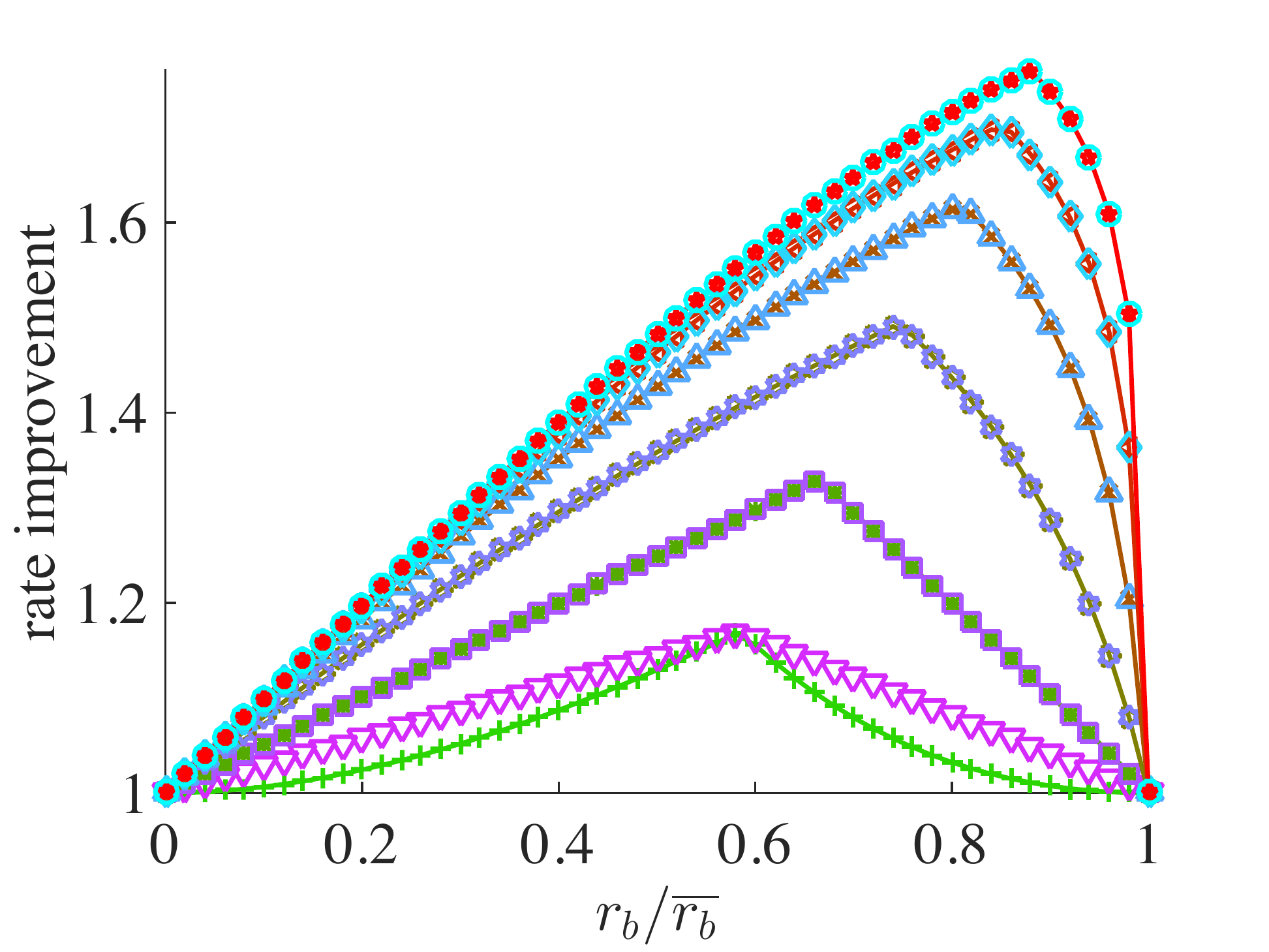}}\hspace{\fill}
\subfloat[$\overline{\gamma_{bb}}=0$dB, $\overline{\gamma_{mm}}=5$dB]{\label{fig:rate_improve_05}\includegraphics[scale = 0.22]{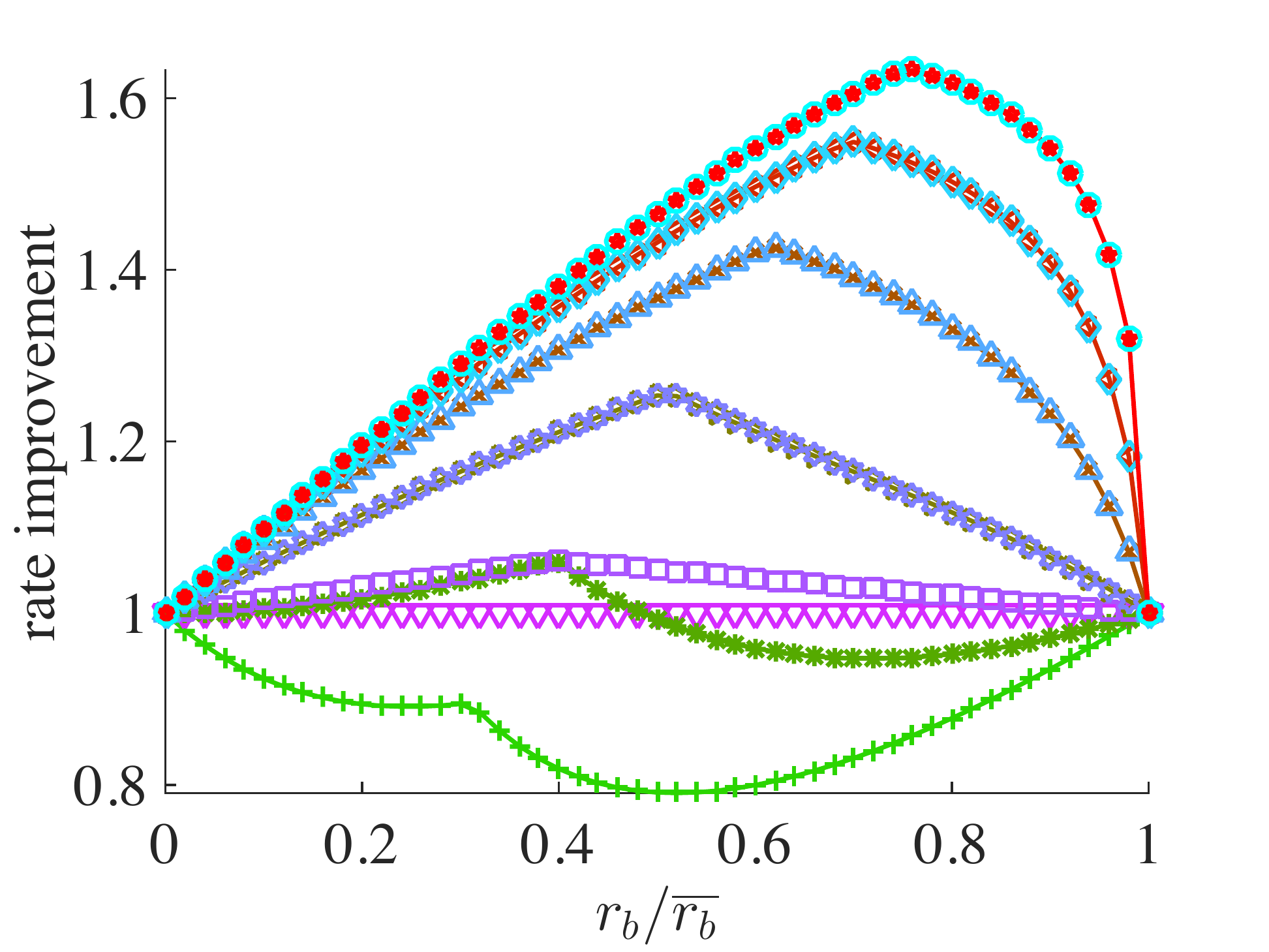}}\hspace{\fill}
\subfloat[$\overline{\gamma_{bb}}=0$dB, $\overline{\gamma_{mm}}=10$dB]{\label{fig:rate_improve_010}\includegraphics[scale = 0.21]{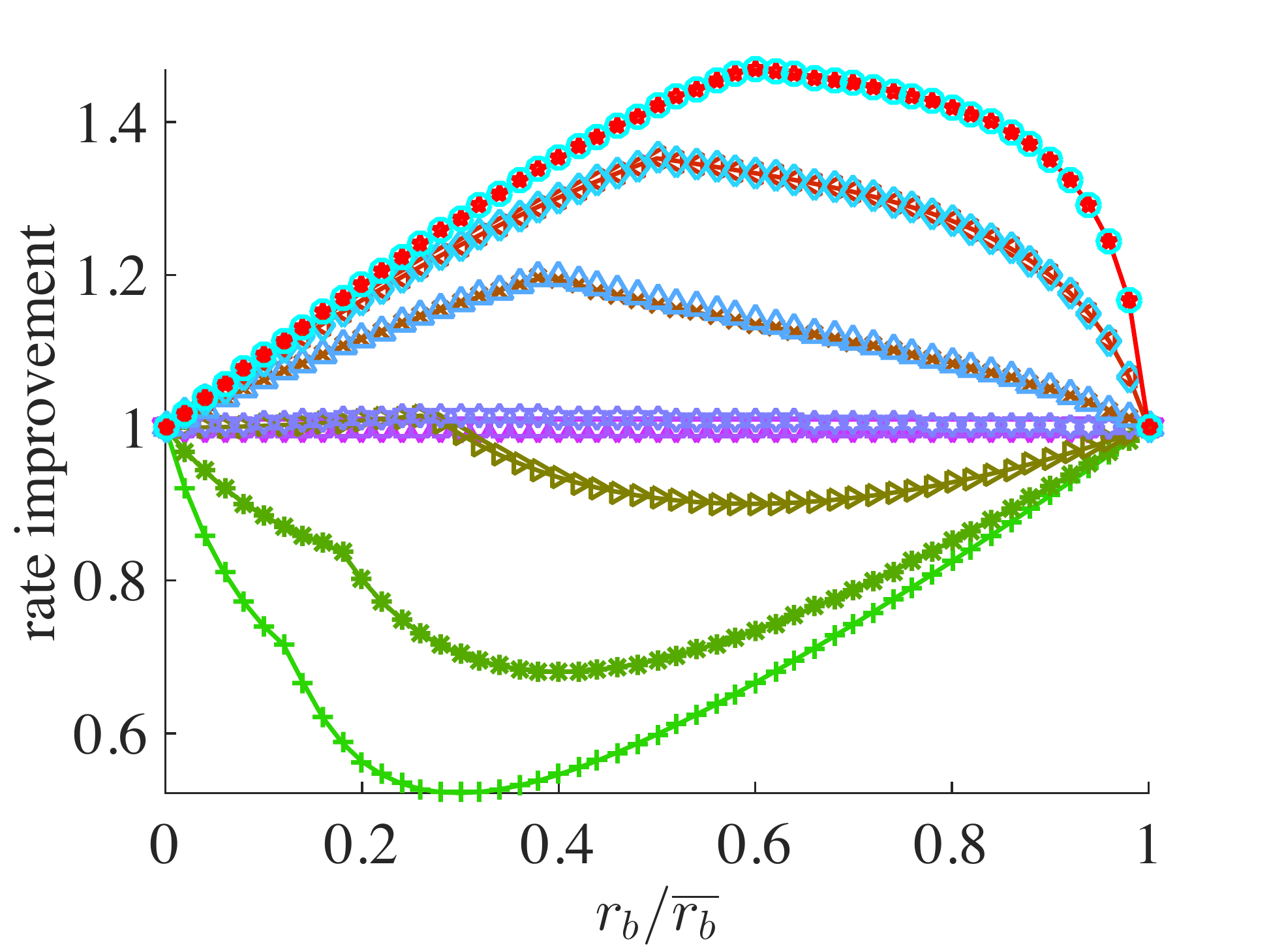}}\hspace{\fill}
\subfloat{\label{fig:rate_improve_legend}\includegraphics[scale = 0.25]{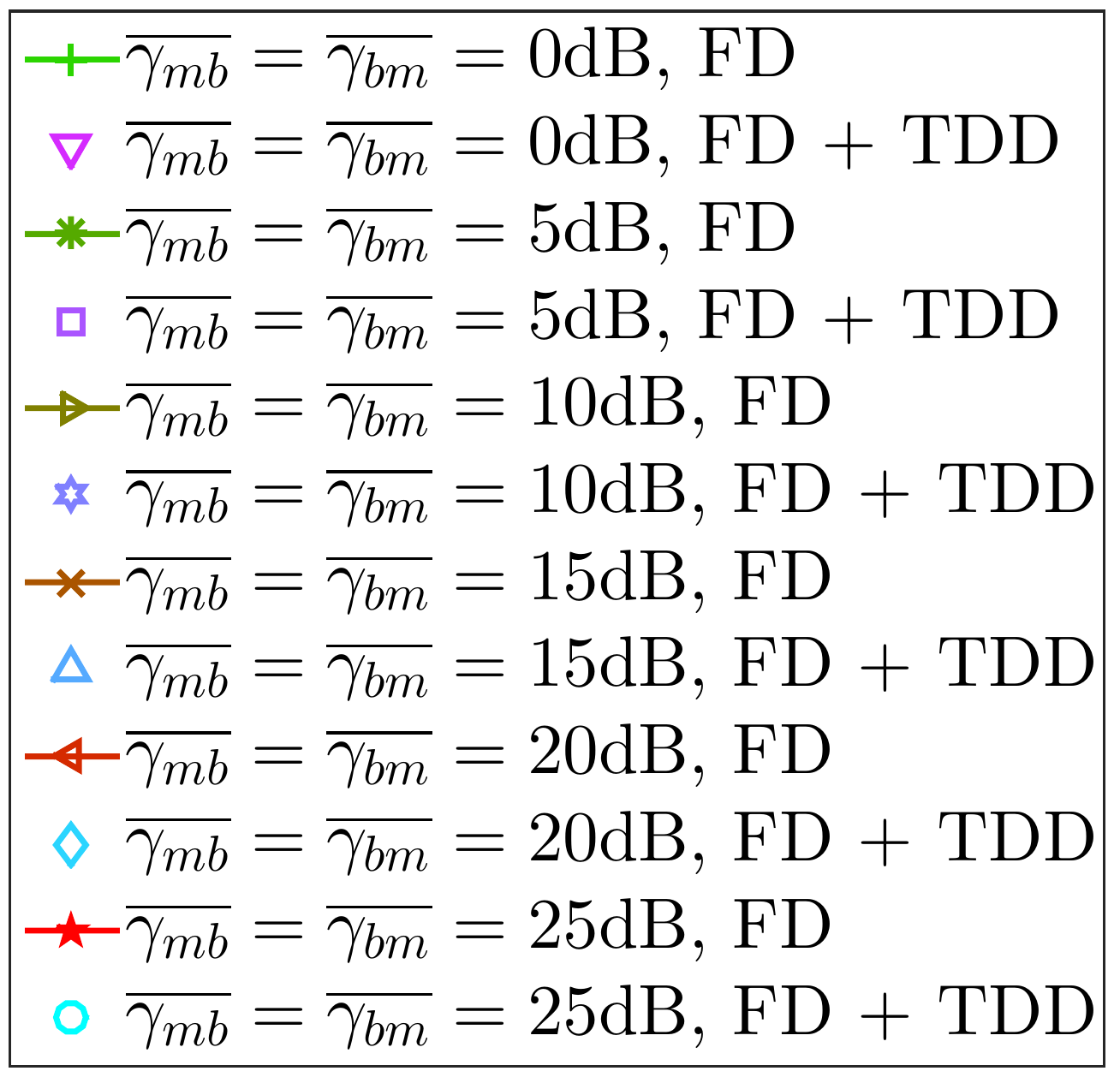}}\\\vspace{-10pt}
\setcounter{subfigure}{3}
\subfloat[$\overline{\gamma_{bb}}=0$dB, $\overline{\gamma_{mm}}=0$dB]{\label{fig:asym_rate_improve_00}\includegraphics[scale = 0.22 ]{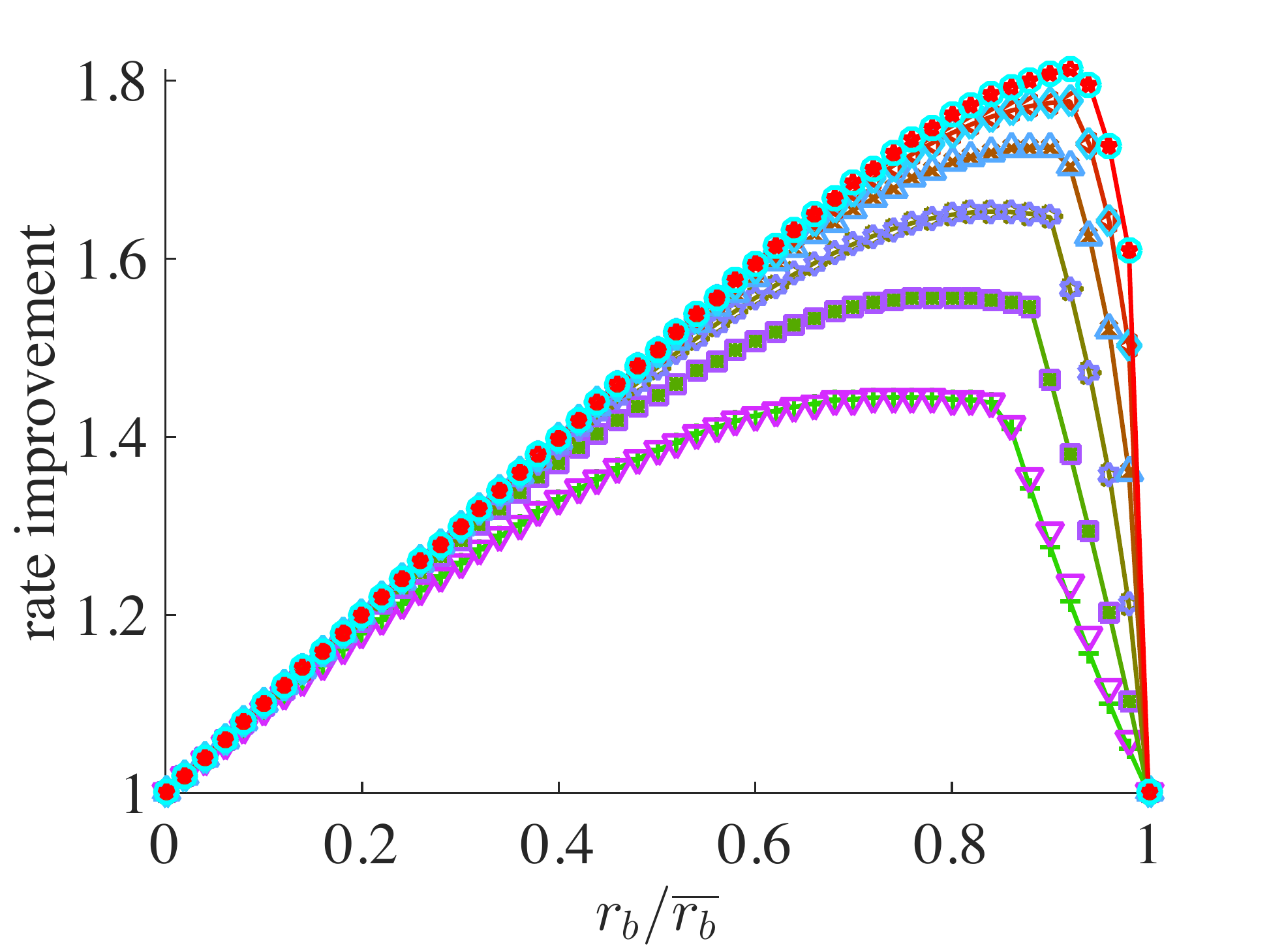}}\hspace{\fill}
\subfloat[$\overline{\gamma_{bb}}=0$dB, $\overline{\gamma_{mm}}=5$dB]{\label{fig:asym_rate_improve_05}\includegraphics[scale = 0.22]{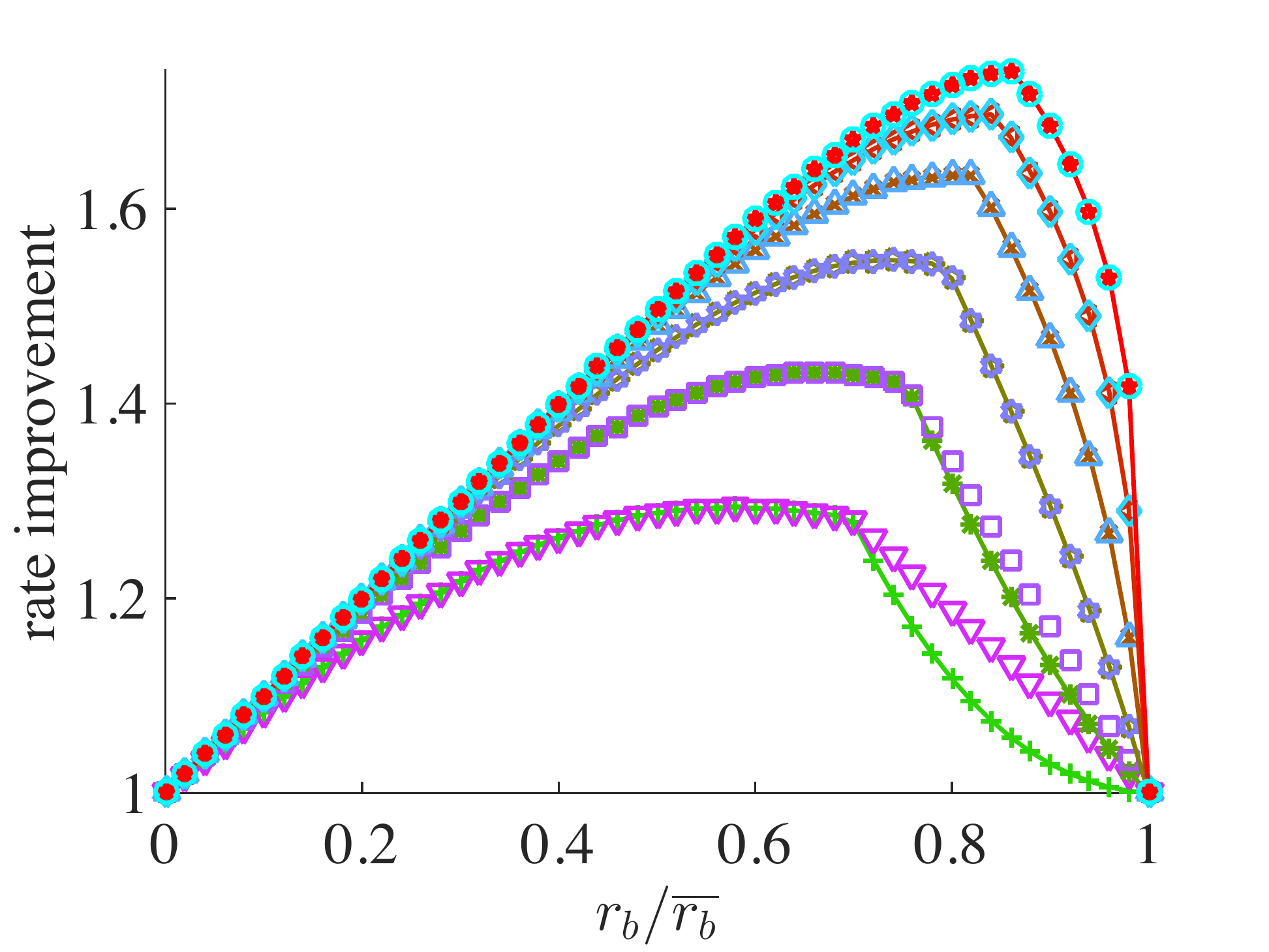}}\hspace{\fill}
\subfloat[$\overline{\gamma_{bb}}=0$dB, $\overline{\gamma_{mm}}=10$dB]{\label{fig:asym_rate_improve_010}\includegraphics[scale = 0.22]{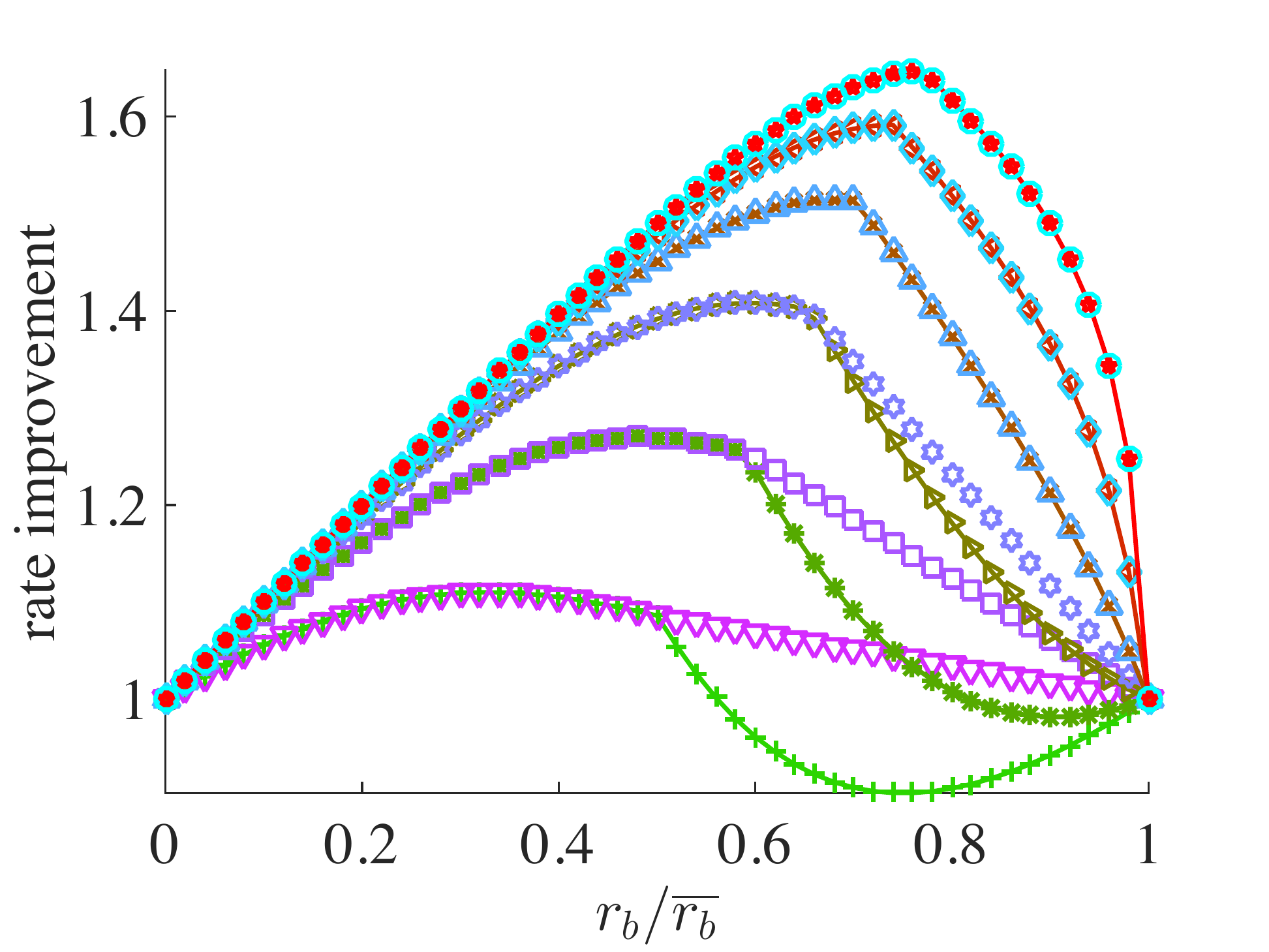}}\hspace{\fill}
\subfloat{\label{fig:asym_rate_improve_legend}\includegraphics[scale = 0.25]{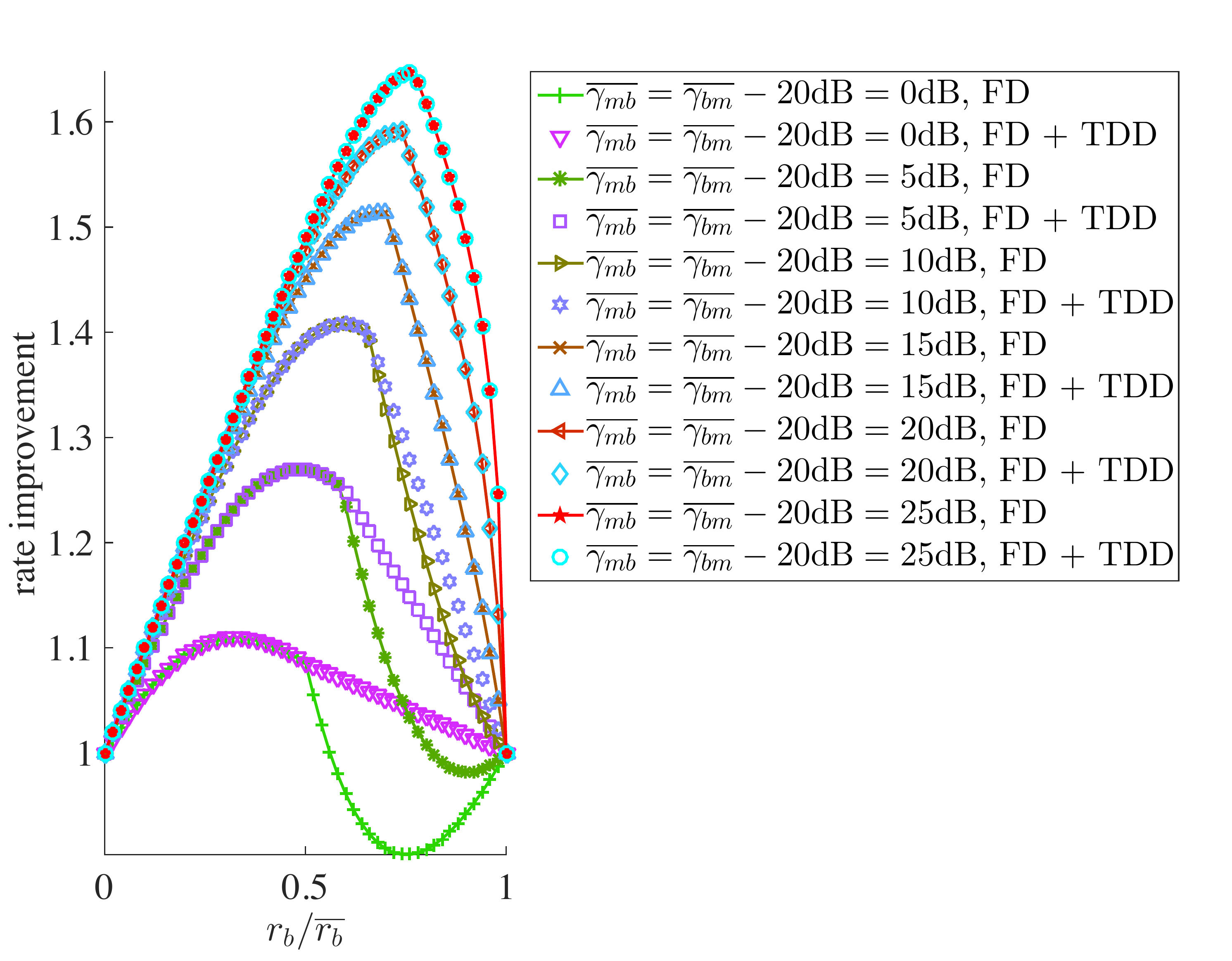}}
\caption{Rate improvements for  \protect\subref{fig:cap_region_00}--\protect\subref{fig:cap_region_010} $\overline{\gamma_{bm}} = \overline{\gamma_{mb}}$ and \protect\subref{fig:asym_cap_region_00}--\protect\subref{fig:asym_cap_region_010} $\overline{\gamma_{bm}} > \overline{\gamma_{mb}}$.}
\vspace{-10pt}
\label{fig:rate_improve-fd-tdd} 
\end{figure*}
\else
\fi

\subsection{Determining TDFD Capacity Region}\label{sec:single-algo}
We now turn to the problem of allocating UL and DL rates, possibly through a combination of FD and TDD, which is equivalent to determining the TDFD capacity region. As before, the problem is to maximize $r_m$ subject to $r_b = r_b^*$ and the power constraints. 
Denote the maximum $r_m$ such that $r_b = r_b^*$ as $r_m^*$. 
From Lemma \ref{lemma:convexity-of-cap-region}, there can be 3 cases:

\noindent\textbf{Case 1:} $r_m(r_b)$ is concave for all $r_b \in [0, s_b]$. From concavity of $r_m(r_b)$, \iffullpaper it follows that \fi$r_m^* = r_m(\alpha_b, 1)$, where $\alpha_b$ solves $r_b(\alpha_b, 1) = r_b^*$.

\noindent\textbf{Case 2:} $r_m(r_b)$ is convex for all $r_b \in [0, s_b]$. Using convexity, if the rate improvement at $(s_b, s_m)$ is less than 1 and $r_b(r_m)$, $r_m \in [0, s_m]$, is convex, it is optimal to use TDD. If the rate improvement is at least 1, it is optimal to place $r_m^*$ on the line connecting $(0, \overline{r_m})$ and $(s_b, s_m)$\iffullpaper through a time sharing between these rate pairs\fi. If $r_b(r_m)$ for $r_m \in [0, s_m]$ is concave for $r_m \in [0, r_m^+]$, where $r_m^+\leq s_b$, and the rate improvement at $(s_b, s_m)$ is less than 1,  $r_m^*$ will lie on the boundary of the TDFD, but not 
FD, capacity region. 

\noindent\textbf{Case 3:} $r_m(r_b)$ is strictly concave for $r_b \in [0, r_b^+)$, strictly convex for $r_b \in (r_b^+, s_b]$, and $\frac{d r_m}{d r_b} = 0$ at $r_b = r_b^+$, where $r_b^+ < s_b$. Then, $r_m^*$ may lie either on (the boundary of) FD or TDFD capacity region, even if $r_b^* \leq r_b^+$.

To determine the optimal $r_m^*$ in Cases 2 and 3, we need to ``convexify'' the capacity region. 
Fortunately, the problem has enough structure so that this ``convexification'' can be done efficiently. We show the following propositions, which will lead to the convexified region. 

\begin{proposition}\label{prop:sb-sm-extremum}
If $(s_b, s_m)$ maximizes the sum of UL and DL rates, then $(s_b, s_m) \geq \lambda (r_b', r_m') + (1-\lambda) (r_b'', r_m'')$ element-wise for any $\lambda \in [0, 1]$, and any two feasible rate pairs $(r_b', r_m')$ and $(r_b'', r_m'')$.
\end{proposition}
\iffullpaper
\begin{proof}
Suppose that for some $\lambda\in [0, 1]$ and some pairs of feasible rates $(r_b', r_m')$ and $(r_b'', r_m'')$: $(s_b, s_m) < \lambda (r_b', r_m') + (1-\lambda) (r_b'', r_m'')$. Then either $(r_b', r_m') > (s_b, s_m)$ or $(r_b'', r_m'') > (s_b, s_m)$, and therefore $r_b'+r_m' > s_b + s_m$ or $r_b''+r_m''>s_b+s_m$, which is a contradiction, as $s_b+s_m$ maximizes the sum of the (UL and DL) rates.
\end{proof}
\fi
\iffullpaper Proposition \ref{prop:sb-sm-extremum} implies that if $(s_b, s_m)$ maximizes the sum of uplink and downlink rates, it must dominate any convex combination of other points from the capacity region. \fi 

\begin{proposition}\label{prop:sb-sm-not-extremum}
If $s_b+s_m < \overline{r_m}$, then $r_m(r_b)$ is convex on the entire segment from $(0, \overline{r_m})$ to $(s_b, s_m)$. \iffullpaper Similarly, if $s_b+s_m < \overline{r_b}$, then $r_b(r_m)$ is convex on the entire segment from $(\overline{r_b}, 0)$ to $(s_b, s_m)$.\else\fi
\end{proposition}
\iffullpaper
\begin{proof}
Suppose that $s_b+s_m < \overline{r_m}$ (the case $s_b+s_m < \overline{r_b}$ is symmetric). Then:
\begin{align}
&\log\Big(1+\frac{\overline{\gamma_{bm}}}{1+\overline{\gamma_{mm}}}\Big) + \log\Big(1+\frac{\overline{\gamma_{mb}}}{1+\overline{\gamma_{bb}}}\Big) < \log\left(1+ \overline{\gamma_{mb}}\right)\notag \\
\Leftrightarrow \quad &\log\bigg(\Big(1+\frac{\overline{\gamma_{bm}}}{1+\overline{\gamma_{mm}}}\Big)\cdot \Big(1+\frac{\overline{\gamma_{mb}}}{1+\overline{\gamma_{bb}}}\Big)\bigg) < \log\left(1+ \overline{\gamma_{mb}}\right)\notag \\
\Leftrightarrow \quad &\Big(1+\frac{\overline{\gamma_{bm}}}{1+\overline{\gamma_{mm}}}\Big)\cdot \Big(1+\frac{\overline{\gamma_{mb}}}{1+\overline{\gamma_{bb}}}\Big) < 1+ \overline{\gamma_{mb}}\notag \\
\Leftrightarrow \quad & \frac{1+ \overline{\gamma_{bb}} + \overline{\gamma_{mb}}}{1+\overline{\gamma_{mb}}}\cdot (1+\overline{\gamma_{mm}}+\overline{\gamma_{bm}}) < (1+\overline{\gamma_{bb}})(1+\overline{\gamma_{mm}})\notag\\
\Leftrightarrow \quad & (1+\overline{\gamma_{mm}})\Big(1+\frac{\overline{\gamma_{bb}}}{1+\overline{\gamma_{mb}}}-1\Big)\notag\\ 
&+ \overline{\gamma_{bm}}\Big(1+\frac{\overline{\gamma_{bb}}}{1+\overline{\gamma_{mb}}}\Big) < \overline{\gamma_{bb}}(1+\overline{\gamma_{mm}})\notag\\
\Rightarrow \quad & \overline{\gamma_{bm}} < \overline{\gamma_{bb}}(1+\overline{\gamma_{mm}}), \label{eq:gbm<gbb(1+gmm)}
\end{align}
as $(1+\overline{\gamma_{mm}})\big(1+\frac{\overline{\gamma_{bb}}}{1+\overline{\gamma_{mb}}}-1\big) = (1+\overline{\gamma_{mm}})\cdot\frac{\overline{\gamma_{bb}}}{1+\overline{\gamma_{mb}}} \geq 0$ and $1+\frac{\overline{\gamma_{bb}}}{1+\overline{\gamma_{mb}}}\geq 1$. 

A necessary condition for $r_m(r_b)$ to be concave for any $r_b \in [0, s_b]$ (see proof of Lemma \ref{lemma:convexity-of-cap-region}, eq. (\ref{eq:gammabm-cond-1})) is that $\overline{\gamma_{bm}} > \overline{\gamma_{bb}}(1+\overline{\gamma_{mm}})$. Therefore, from (\ref{eq:gbm<gbb(1+gmm)}), $r_m(r_b)$ is convex for any $r_b \in [0, s_b]$.
\end{proof}
\fi
Now we are ready to handle Case 3 and the last part of Case 2, in the following (constructive) lemma. 
Note that the convexification needs to be performed only once; after that, $r_m(r_b)$ (and $r_b(r_m)$) can be represented in a black-box manner, requiring constant computation to determine any rate pair $(r_b^*, r_m^*)$, given either $r_b^*$ or $r_m^*$. 
\begin{lemma}\label{lemma:convexification}
The boundary of the TDFD capacity region can be determined in time {$O(\log(\varepsilon^{-1}\overline{r_b}))$}, where $\varepsilon$ is the additive error of $r_m^* = \max\{r_m: r_b=r_b^*\}$, and the binary search, if employed, takes at most $\lceil\log(\varepsilon^{-1}\cdot 1.4\overline{r_b})\rceil$ steps. 
\end{lemma}
\iffullpaper
\begin{proof}
Note that the time to determine $r_m^*$ on the boundary of TDFD capacity region may not be constant only in Case 3 and the last part of Case 2. We start with the proof for Case 3.

If $s_b + s_m \geq \max\{r_b, r_m\}$, then, from Lemma 5.1 in \cite{full-duplex-sigmetrics}, $(s_b, s_m)$ maximizes the sum of uplink and downlink rates. Using Proposition \ref{prop:sb-sm-extremum} and simple geometric arguments, it follows that in the ``convexified'' capacity region there exists $r_b' \leq r_b^+$ such that the boundary of the region is equal to $r_m(r_b)$ for $r_b \in [0, r_b']$ joined with a line segment from a point $(r_b', r_m(r_b'))$ to $(s_b, s_m)$, where the line through points $(r_b', r_m(r_b'))$ and $(s_b, s_m)$ is tangent to $r_m(r_b)$ at point $(r_b', r_m(r_b'))$ (see Fig. \ref{fig:case 3}\subref{fig:symmetric}). Since the tangent from $(s_b, s_m)$ onto $r_m(r_b)$ must touch $r_m(r_b)$ at a point $(r_b', r_m(r_b'))$ where $r_m(r_b)$ is concave, it follows that we can find $r_b'$ by performing a binary search over $r_b\in [0, r_b^+]$, since every concave function has a monotonically decreasing derivative. It follows that $r_m^* = r_m(r_b^*)$ if $r_b^* \leq r_b'$, and $r_m^* = r_m(r_b') + \left.\frac{d r_m}{d r_b}\right|_{r_b = r_b'}(r_b^* - r_b')$.
\begin{figure}[t!]
\center
\hspace{\fill}\subfloat[]{\label{fig:symmetric}\includegraphics[height = 1in]{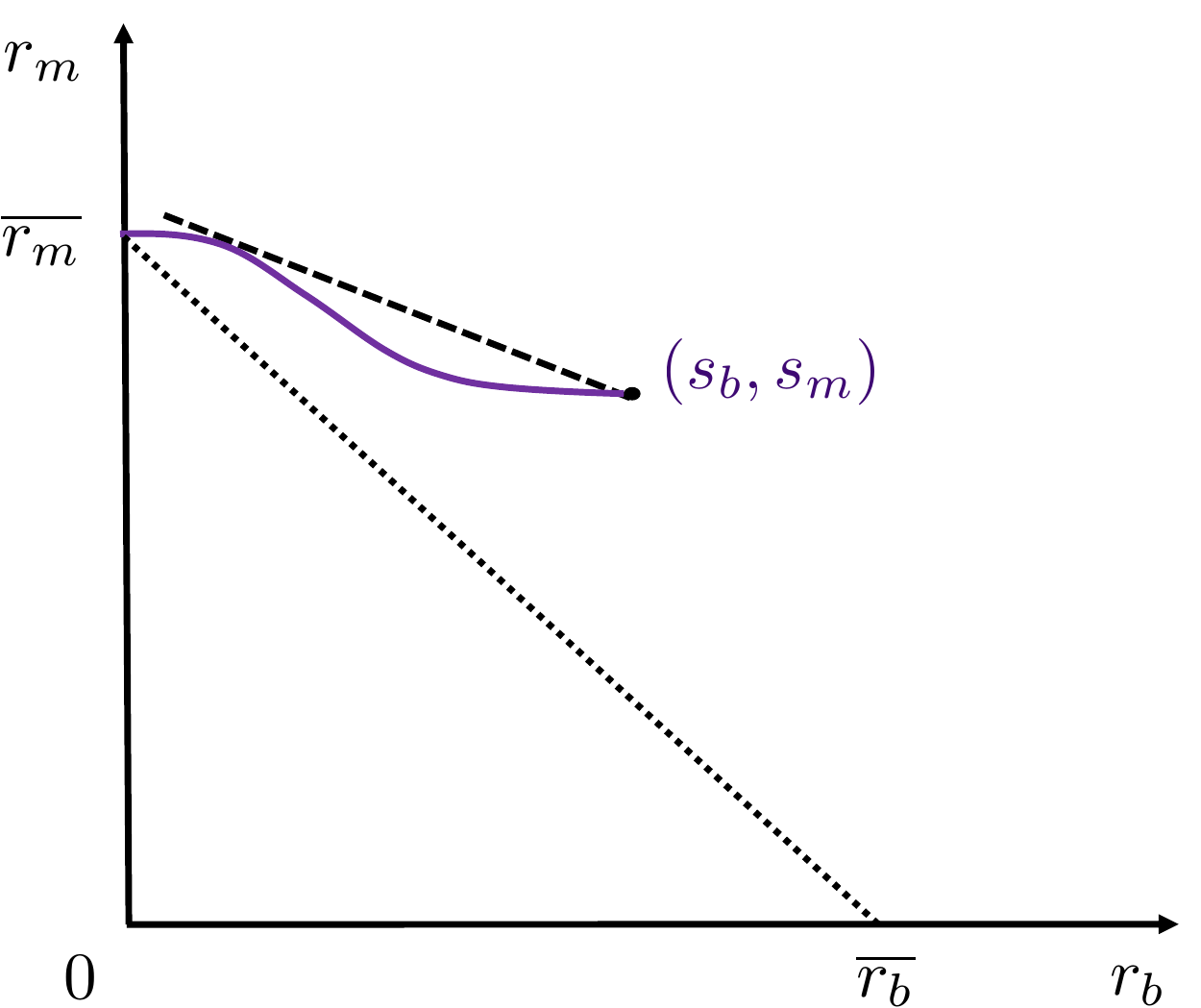}}\hspace{\fill}
\subfloat[]{\label{fig:non-symmetric}\includegraphics[height = 1in]{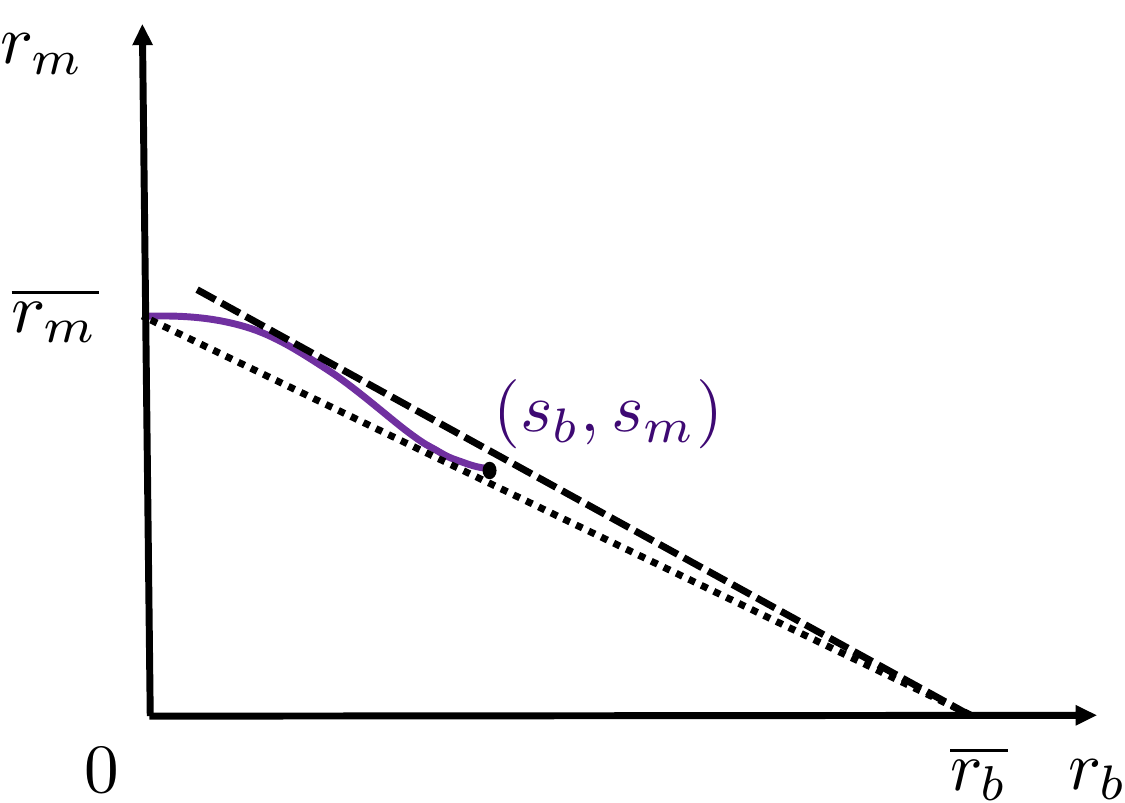}}\hspace*{\fill}
\caption{Two possible scenarios for Case 3.}
\label{fig:case 3}
\end{figure}

Consider now the case that $s_b + s_m < \max\{r_b, r_m\}$. Using the same approach as described above, we can determine a point $r_b'\leq r_b^+)$ such that the line through $(r_b', r_m(r_b'))$ and $(s_b, s_m)$ is tangent onto $r_m(r_b)$. However, this approach may not always lead to the convexified capacity region. 

Consider the case illustrated in Fig.~\ref{fig:case 3}\subref{fig:non-symmetric}. From Proposition \ref{prop:sb-sm-not-extremum}, $r_b(r_m)$ for $r_m \in [0, s_m]$ must be convex, and therefore there exists an $r_b'' \leq r_b^+$ such that the boundary of the convexified capacity region is determined by $r_m(r_b)$ for $r_b\in [0, r_b'']$ and by a line through $(r_b'', r_m(r_b''))$ and $(s_b, s_m)$ for $r_b \in [r_b'', s_b]$, where the line through $(r_b'', r_m(r_b''))$ and $(s_b, s_m)$ is tangent onto $r_m(r_b)$ at point $r_b = r_b''$. Since $r_b''$ must belong to the segment where $r_m(r_b)$ is concave, it follows that $r_b''$ can be found through a binary search over $r_b \in [0, r_b^+]$. To determine which one of the two tangents delimits the convexified capacity region, it is sufficient to compare $r_m(r_b') +  \left.\frac{d r_m}{d r_b}\right|_{r_b = r_b'}(s_b - r_b')$ and $r_m(r_b'') +  \left.\frac{d r_m}{d r_b}\right|_{r_b = r_b''}(s_b - r_b'')$ and choose the one with the maximum value.  

The last part of the Case 2 is symmetric to the case illustrated in Fig.~\ref{fig:case 3}\subref{fig:non-symmetric}, and can be handled by the approach described above.

Finally, we need to show that the binary search can be implemented with low running time. To do so, we first bound the change in the derivative $\frac{d r_m}{d r_b}$ on the segment where $r_m(r_b)$ is concave.

\begin{proposition}
For all $r_b\in [0, s_b]$ such that $r_m(r_b)$ is concave: $\big|\frac{d^2 r_m}{d r_b^2}\big|<1.4$.
\end{proposition}
\begin{proof}
Fix any $r_b$ such that $r_m(r_b)$ is concave, and let $\alpha_b$ be such that $r_b = r_b(\alpha_b, 1)$. The proof of Lemma \ref{lemma:convexity-of-cap-region} implies that (using Eq.'s (\ref{eq:d2rm-drb2})--(\ref{eq:rm-concave-cond})):
\begin{align}
\Big|\frac{d^2 r_m}{d r_b^2}\Big| \leq & \frac{\overline{\gamma_{bb}}}{\ln(2)}\Big(\frac{1}{1+\alpha_b \overline{\gamma_{bb}}} - \frac{1}{1+\alpha_b \overline{\gamma_{bb}}+\overline{\gamma_{mb}}}\Big)\notag\\ 
&\cdot \ln^2(2)\cdot 2^{r_b} \frac{1+\overline{\gamma_{mm}}}{\overline{\gamma_{bm}}}\notag\\
=& \ln(2) \overline{\gamma_{bb}}\cdot \frac{\overline{\gamma_{mb}}}{(1+\alpha_b\overline{\gamma_{bb}})(1+\alpha_b\overline{\gamma_{bb}}+\overline{\gamma_{mb}})}\notag\\ 
&\cdot \Big(\alpha_b + \frac{1+\overline{\gamma_{mm}}}{\overline{\gamma_{bm}}}\Big)\notag\\
<&\ln(2) \frac{\overline{\gamma_{bb}}}{1+\alpha_b\overline{\gamma_{bb}}}\Big(\alpha_b + \frac{1+\overline{\gamma_{mm}}}{\overline{\gamma_{bm}}}\Big)\notag\\
=& \ln(2) \Big(\frac{\alpha_b\overline{\gamma_{bb}}}{1+\alpha_b\overline{\gamma_{bb}}} + \frac{\overline{\gamma_{bb}}(1+\overline{\gamma_{mm}})}{\overline{\gamma_{bm}}}\cdot \frac{1}{1+\alpha_b\overline{\gamma_{bb}}}\Big)\notag\\
\leq & 2\ln(2) < 1.4,
\end{align}
where we have used: $\frac{\overline{\gamma_{mb}}}{1+\alpha_b\overline{\gamma_{bb}}+\overline{\gamma_{mb}}}\leq \frac{\overline{\gamma_{mb}}}{1+\overline{\gamma_{mb}}}<1$, $\frac{\alpha_b\overline{\gamma_{bb}}}{1+\alpha_b\overline{\gamma_{bb}}}<1$, $\frac{1}{1+\alpha_b\overline{\gamma_{bb}}}\leq 1$, and $\frac{\overline{\gamma_{bb}}(1+\overline{\gamma_{mm}})}{\overline{\gamma_{bm}}} < 1$ (from a necessary condition (\ref{eq:rm-concave-cond}) for $r_m(r_b)$ to be concave in any $r_b\in [0, s_b]$ in the proof of Lemma \ref{lemma:convexity-of-cap-region}).
\end{proof}

For $r_b'$ or $r_b''$ to be determined with an absolute error $\varepsilon$, it takes at most $\lceil\log(\varepsilon^{-1})\rceil$ binary search steps. In terms of $r_m^*$, the error is then less than $1.4\varepsilon \overline{r_b}$, and to find $r_m^*$ with an absolute error $\varepsilon$, the binary search should perform at most $\lceil\log(\varepsilon^{-1}\cdot 1.4\overline{r_b})\rceil$ steps. 
\qed
\end{proof}
\fi 

\iffullpaper To put the number of binary search steps in perspective, the highest SNR typically measured in Wi-Fi and cellular networks is about 50dB ($10^5$). 50dB SNR maps to $\overline{r_b} \approx 16.61$ b/s/Hz, which result in at most $\lceil 4.53 + \log(\varepsilon^{-1}) \rceil$ binary search steps. Since each step requires constant computation time, the  computation time for determining the convexified capacity region is very low. \fi

Using the methods mentioned above, FD and TDFD capacity regions 
were obtained 
for different combinations of $\overline{\gamma_{bm}}, \overline{\gamma_{mb}}, \overline{\gamma_{mm}}$, and $\overline{\gamma_{bb}}$ (Fig.~\ref{fig:cap-region-fd-tdd}).  
As expected, as $\overline{\gamma_{mm}}$ increases and $\overline{\gamma_{mb}}$ and $\overline{\gamma_{bm}}$ decrease, the rate improvements decrease and more FD regions become non-convex.

\begin{figure*}[t!]
\center
\subfloat[]{\label{fig:cap_region_conv}\includegraphics[scale = 0.22 ]{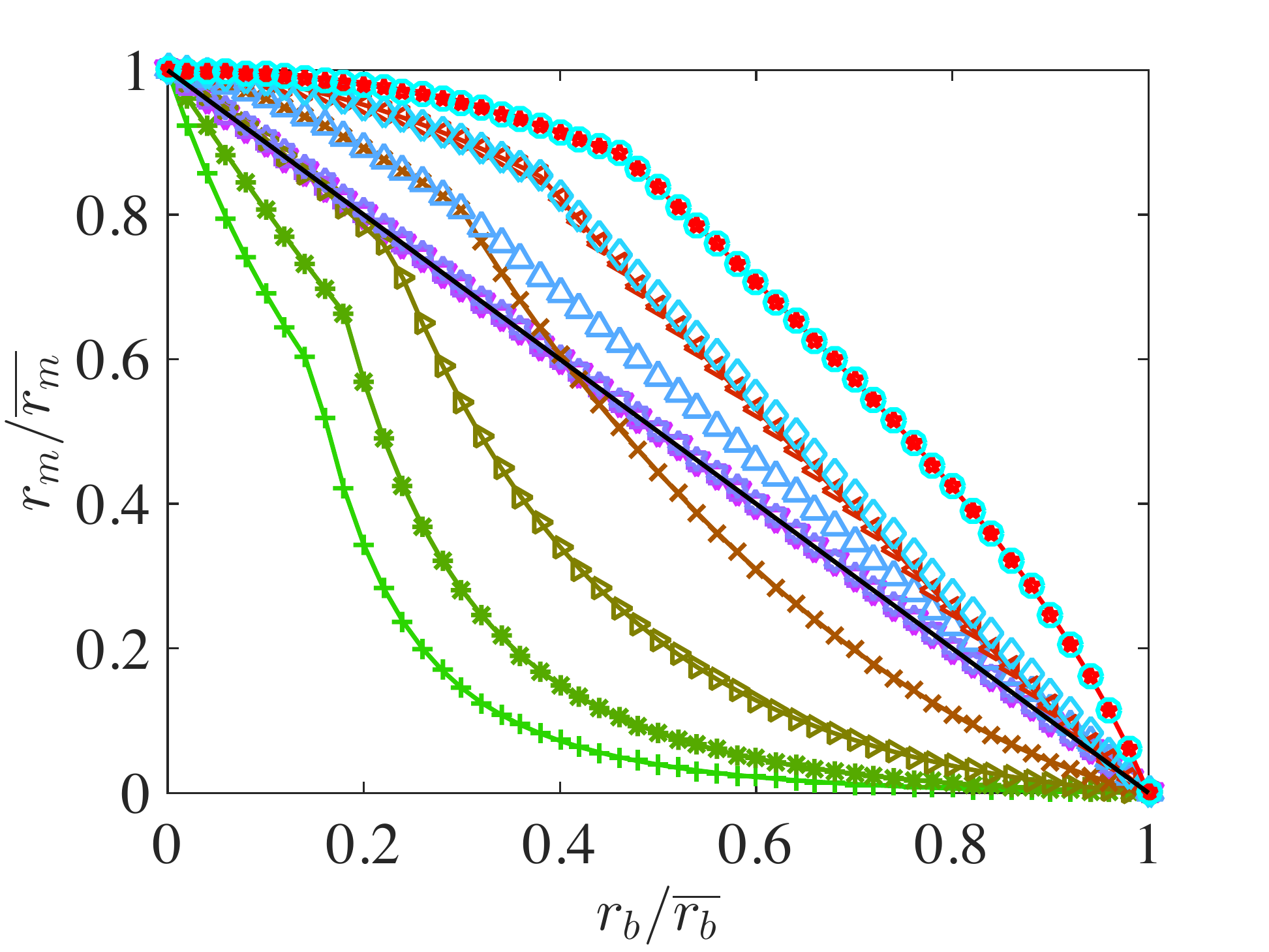}}\hspace{\fill}
\subfloat[]{\label{fig:cap_region_FDE_1}\includegraphics[scale = 0.22]{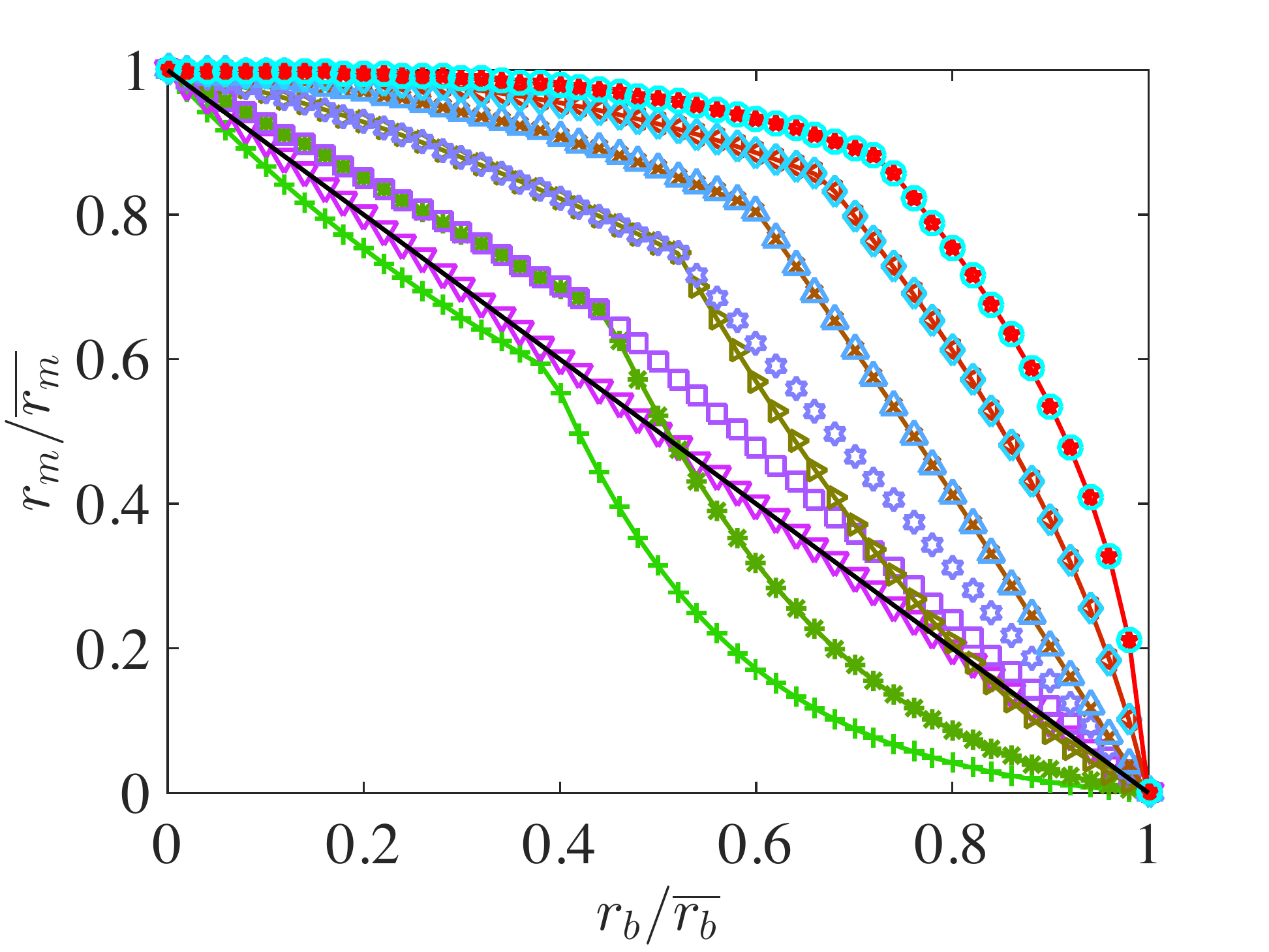}}\hspace{\fill}
\subfloat[]{\label{fig:cap_region_FDE_2}\includegraphics[scale = 0.22]{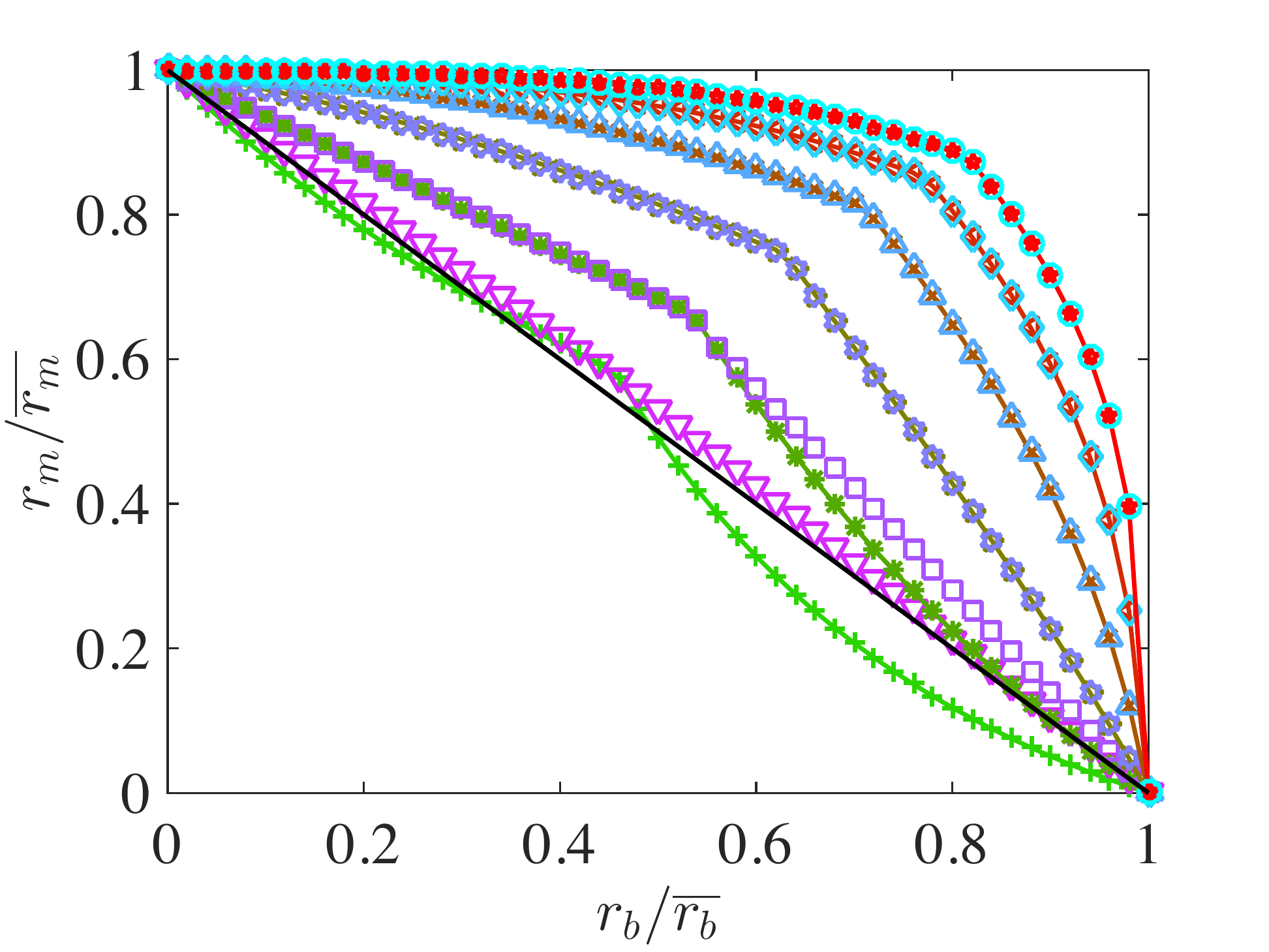}}\hspace{\fill}
\subfloat{\includegraphics[scale = 0.22]{cap_reg_legend.pdf}}\\\vspace{-10pt}
\setcounter{subfigure}{3}
\subfloat[]{\label{fig:asym_cap_region_conv}\includegraphics[scale = 0.22 ]{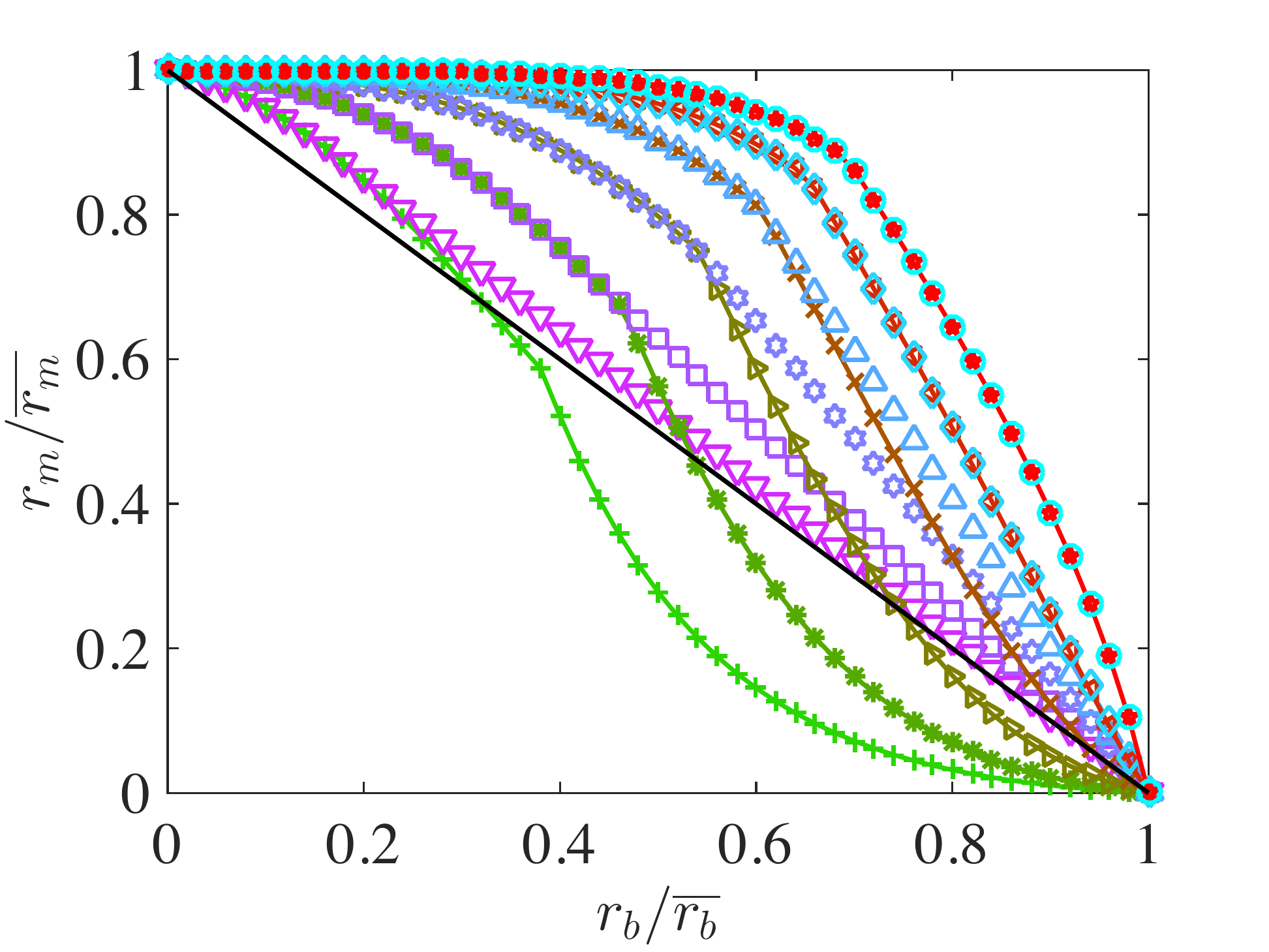}}\hspace{\fill}
\subfloat[]{\label{fig:asym_cap_region_FDE_1}\includegraphics[scale = 0.22]{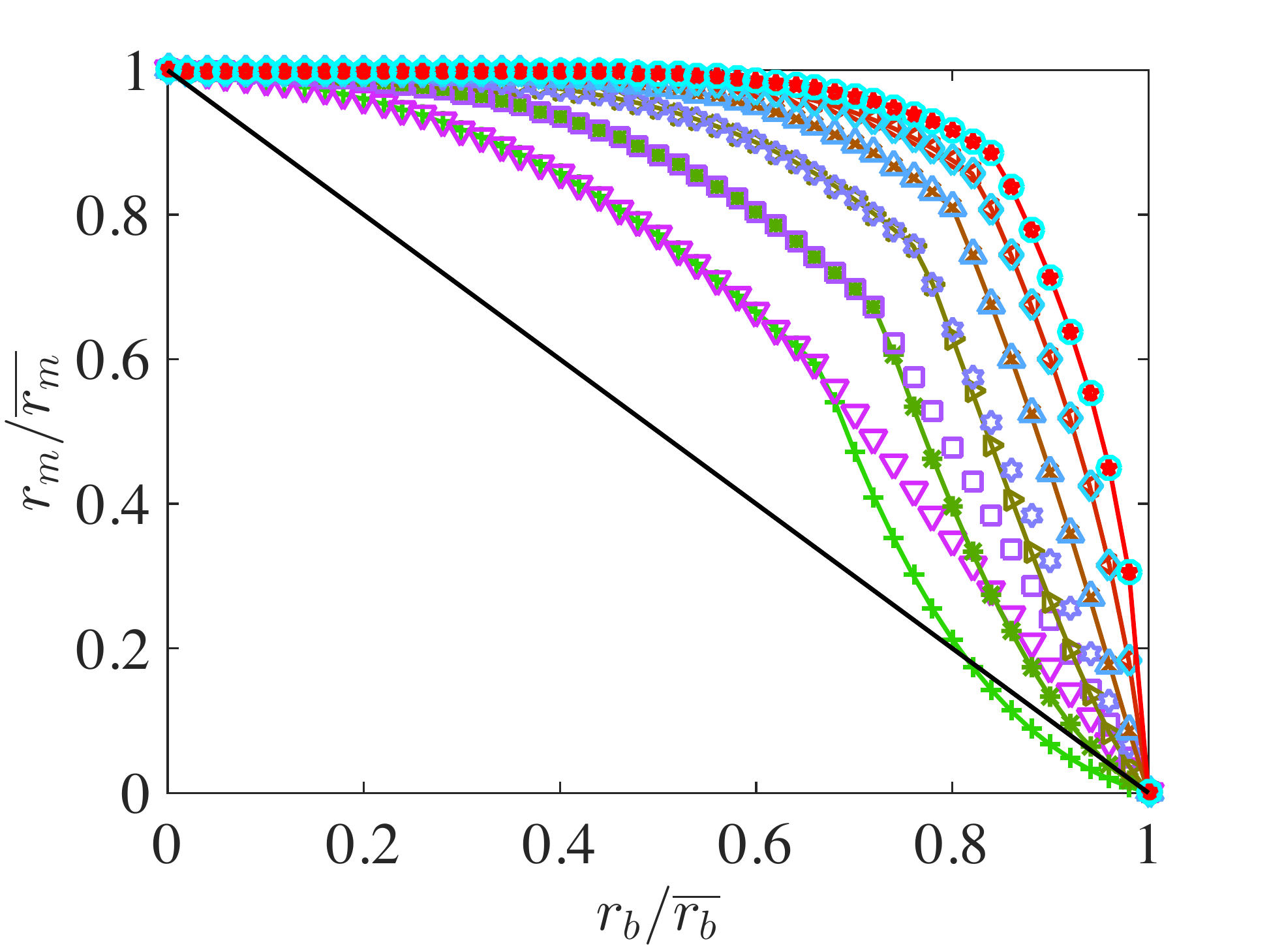}}\hspace{\fill}
\subfloat[]{\label{fig:asym_cap_region_FDE_2}\includegraphics[scale = 0.22]{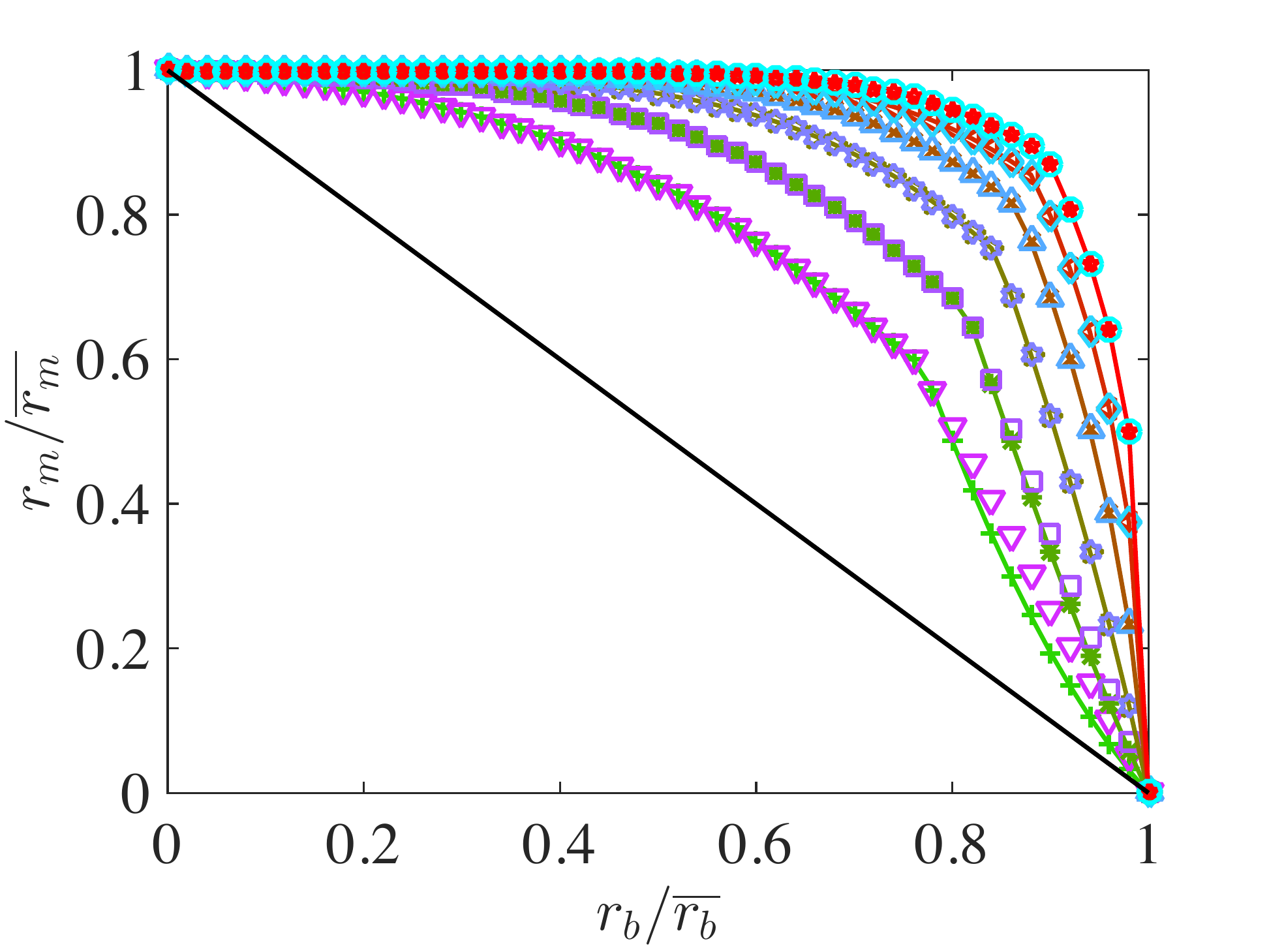}}\hspace{\fill}
\subfloat{\includegraphics[scale = 0.22]{asym_cap_reg_legend.pdf}}
\caption{Capacity regions for $\overline{\gamma_{bb, k}}$ from Fig.~\ref{fig:cancellation-profiles}\protect\subref{fig:gamma_bb}, and $\overline{\gamma_{mm, k}}$ from \protect\subref{fig:cap_region_conv}, \protect\subref{fig:asym_cap_region_conv} Fig.~\ref{fig:cancellation-profiles}\protect\subref{fig:gamma_mm_conv}, \protect\subref{fig:cap_region_FDE_1}, \protect\subref{fig:asym_cap_region_FDE_1} Fig.~\ref{fig:cancellation-profiles}\protect\subref{fig:gamma_mm_FDE_1}, and \protect\subref{fig:cap_region_FDE_2}, \protect\subref{fig:asym_cap_region_FDE_2} Fig.~\ref{fig:cancellation-profiles}\protect\subref{fig:gamma_mm_FDE_2}.}
\label{fig:cap-region-fd-tdd-multi-fixed} \vspace{-10pt}
\end{figure*}
\iffullpaper
\begin{figure*}[t!]
\center
\subfloat[]{\label{fig:rate_improve_conv}\includegraphics[scale = 0.22 ]{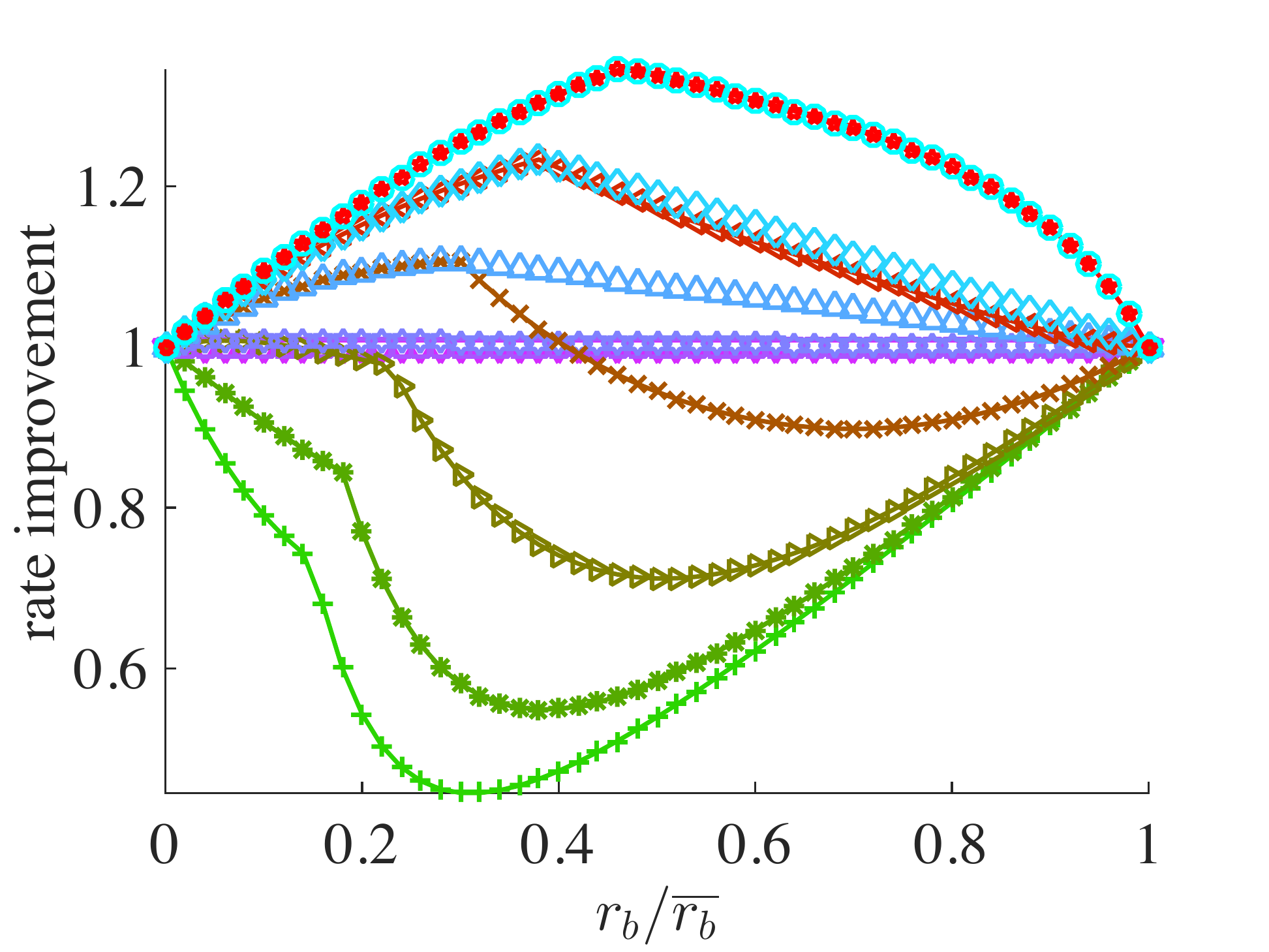}}\hspace{\fill}
\subfloat[]{\label{fig:rate_improve_FDE_1}\includegraphics[scale = 0.22]{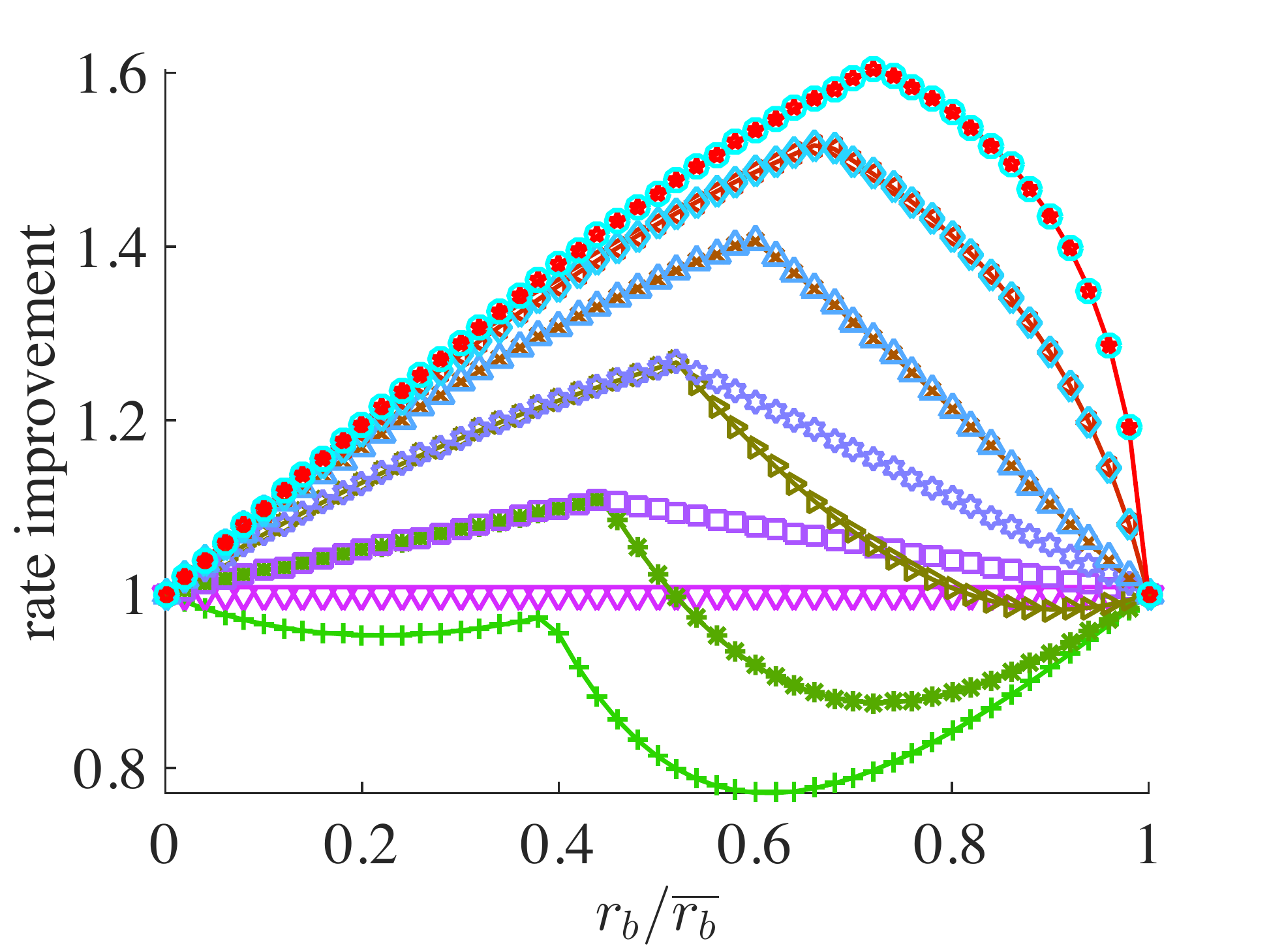}}\hspace{\fill}
\subfloat[]{\label{fig:rate_improve_FDE_2}\includegraphics[scale = 0.22]{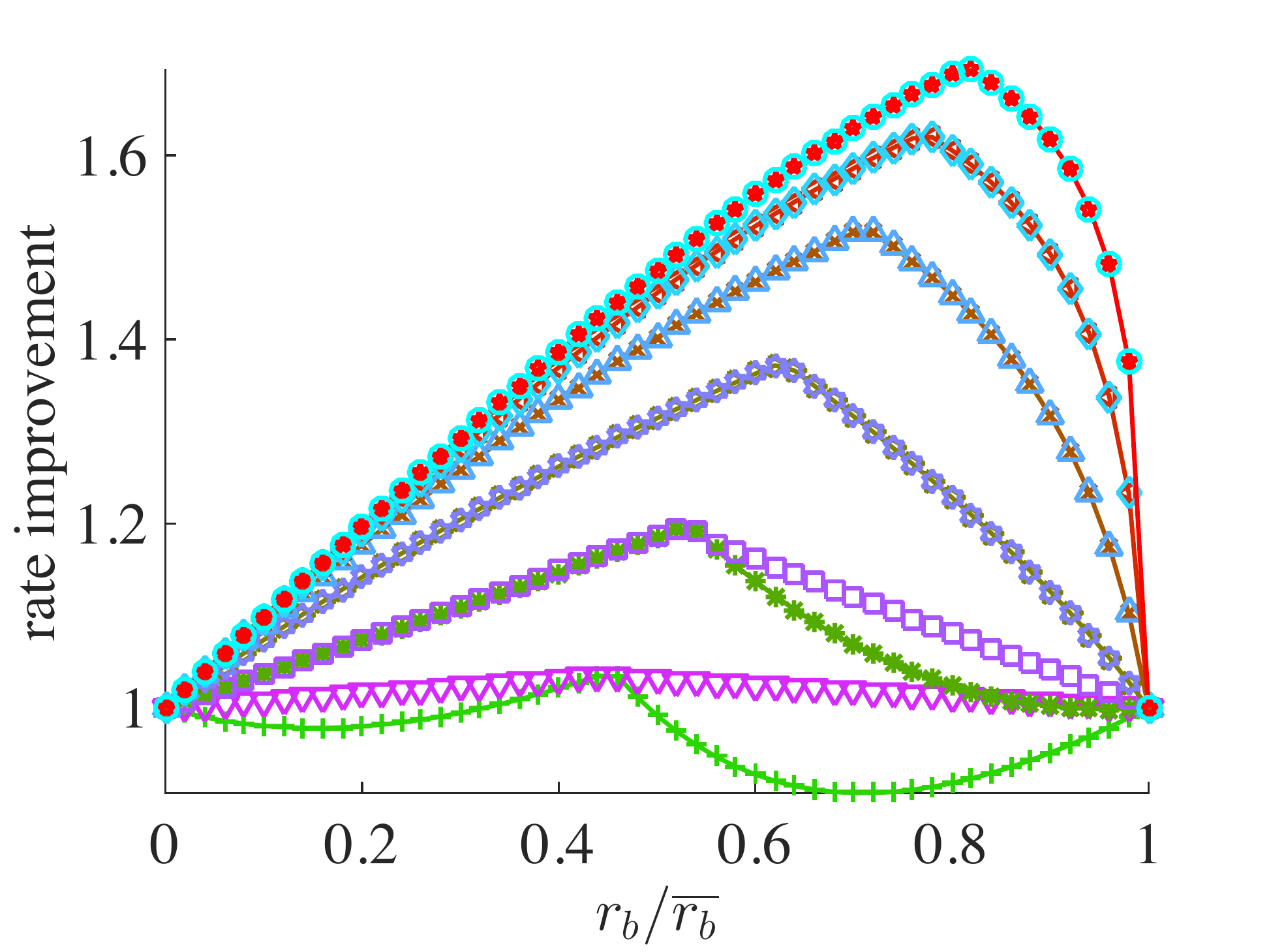}}\hspace{\fill}
\subfloat{\includegraphics[scale = 0.25]{rate_improve_legend.pdf}}\\\vspace{-10pt}
\setcounter{subfigure}{3}
\subfloat[]{\label{fig:asym_rate_improve_conv}\includegraphics[scale = 0.22 ]{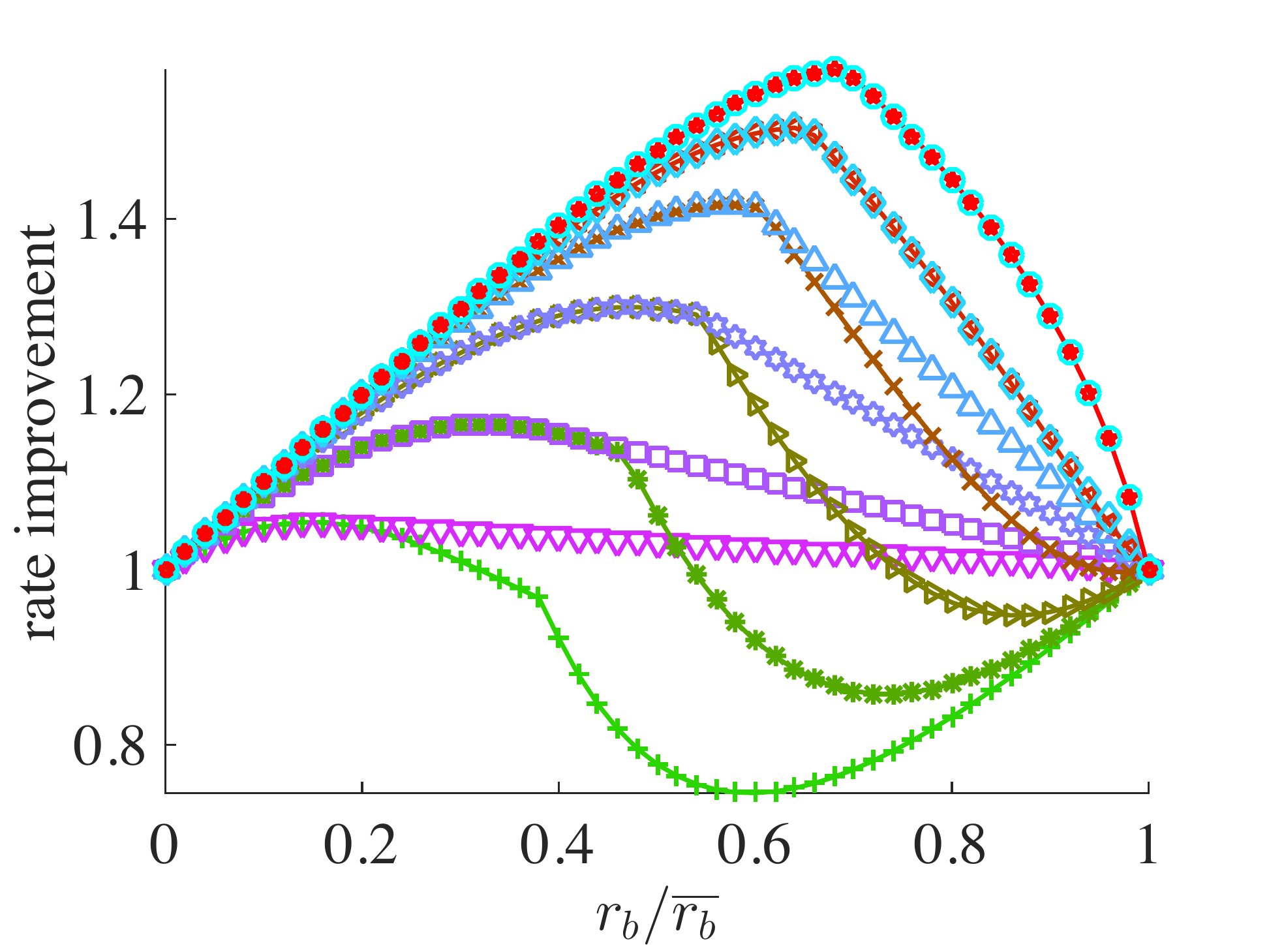}}\hspace{\fill}
\subfloat[]{\label{fig:asym_rate_improve_FDE_1}\includegraphics[scale = 0.22]{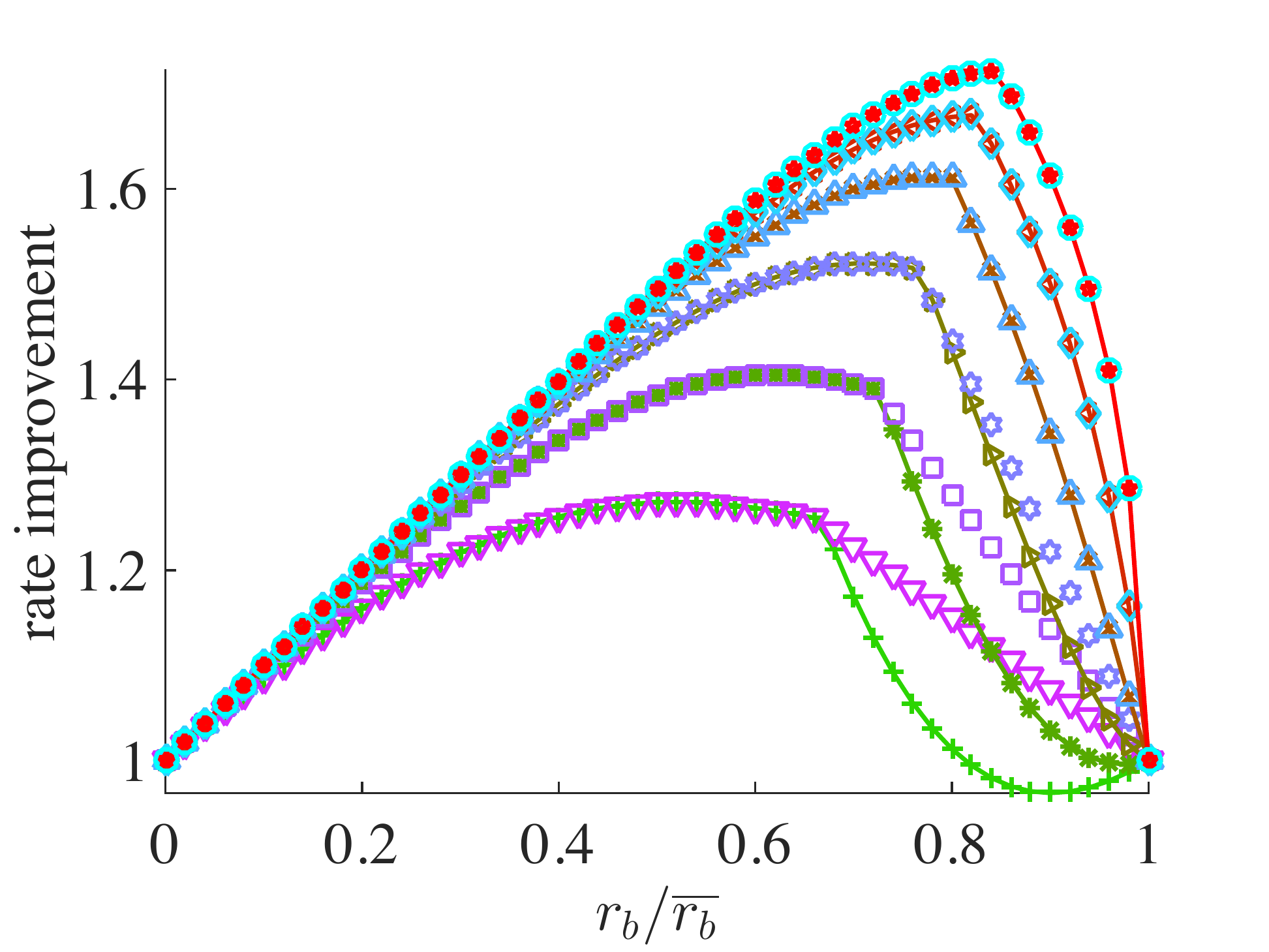}}\hspace{\fill}
\subfloat[]{\label{fig:asym_rate_improve_FDE_2}\includegraphics[scale = 0.22]{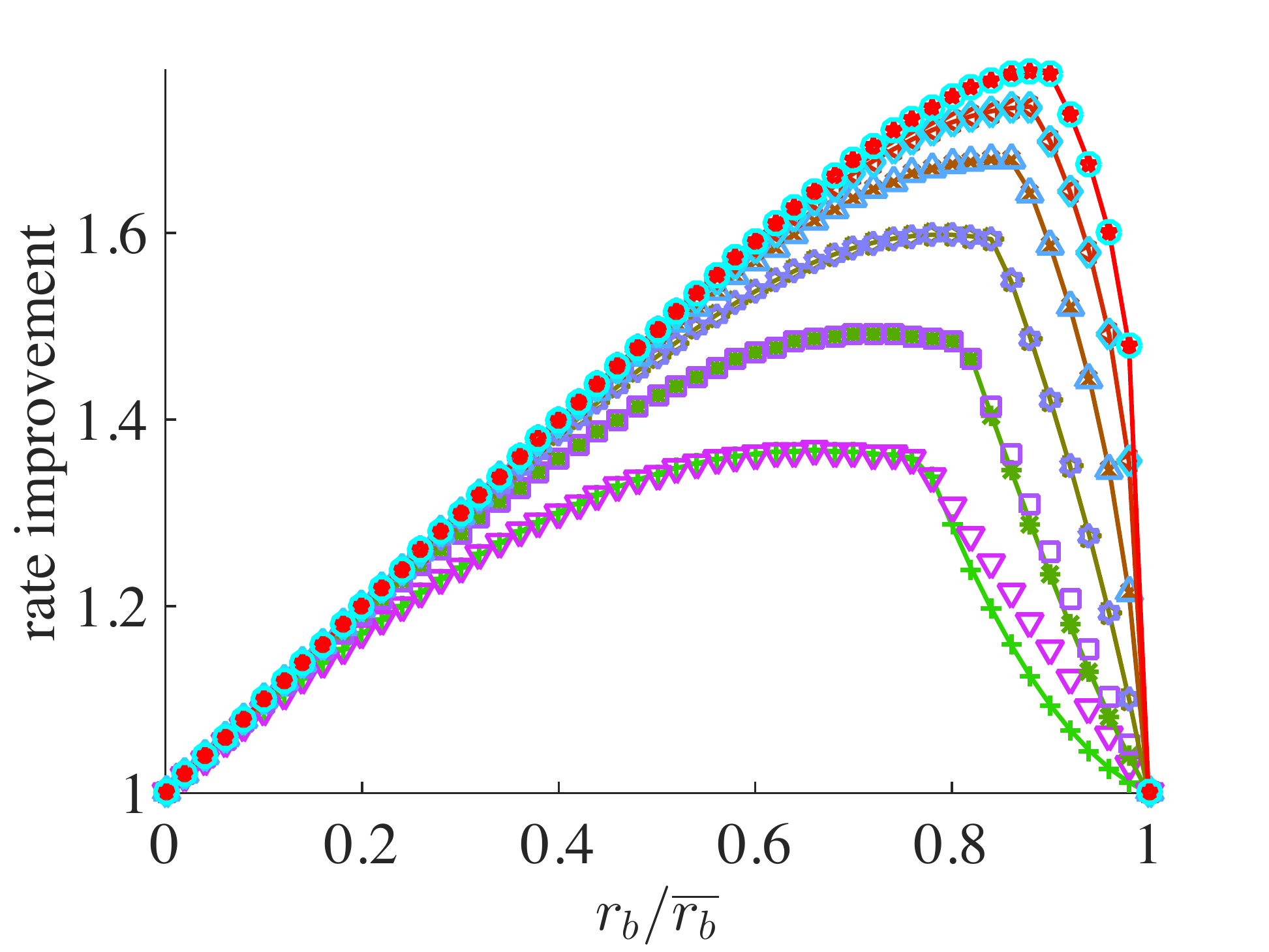}}\hspace{\fill}
\subfloat{\includegraphics[scale = 0.25]{asym_rate_improve_legend.pdf}}
\caption{Rate improvements for $\overline{\gamma_{bb, k}}$ from Fig.~\ref{fig:cancellation-profiles}\protect\subref{fig:gamma_bb}, and $\overline{\gamma_{mm, k}}$ from \protect\subref{fig:cap_region_conv}, \protect\subref{fig:asym_cap_region_conv} Fig.~\ref{fig:cancellation-profiles}\protect\subref{fig:gamma_mm_conv}, \protect\subref{fig:cap_region_FDE_1}, \protect\subref{fig:asym_cap_region_FDE_1} Fig.~\ref{fig:cancellation-profiles}\protect\subref{fig:gamma_mm_FDE_1}, and \protect\subref{fig:cap_region_FDE_2}, \protect\subref{fig:asym_cap_region_FDE_2} Fig.~\ref{fig:cancellation-profiles}\protect\subref{fig:gamma_mm_FDE_2}.}
\label{fig:rate_improve-fd-tdd-multi-fixed} \vspace{-10pt}
\end{figure*}
\fi

\section{Multi-Channel -- Fixed Power}\label{sec:multi}

In this section, we consider the problem of determining FD and TDFD capacity regions over multiple channels when the (shape of) the power allocation is fixed, but the total transmission power level can be varied. 
We first provide characterization of the FD capacity region, which allows computing any point on the FD capacity region via a binary search. 
Then, we turn to the problem of determining the TDFD capacity region. Due to the lack of structure as in the single channel case, in the multi-channel case the TDFD capacity region cannot in general be determined by  a binary search. We argue, however, that for inputs that are relevant in practice this problem can be solved in real time.

\subsection{Capacity Region}\label{section:fixed-PA}

Suppose that we want to determine the FD capacity region, given a fixed power allocation over $K$ orthogonal channels: $\alpha_{b, 1} = \alpha_{b, 2} =...=\alpha_{b, K} \equiv \alpha_b$ and $\alpha_{m, 1} = \alpha_{m, 2} =...=\alpha_{m, K} \equiv \alpha_m$. Note that setting the power allocation so that all $\alpha_{b, k}$'s and all $\alpha_{m, k}$'s are equal is without loss of generality, since we can represent an arbitrary fixed power allocation in this manner by appropriately scaling the values of $\overline{\gamma_{bm}}, \overline{\gamma_{mb}}, \overline{\gamma_{mm}}$, and $\overline{\gamma_{bb}}$ \iffullpaper (see Eq.'s (\ref{eq:rb-ofdm-equal}) and (\ref{eq:rm-ofdm-equal}) below). \fi The sum of the UL and DL rates over the (orthogonal) channels can then be written as 
$
r = r_b + r_m,
$
where:
\iffullpaper\begin{align}
r_b = \sum_{k=1}^K \log\Big(1+\frac{\alpha_b \overline{\gamma_{bm, k}}}{1+\alpha_m \overline{\gamma_{mm, k}}}\Big), \text{ and} \label{eq:rb-ofdm-equal}\\
r_m = \sum_{k=1}^K \log\Big(1+\frac{\alpha_m \overline{\gamma_{mb, k}}}{1+\alpha_b \overline{\gamma_{bb, k}}}\Big).\label{eq:rm-ofdm-equal}
\end{align} 
\fi
Let $s_b = r_b(\alpha_b = \frac{1}{K}, \alpha_m = \frac{1}{K})$, $s_m = r_m(\alpha_b = \frac{1}{K}, \alpha_m = \frac{1}{K})$. We characterize the FD capacity region in the following lemma\iffullpaper. \fi

\begin{lemma}\label{lemma:strictly-fd-ofdm-equal}
For a fixed $r_b = r_b^*\leq s_b$, $r_m$ is maximized for $\alpha_m = 1/K$. \iffullpaper Similarly, for a fixed $r_m = r_m^* \leq s_m$, $r_b$ is maximized for $\alpha_b = 1/K$.\fi
\end{lemma}
\iffullpaper
\begin{proof}
We will only prove the first part of the lemma, while the second part will follow using symmetric arguments.

Since $r_m$ is being maximized for a fixed $r_b = r_b^* \leq s_b$, we can think think of maximizing $r_m$ by only varying $\alpha_m$, while $\alpha_b$ changes as a function of $\alpha_m$ to keep $r_b = r_b^*$ as $\alpha_m$ is varied. Observe that for a fixed $\alpha_m \in [0, 1/K]$, $\alpha_b$ such that $r_b = r_b^*$ is uniquely defined since $r_b$ is monotonic in $\alpha_b$. Because $r_b^*\leq s_b$ and $r_b$ is decreasing in $\alpha_m$, a solution for $\alpha_b$ such that $r_b = r_b^*$ exists for any $\alpha_m \in [0, 1/K]$. It is not hard to see that $\alpha_b(\alpha_m)$ that keeps $r_b = r_b^*$ is a continuous and differentiable function. This follows from basic calculus, as $\alpha_b(\alpha_m)$ is an inverse function of $r_b$, $r_b$ is continuous and strictly increasing in $\alpha_b$, with $\frac{\partial r_b}{\partial \alpha_b}\neq 0$, $\forall (\alpha_b, \alpha_m) \in [0, 1]^2$. Therefore, we can write:
\begin{align}
\frac{d r_m(\alpha_m)}{d \alpha_m} = \frac{\partial r_m(\alpha_b, \alpha_m)}{\partial \alpha_m} + \frac{\partial r_m(\alpha_b, \alpha_m)}{\partial \alpha_b}\cdot\frac{d \alpha_b}{d \alpha_m}. \label{eq:drm-dalpham-ofdm-equal}
\end{align}
From (\ref{eq:rm-ofdm-equal}), we have:
\begin{align}
\frac{\partial r_m(\alpha_b, \alpha_m)}{\partial \alpha_m} = \sum_{k=1}^K \frac{\overline{\gamma_{mb, k}}}{1 + \alpha_m \overline{\gamma_{mb, k}} + \alpha_b \overline{\gamma_{bb, k}}}, \label{eq:partialrm-partialalpham-ofdm-equal}
\end{align}
and
\begin{align}
\frac{\partial r_m(\alpha_b, \alpha_m)}{\partial \alpha_b} = -\sum_{k=1}^K \frac{\frac{\overline{\gamma_{bb, k}}}{1 + \alpha_b \overline{\gamma_{bb, k}}}\cdot \alpha_m \overline{\gamma_{mb, k}}}{1 + \alpha_m \overline{\gamma_{mb, k}} + \alpha_b \overline{\gamma_{bb, k}}}.\label{eq:partialrm-partialalphab-ofdm-equal}
\end{align}
To find $\frac{d \alpha_b}{d \alpha_m}$, we will differentiate $r_b = r_b^*$ ($=\text{const.}$) w.r.t. $\alpha_m$, using (\ref{eq:rb-ofdm-equal}):
\begin{align}
\sum_{k=1}^K \frac{\overline{\gamma_{mm, k}} + \overline{\gamma_{bm, k}}\cdot \frac{d \alpha_b}{d \alpha_m}}{1 + \alpha_b \overline{\gamma_{bm, k}}+ \alpha_m \overline{\gamma_{mm, k}}} - \sum_{k=1}^K \frac{\overline{\gamma_{mm, k}}}{1 + \alpha_m \overline{\gamma_{mm, k}}} = 0 \notag 
\end{align}
\begin{align}
\Leftrightarrow \;  \frac{d \alpha_b}{d \alpha_m} =& \left(\sum_{k=1}^K \frac{\overline{\gamma_{bm, k}}}{1 + \alpha_b \overline{\gamma_{bm, k}}+ \alpha_m \overline{\gamma_{mm, k}}}\right)^{-1}\notag\\ 
&\cdot \sum_{k=1}^K \frac{\frac{\overline{\gamma_{mm, k}}}{1 + \alpha_m \overline{\gamma_{mm, k}}}\cdot \alpha_b \overline{\gamma_{bm, k}} }{1 + \alpha_b \overline{\gamma_{bm, k}}+ \alpha_m \overline{\gamma_{mm, k}}} \label{eq:dalphab-dalpham}\\
\leq & \alpha_b\cdot \max_{1\leq j\leq K} \frac{\overline{\gamma_{mm, j}}}{1 + \alpha_m \overline{\gamma_{mm, j}}}. \label{eq:ub-dalphab-dalpham}
\end{align}
Plugging (\ref{eq:partialrm-partialalpham-ofdm-equal}), (\ref{eq:partialrm-partialalphab-ofdm-equal}), and (\ref{eq:ub-dalphab-dalpham}) back into (\ref{eq:drm-dalpham-ofdm-equal}), we have:
\begin{align*}
\frac{d r_m(\alpha_m)}{d \alpha_m} &\geq \sum_{k=1}^K \frac{\overline{\gamma_{mb, k}}}{1 + \alpha_m \overline{\gamma_{mb, k}} + \alpha_b \overline{\gamma_{bb, k}}} \notag \\ 
&-\sum_{k=1}^K \frac{\frac{\overline{\gamma_{bb, k}}}{1 + \alpha_b \overline{\gamma_{bb, k}}}\cdot \alpha_m \overline{\gamma_{mb, k}}}{1 + \alpha_m \overline{\gamma_{mb, k}} + \alpha_b \overline{\gamma_{bb, k}}}\cdot \max_{1\leq j\leq K} \frac{\alpha_b\overline{\gamma_{mm, j}}}{1 + \alpha_m \overline{\gamma_{mm, j}}}\\
&= \sum_{k=1}^K \frac{\overline{\gamma_{mb, k}}}{1 + \alpha_m \overline{\gamma_{mb, k}} + \alpha_b \overline{\gamma_{bb, k}}} - \max_{1\leq j\leq K} \frac{\alpha_m\overline{\gamma_{mm, j}}}{1 + \alpha_m \overline{\gamma_{mm, j}}}  \notag \\  
&\cdot\sum_{k=1}^K \frac{\frac{\alpha_b\overline{\gamma_{bb, k}}}{1 + \alpha_b \overline{\gamma_{bb, k}}}\cdot \overline{\gamma_{mb, k}}}{1 + \alpha_m \overline{\gamma_{mb, k}} + \alpha_b \overline{\gamma_{bb, k}}}\\
&> 0, 
\end{align*}
where the last inequality follows from $\frac{\alpha_m\overline{\gamma_{mm, j}}}{1 + \alpha_m \overline{\gamma_{mm, j}}} < 1$ and $\frac{\alpha_b\overline{\gamma_{bb, k}}}{1 + \alpha_b \overline{\gamma_{bb, k}}}<1$, $\forall j, k$. It follows that $r_m$ is strictly increasing in $\alpha_m$, and, therefore, maximized for $\alpha_m = 1/K$.
\end{proof}

We now point out the difference between the proof of Lemma \ref{lemma:strictly-fd-ofdm-equal} and the proof of Theorem 3 in \cite{li2014rate}. The proof of Theorem 3 in \cite{li2014rate} uses similar arguments as the proof of Lemma \ref{lemma:strictly-fd-ofdm-equal} up to Eq. (\ref{eq:drm-dalpham-ofdm-equal}). However, the proof then concludes with the statement that $\frac{\partial r_m}{\partial \alpha_b}<0$ and $\frac{d \alpha_b}{d \alpha_m}<0$, which is not correct, as we see from (\ref{eq:dalphab-dalpham}) that $\frac{d \alpha_b}{d \alpha_m}>0$.\footnote{In a private communication, the authors of \cite{li2014rate} confirmed that our observation was correct and prepared an erratum.}

\fi
Using Lemma \ref{lemma:strictly-fd-ofdm-equal}, we can construct the entire FD capacity region by solving (i) $r_b = r_b^*$ for $\alpha_b$, when $\alpha_m = 1/K$ and $r_b^* \in [0, s_b]$, and (ii) $r_b = r_b^*$ for $\alpha_m$, when $\alpha_b = 1/K$ and $r_b^* \in (s_b, \overline{r_b}]$. Note that $r_b = r_b^*$ can be solved for $\alpha_b$ when $r_b \in [0, s_b]$ (resp.\ for $\alpha_m$) by using a binary search, since $r_b$ is monotonic and bounded in $\alpha_b$ for $r_b \in [0, s_b]$ (resp.\ $\alpha_m$ for $r_b \in (s_b, \overline{r_b}]$). The pseudocode is provided in Algorithm \ref{algo:multi-channel-fixed} (\textsc{MCFind-}$r_m$). The bound on the running time is provided in Proposition \ref{prop:mc-fixed-FD-algo-running-time}\iffullpaper.\fi
\begin{algorithm}

\caption{\textsc{MCFind-}$r_m$($r_b^*, K$)}
\begin{algorithmic}[1]
\Statex Input: $\overline{\gamma_{mb}}, \overline{\gamma_{bm}}, \overline{\gamma_{mm}}, \overline{\gamma_{bb}}$
\State $s_b = \sum_{k=1}^K \log(1 + \frac{1 + \overline{\gamma_{bm}}/K}{1+\overline{\gamma_{mm}}/K})$
\If {$r_b^* \leq s_b$}
\State Via binary search, find $\alpha_b$ s.t. $r_b(\alpha_b, 1/K) = r_b^*$
\State $r_m^* = \sum_{k=1}^K \log(1 + \frac{1 + \overline{\gamma_{mb}}/K}{1+\alpha_b\overline{\gamma_{bb}}})$
\Else
\State Via binary search, find $\alpha_m$ s.t. $r_b(1/K, \alpha_m) = r_b^*$
\State $r_m^* = \sum_{k=1}^K \log(1 + \frac{1 + \alpha_m\overline{\gamma_{mb}}}{1+\overline{\gamma_{bb}}/K})$
\EndIf
\Return $r_m^*$
\end{algorithmic}\label{algo:multi-channel-fixed}
\end{algorithm}

\begin{proposition}\label{prop:mc-fixed-FD-algo-running-time}
The running time of \textsc{MCFind-}$r_m$ is $O(K\log(\sum_{k}\frac{\overline{\gamma_{bb, k}}}{K\varepsilon}))$, where $\varepsilon$ is the additive error for $r_m^*$.
\end{proposition}
\iffullpaper
\begin{proof}
To determine $\alpha_b$ with the accuracy $\varepsilon_{\alpha}$, the binary search takes $\lceil\log({\varepsilon_{\alpha}}^{-1}/K)\rceil$ steps, as $\alpha_b \in [0, 1/K]$. From (\ref{eq:partialrm-partialalphab-ofdm-equal}), we can bound $|\frac{d r_m}{d \alpha_b}|$  as:
\begin{equation*}
\Big|\frac{d r_m}{d \alpha_b}\Big|\leq \sum_k \frac{\overline{\gamma_{bb, k}}}{1+\alpha_b\overline{\gamma_{bb, k}}}\leq \sum_k \overline{\gamma_{bb, k}},
\end{equation*}
as $\frac{\alpha_m \overline{\gamma_{mb, k}}}{1 + \alpha_m \overline{\gamma_{mb, k}}+\alpha_b\overline{\gamma_{bb, k}}}\leq 1$, and $1+\alpha_b\overline{\gamma_{bb, k}}\geq 1$, $\forall k$. Therefore, to find $r_m$ with the accuracy $\varepsilon$, it suffices to take $\varepsilon = \frac{\varepsilon_{\alpha}}{\sum_k \overline{\gamma_{bb, k}}}$. As each binary search step takes $O(K)$ computation (due to the computation of $r_b(\alpha_b, 1/K)$), we get the claimed running time bound.
\end{proof}
\fi
Notice that in practice $\overline{\gamma_{bb, k}}/K\leq 1$, $\overline{\gamma_{mm, k}}/K\leq 100$, and $K$ is at the order of 100, which makes the running time of \textsc{MCFind-}$r_m$ suitable for a real-time implementation.

Unlike in the single channel case, where the shape of the FD region boundary is very structured, in the multi-channel case the region does not necessarily have the property that $r_m(r_b)$ (and $r_b(r_m)$) has at most one concave and at most one convex piece\iffullpaper. \fi \iffullpaper To see why this holds, consider the following proposition.

\begin{proposition}\label{prop:fixed-multi-d2rm-drb2}
If $r_b\in[0, s_b]$, then $\frac{d^2 r_m}{d {r_b}^2} =\big(\frac{d r_b}{d \alpha_b}\big)^{-3} \frac{d^2 r_m}{d {\alpha_b}^2}\cdot \frac{d r_b}{d \alpha_b} - \frac{d r_m}{d {\alpha_b}}\cdot \frac{d^2 r_b}{d {\alpha_b}^2}$.
\end{proposition}
\begin{proof}
Fix $\alpha_m = 1/K$. As both $r_b(\alpha_b)$ and $\frac{d r_b}{d \alpha_b}$ are increasing and differentiable w.r.t. $\alpha_b$ and $\frac{d r_b}{d \alpha_b}\neq 0$, $\frac{d^2 r_b}{d {\alpha_b}^2}\neq 0$, $\forall \alpha_b\in[0, 1/K]$, it follows that $\alpha_b(r_b)$ is continuous and twice-differentiable w.r.t. $r_b$. Therefore, we can write:
\begin{align}
\frac{d^2 r_m}{d {r_b}^2} = \frac{d^2 r_m}{d {\alpha_b}^2}\cdot \left(\frac{d \alpha_b}{d r_b}\right)^2 + \frac{d r_m}{d \alpha_b} \cdot \frac{d^2 \alpha_b}{d {r_b}^2}. \label{eq:d2rm-drb2-ofdm-equal}
\end{align}
From (\ref{eq:rm-ofdm-equal}), we can determine $\frac{d r_m}{d \alpha_b}$ and $\frac{d^2 r_m}{d {\alpha_b}^2}$:
\begin{align}
\frac{d r_m}{d \alpha_b} = \sum_{k=1}^K \Big( \frac{\overline{\gamma_{bb, k}}}{1+\alpha_b \overline{\gamma_{bb, k}}+ \overline{\gamma_{mb, k}}/K} - \frac{\overline{\gamma_{bb, k}}}{1+\alpha_b \overline{\gamma_{bb, k}}}\Big), \label{drm-dalphab-ofdm-equal}
\end{align}
\begin{align}
\frac{d^2 r_m}{d {\alpha_b}^2} = \sum_{k=1}^K \bigg( \Big(\frac{\overline{\gamma_{bb, k}}}{1+\alpha_b \overline{\gamma_{bb, k}}}\Big)^2 - \Big(\frac{\overline{\gamma_{bb, k}}}{1+\alpha_b \overline{\gamma_{bb, k}}+ \overline{\gamma_{mb, k}}/K}\Big)^2\bigg). \label{d2rm-dalphab2-ofdm-equal}
\end{align}
To find $\frac{d \alpha_b}{d r_b}$ and $\frac{d^2 \alpha_b}{d {r_b}^2}$, we differentiate (\ref{eq:rb-ofdm-equal}) w.r.t. $r_b$. This gives:
\begin{align}
\frac{d \alpha_b}{d r_b} = \bigg(\sum_{k=1}^K \frac{\overline{\gamma_{bm, k}}}{1 + \alpha_b \overline{\gamma_{bm, k}} + \overline{\gamma_{mm, k}}/K}\bigg)^{-1} = \Big(\frac{d r_b}{d \alpha_b}\Big)^{-1}, \label{eq:dalphab-drb-ofdm-equal}
\end{align}
\begin{align}
\frac{d^2 \alpha_b}{d {r_b}^2} =& \left(\sum_{k=1}^K \frac{\overline{\gamma_{bm, k}}}{1 + \alpha_b \overline{\gamma_{bm, k}} + \overline{\gamma_{mm, k}}/K}\right)^{-3} \notag\\
&\cdot\sum_{k=1}^K \left(\frac{\overline{\gamma_{bm, k}}}{1 + \alpha_b \overline{\gamma_{bm, k}} + \overline{\gamma_{mm, k}}/K}\right)^2  \notag\\
&= -\left(\frac{d r_b}{d \alpha_b}\right)^{-3}\cdot \frac{d^2 r_b}{d {\alpha_b}^2}. \label{eq:d2alphab-drb2-ofdm-equal}
\end{align}
Plugging (\ref{eq:dalphab-drb-ofdm-equal}) and (\ref{eq:d2alphab-drb2}) back into (\ref{eq:d2rm-drb2}), we have:
\begin{align}
\frac{d^2 r_m}{d {r_b}^2} = \left(\frac{d r_b}{d \alpha_b}\right)^{-3} \left(\frac{d^2 r_m}{d {\alpha_b}^2}\cdot \frac{d r_b}{d \alpha_b} - \frac{d r_m}{d \alpha_b}\cdot \frac{d^2 r_b}{d {\alpha_b}^2} \right). 
\end{align}
\end{proof}
From Proposition \ref{prop:fixed-multi-d2rm-drb2}, as $\big(\frac{d r_b}{d \alpha_b}\big)^{-3}>0$, the sign of $\frac{d^2 r_m}{d {r_b}^2}$ is determined by the sign of $\frac{d^2 r_m}{d {\alpha_b}^2}\cdot \frac{d r_b}{d \alpha_b} - \frac{d r_m}{d \alpha_b}\cdot \frac{d^2 r_b}{d {\alpha_b}^2}$, which can be equivalently written as a rational function of $\alpha_b$ with linear-in-$K$ degree of the polynomial in its numerator. Therefore, the number of roots of $\frac{d^2 r_m}{d {r_b}^2}$ can be linear in $K$, and so $r_m$ can have up to linear in $K$ concave and convex pieces. When $K=1$, $\frac{d^2 r_m}{d {\alpha_b}^2}\cdot \frac{d r_b}{d \alpha_b} - \frac{d r_m}{d \alpha_b}\cdot \frac{d^2 r_b}{d {\alpha_b}^2}$ can be factored as:

\begin{align}
&\frac{\overline{\gamma_{bm}}}{1 + \alpha_b \overline{\gamma_{bm}} + \overline{\gamma_{mm}}} 
\cdot\Big(\frac{\overline{\gamma_{bb}}}{1+\alpha_b \overline{\gamma_{bb}}+ \overline{\gamma_{mb}}} - \frac{\overline{\gamma_{bb}}}{1+\alpha_b \overline{\gamma_{bb}}}\Big)\notag\\
&\cdot \Big(\frac{\overline{\gamma_{bb}}}{1+\alpha_b \overline{\gamma_{bb}}+ \overline{\gamma_{mb}}} + \frac{\overline{\gamma_{bb}}}{1+\alpha_b \overline{\gamma_{bb}}} - \frac{\overline{\gamma_{bm}}}{1 + \alpha_b \overline{\gamma_{bm}} + \overline{\gamma_{mm}}}\Big).\label{eq:}
\end{align}
Simplifying the rational expressions in $\Big(\frac{\overline{\gamma_{bb}}}{1+\alpha_b \overline{\gamma_{bb}}+ \overline{\gamma_{mb}}} + \frac{\overline{\gamma_{bb}}}{1+\alpha_b \overline{\gamma_{bb}}} - \frac{\overline{\gamma_{bm}}}{1 + \alpha_b \overline{\gamma_{bm}} + \overline{\gamma_{mm}}}\Big)$, we can recover the same quadratic function in the numerator as we had in (\ref{eq:quad-ineq-alphab}) and yield the same conclusions as in Lemma \ref{lemma:convexity-of-cap-region}, since $\frac{\overline{\gamma_{bm}}}{1 + \alpha_b \overline{\gamma_{bm}} + \overline{\gamma_{mm}}} 
\cdot\Big(\frac{\overline{\gamma_{bb}}}{1+\alpha_b \overline{\gamma_{bb}}+ \overline{\gamma_{mb}}} - \frac{\overline{\gamma_{bb}}}{1+\alpha_b \overline{\gamma_{bb}}}\Big)>0$. However, there does not seem to be a direct extension of this result to the $K>1$ case.
\fi

Although in general the problem of convexifying the FD region seems difficult in the multi-channel case, in practice 
it can be solved efficiently. The reason is that in Wi-Fi and cellular networks the output power levels take values from a discrete set of size $N$, where $N<100$. Therefore, (for fixed $\overline{\gamma_{mb, k}}$, $\overline{\gamma_{bm, k}}$, $\overline{\gamma_{bb, k}}$, $\overline{\gamma_{mm, k}}$, $\forall k$) $r_b$ can take at most $N$ distinct values. 
 To find the TDFD capacity region, since the points of the FD  region are determined in order increasing in $r_b$, $\Theta(N)$ computation suffices (Ch.\ 33, \cite{cormen2009introduction}).

The capacity regions and the rate improvements for $\overline{\gamma_{bb, k}}$ described by Fig.~\ref{fig:cancellation-profiles}\subref{fig:gamma_bb} and the three cases of $\overline{\gamma_{mm, k}}$ described by Fig.~\ref{fig:cancellation-profiles}\subref{fig:gamma_mm_conv}--\subref{fig:gamma_mm_FDE_2}, for equal power allocation and equal SNR over channels, are shown in Fig.~\ref{fig:cap-region-fd-tdd-multi-fixed} and \ref{fig:rate_improve-fd-tdd-multi-fixed}, respectively. As the cancellation becomes more broadband, namely as $\overline{\gamma_{mm, k}}$'s change from those described in Fig.~\ref{fig:cancellation-profiles}\subref{fig:gamma_mm_conv} over \ref{fig:cancellation-profiles}\subref{fig:gamma_mm_FDE_1} to \ref{fig:cancellation-profiles}\subref{fig:gamma_mm_FDE_2}, the rate improvements become higher and the capacity region becomes convex for lower values of $\overline{\gamma_{mb}}$ and $\overline{\gamma_{bm}}$.

\section{Multi-Channel -- General Power}\label{sec:multi-power}

We now consider the computation of TDFD capacity regions under general power allocations. In this case there are $2K$ variables ($\alpha_{b, 1},..., \alpha_{b, K}$, $\alpha_{m, 1},...,\alpha_{m, K}$), compared to 2 variables ($\alpha_b$ and $\alpha_m$) from the previous section. 

Computing $r_m^* = \max\{r_m:r_b = r_b^*\}$ is a non-convex problem, and  is hard to optimize in general. Yet, we present an algorithm that is guaranteed to converge to a stationary point, under certain restrictions. In practice, the stationary point to which it converges is also a global maximum. The restrictions are based on \cite{full-duplex-sigmetrics} and they guarantee that $r_b+r_m$ is concave when either the $\alpha_{b, k}$'s or $\alpha_{m, k}$'s are fixed. Note that the restrictions do not make the problem  $r_m^* = \max\{r_m:r_b = r_b^*\}$ convex (see Section \ref{section:gpa-cap-region}). 
The restrictions are mild in the sense that they do not affect the optimum by much whenever $\overline{\gamma_{bm, k}}$ and $\overline{\gamma_{mb, k}}$ do not differ much. 

Though for many practical cases the algorithm is near-optimal and runs in polynomial time, its running time in general is not suitable for a real-time implementation. To combat the high running time, in Section \ref{section:heuristic} we develop a simple heuristic that in most cases has similar performance.

\subsection{Capacity Region}\label{section:gpa-cap-region}

Determining the FD region under a general power allocation is equivalent to solving $\{\max r_m: r_b = r_b^*\}$ for any $r_b^* \in [0, \overline{r_b}]$ over $\alpha_{b, k}, \alpha_{m, k} \geq 0$, $\sum_{k}\alpha_{b, k}\leq 1$, $\sum_k \alpha_{m, k}\leq 1$. It is not hard to show that $\frac{d r_m}{d r_b}<0$, and, therefore, the problem is equivalent to $(P) = \{\max r_m: r_b \geq r_b^*\}$. 

Problem $(P)$ is not convex, even when some of the variables are fixed. When the $\alpha_{m, k}$'s are fixed, $r_b$ is concave in $\alpha_{b, k}$'s and the feasible region is convex, however, $r_m$ is convex as well. Conversely, when the $\alpha_{b, k}$'s are fixed, $r_m$ is concave in $\alpha_{m, k}$'s, but the feasible region is not convex since $r_b$ is convex in $\alpha_{m, k}$'s. Therefore, the natural approach to determining the FD region fails.

On the other hand, \cite{full-duplex-sigmetrics} provides conditions that guarantee that $\forall k$, $r = r_b + r_m$ is (i) concave and increasing in $\alpha_{m, k}$ when $\alpha_{b, k}$ is fixed, and (ii) concave and increasing in $\alpha_{b, k}$ when $\alpha_{m, k}$ is fixed. 
These conditions are not very restrictive: when they cannot be satisfied, one cannot gain much from FD additively -- the additive gain is less than 1b/s/Hz compared to the maximum of the UL and DL rates. However, these conditions can be very restrictive when the difference between $\overline{r_b}$ and $\overline{r_m}$ is high. The conditions are:
\begin{align}
\overline{\gamma_{bm, k}} \geq \overline{\gamma_{bb,k}}(1+\alpha_{m, k}\overline{\gamma_{mm, k}}), \forall k \tag{C1} \label{eq:C1}\\
\overline{\gamma_{mb, k}} \geq \overline{\gamma_{mm,k}}(1+\alpha_{b, k}\overline{\gamma_{bb, k}}), \forall k. \label{eq:C2}\tag{C2}
\end{align}
Notice that when $\overline{\gamma_{bm, k}} \geq \overline{\gamma_{bb,k}}(1+\overline{\gamma_{mm, k}})$ and $\overline{\gamma_{mb, k}} \geq \overline{\gamma_{mm,k}}(1+\overline{\gamma_{bb, k}})$, conditions (\ref{eq:C1}) and (\ref{eq:C2}) are non-restrictive (as they hold for any $\alpha_{b, k} \leq 1$, $\alpha_{mk}\leq 1$). When $\overline{\gamma_{bm, k}} < \overline{\gamma_{bb,k}}$, 
(\ref{eq:C1}) 
cannot be satisfied for any $\alpha_{m, k}$ 
as $\alpha_{m, k} \geq 0$. Similarly for $\overline{\gamma_{mb, k}} < \overline{\gamma_{mm,k}}$, (\ref{eq:C2}) cannot hold for any $\alpha_{b, k}$. 

We will use conditions (\ref{eq:C1}) and (\ref{eq:C2}) to formulate a new problem that is still non-convex, but more tractable than the original problem $(P)$. This way, we will get an upper bound on the capacity region and rate improvements when the conditions are non-restrictive. The new problem will also allow us to make a good estimate of the capacity region in the cases when $\overline{\gamma_{bm, k}}$ and $\overline{\gamma_{mb, k}}$ do not differ much. 

Let $(s_b, s_m)$ denote the UL-DL rate pair that maximizes the sum of the rates over UL and DL channels.
\begin{lemma}\label{lemma:motivation-for-q}
If conditions (\ref{eq:C1}) and (\ref{eq:C2}) are non-restrictive, then, given $\overline{\gamma_{bm, k}}, \overline{\gamma_{mb, k}}, \overline{\gamma_{mm, k}}, \overline{\gamma_{bb, k}}$ for $k\in\{1,...,K\}$, the TDFD capacity region can be determined by solving:
\begin{align*}
(Q) = \begin{cases}
\max &\sum_{k=1}^K (r_{b, k}(\alpha_{b, k}, \alpha_{m, k}) + r_{m, k}(\alpha_{b, k}, \alpha_{m, k}))\\
\mathrm{s.t.} & \sum_{k=1}^K r_{b, k}(\alpha_{b, k}, \alpha_{m, k})\; \mathrm{op} \; r_b^*\\
& \sum_{k=1}^K \alpha_{b, k} \leq 1, \sum_{k=1}^K \alpha_{m, k} \leq 1\\
& \alpha_{b, k} \geq 0, \alpha_{m, k} \geq 0, \forall k
\end{cases},
\end{align*}
where $\mathrm{op} = '\leq'$, if $r_b^* \leq s_b$ and $\mathrm{op} = '\geq'$, if $r_b^* \geq s_b$.
\end{lemma}
\iffullpaper\begin{proof}
First, observe that if we had $op = '='$, then $(Q)$ would be equivalent to $(P)$. Therefore, if an optimal solution to $(Q)$ satisfies $r_b = r_b^*$, then it also optimally solves $(P)$.

Suppose that $r_b^* \leq s_b$ and that an optimal solution $(r_b^Q, r_m^Q)$ to $(Q)$ satisfies $r_b^Q < r_b^*$. Let $(r_b^P, r_m^P)$, where $r_b^P = r_b^*$, be the optimal solution to $(P)$. Observe that $r_b^P + r_m^P \leq r_b^Q + r_m^Q$, and, as $r_b^Q < r_b^* = r_b^P$, it also holds that $r_m^Q > r_m^P$. Let $\lambda \in (0, 1)$ be the solution to $r_b^* = \lambda r_b^Q + (1-\lambda)s_b$ (such a $\lambda$ exists and is unique as $r_b < s_b$). Then, as $s_b+s_m \geq r_b^P + r_m^P$ and $r_b^Q + r_m^Q \geq r_b^P + r_m^P$, we have:
\begin{align*}
\lambda(r_b^Q + r_m^Q) + (1-\lambda)(s_b+s_m) &= r_b^P + \lambda r_m^Q + (1-\lambda)s_m\\
&\geq r_b^P + r_m^P,
\end{align*}
and we have $\lambda r_m^Q + (1-\lambda)s_m\geq r_m^P$. Therefore, we can get a point $(r_b^*, r_m)$ with $r_m \geq r_m^P$ as a convex combination of the points that optimally solve both $(P)$ and $(Q)$. In other words, the convex hull of the points determined by $(Q)$ is the TDFD capacity region. To find the the convex hull of the points determined by $(Q)$, we can employ an algorithm for finding a convex hull of given points from e.g., \cite{cormen2009introduction}.

A similar argument follows for $r_b^* > s_b$.
\end{proof}
\fi

When conditions (\ref{eq:C1}) and (\ref{eq:C2}) are restrictive, they provide upper bounds on $\alpha_{b, k}$ and $\alpha_{m, k}$ and they do not affect the optimal solution to $(Q)$ unless $\overline{\gamma_{bm, k}}>>\overline{\gamma_{mb, k}}$ or $\overline{\gamma_{mb, k}}>>\overline{\gamma_{bm, k}}$ for some $k$. To avoid infeasibility when restricting the feasible region of $(Q)$ by requiring (\ref{eq:C1}) and (\ref{eq:C2}), similar to \cite{full-duplex-sigmetrics}, we will set either $\alpha_{b, k} = 0$ or $\alpha_{m, k} = 0$.\footnote{Recall that when $\alpha_{b, k} = 0$, the sum of the rates is concave in $\alpha_{m, k}$ for any $\alpha_{m, k} \in [0, 1]$. Similarly when $\alpha_{b, k} = 0$.}

We write the restrictions imposed by (\ref{eq:C1}) and (\ref{eq:C2}) on the feasible region of $(Q)$ as follows, where $\alpha_{b, k}\leq A_b(k)$ and $\alpha_{m, k}\leq A_m(k)$, $\forall k$. Notice that the restrictions are fixed for fixed $\overline{\gamma_{bm, k}}, \overline{\gamma_{mb, k}}, \overline{\gamma_{mm, k}}, \overline{\gamma_{bb, k}}$, and $r_b^*$. We refer to the restricted version of problem $(Q)$ as $(Q_R)$.
\begin{algorithm}

\begin{algorithmic}
\State Let $A_b$ and $A_m$ be size-$K$ arrays
\For{k = 1 \textbf{to} K}
\State $A_b(k) = \frac{\overline{\gamma_{mb, k}}/\overline{\gamma_{mm, k}}-1}{\overline{\gamma_{bb, k}}}$, $A_m(k) = \frac{\overline{\gamma_{bm, k}}/\overline{\gamma_{bb, k}}-1}{\overline{\gamma_{mm, k}}}$
\If{$r_b^* \leq s_b$}
\If {$A_{b}(k) \leq 0$}
 $A_{b}(k) = 0$, $A_{m}(k) = 1$
\EndIf
\If {$A_{m}(k) \leq 0$}
 $A_{m}(k) = 0$, $A_{b}(k) = 1$
\EndIf
\Else
\If {$A_{m}(k) \leq 0$}
 $A_{m}(k) = 0$, $A_{b}(k) = 1$
\EndIf
\If {$A_{b}(k) \leq 0$}
 $A_{b}(k) = 0$, $A_{m}(k) = 1$
\EndIf
\EndIf
\EndFor
\end{algorithmic}
\end{algorithm}

To solve $(Q_R)$, we will use a well-known practical method called alternating minimization (or maximization, as in our case) \cite{ortega1970iterative}. For a given problem $(P_i)$, the method partitions the variable set $x$ into two sets $x_1$ and $x_2$, and then iteratively applies the following procedure: (i) optimize $(P_i)$ over $x_1$ by treating the variables from $x_2$ as constants, (ii) optimize $(P_i)$ over $x_2$ by treating the variables from $x_1$ as constants, until a stopping criterion is reached. 

\begin{algorithm}

\caption{\textsc{AltMax}($(Q_R), \varepsilon$)}
\begin{algorithmic}[1]
\State Let $\{\alpha_{b, k}^0\}, \{\alpha_{m ,k}^0\}$ be a feasible solution to $(Q_R)$, $n = 0$
\Repeat
\State $n = n+1$
\State $\{\alpha_{b, k}^n\} = \arg\max \{(Q_{R, b}): \{\alpha_{m ,k}^n\} = \{\{\alpha_{m ,k}^{n-1}\}\}$
\State $\{\alpha_{m, k}^n\} = \arg\max \{(Q_{R, m}): \{\alpha_{b ,k}^n\} = \{\alpha_{b ,k}^{n-1}\}\}$
\Until{$\max_k \{|\alpha_{b, k}^n - \alpha_{b, k}^{n-1}| + |\alpha_{m, k}^n - \alpha_{m, k}^{n-1}|\} < \varepsilon$}
\end{algorithmic}\label{algo:alpt-max}
\end{algorithm}
Even in the cases when $(P_i)$ is non-convex, if subproblems from (i) and (ii) have unique solutions and are solved optimally in each iteration, the method converges to a stationary point with rate $O(1/\sqrt{n})$, where $n$ is the iteration count \cite{beck2015convergence}. In the cases when, in addition, for each of the subproblems the objective is convex (concave for maximization problems), for each stationary point there exists an initial point such that the alternating minimization converges to that stationary point \cite{gorski2007biconvex}. A common approach that works well in practice is to generate many random initial points  and choose the best solution found. In our experiments, choosing $\alpha_{b, k} = \alpha_{m, k} = 0$ as the initial point typically led to the best solution.

Due to the added restrictions in problem $(Q_R)$ imposed by (\ref{eq:C1}) and (\ref{eq:C2}), the objective in $(Q_R)$ is concave whenever either all $\alpha_{b, k}$'s or all $\alpha_{m, k}$'s are fixed, while the remaining variables are varied. Hence, our two subproblems for $Q_R$ will be: (i) $(Q_{R, b})$, which is equivalent to $(Q_R)$ except that it treats $\alpha_{b, k}$'s as variables and $\alpha_{m, k}$'s as constants, and (ii) $(Q_{R, m})$, which is equivalent to $(Q_R)$ except that it treats $\alpha_{m, k}$'s as variables and $\alpha_{b, k}$'s as constants. Given accuracy $\varepsilon$, the pseudocode is provided in Algorithm \ref{algo:alpt-max} (\textsc{AltMax}). The rate pair $(s_b, s_m)$ can be determined using the same algorithm by omitting the constraint $r_b\leq r_b^*$ (or $r_b\geq r_b^*$).

What remains to show is that both $(Q_{R, b})$ and $(Q_{R, m})$ have unique solutions that can be found in polynomial time. We do that in the following (constructive) lemma. Note that without the constraint $r_b^*\leq s_b$ or $r_b^* \geq s_b$, both $(Q_{R, b})$ and $(Q_{R, m})$ are convex and have strictly concave objectives, and therefore, we can determine $s_b$ using \textsc{AltMax}. 

\begin{lemma}
Starting with a feasible solution $\{\alpha_{b, k}^0, \alpha_{m ,k}^0\}$ to $(Q_R)$,
in each iteration of \textsc{AltMax} the solutions to $(Q_{R, b})$ and $(Q_{R, m})$ are unique and can be found in polynomial time.
\end{lemma}
\begin{proof}
Suppose that $r_b^* \leq s_b$. Then it is not hard to verify that $(Q_{R, m})$ is a convex problem with a strictly concave objective. The objective is strictly concave due to the enforcement of conditions (\ref{eq:C1}) and (\ref{eq:C2}), while all the constraints except for $r_b \leq r_b^*$ are linear. The constraint $r_b \leq r_b^*$ is convex as $r_b$ is convex in $\alpha_{m, k}$'s. Therefore, $(Q_{R, m})$ admits a unique solution that can be found in polynomial time through convex programming. By similar arguments, when $r_b^* > s_b$, $(Q_{R, b})$ admits a unique solution that can be found in polynomial time through convex programming.

Consider $(Q_{R, b})$ when $r_b^* \leq s_b$. This problem is not convex due to the constraint $r_b \leq r_b^*$, as $r_b$ is concave in $\alpha_{b, k}$'s. However, we will show that the problem has enough structure so that it is solvable in polynomial time. 

Let $k^* = \arg\max_k\big\{\frac{\gamma_{bm, k}}{1 + \alpha_{m, k}\overline{\gamma_{mm, k}}} - \gamma_{bb, k} + \frac{\gamma_{bb, k}}{1+\alpha_{m, k}\overline{\gamma_{mb, k}}}\big\}$ ($= \arg\max_k\big\{\frac{d r}{d \alpha_{b, k}}\big|_{\alpha_{b, k} = 0}\big\}$). 
Recall that, due to conditions (\ref{eq:C1}) and (\ref{eq:C2}), we have that $\frac{d^2 r}{d {\alpha_{b, k}}^2} < 0$, and therefore $\frac{d r}{d \alpha_{b, k}}$ is monotonically decreasing, $\forall k$. It follows that for any $\alpha_{b, k^*}\in[0, 1]$ and any $k \in \{1,...,K\}$, either there exists a (unique) $\alpha_{b, k}\in [0, 1]$ such that $\frac{d r}{d \alpha_{b, k}} = \frac{d r}{d \alpha_{b, k^*}}$, or $\frac{d r}{d \alpha_{b, k}} < \frac{d r}{d \alpha_{b, k^*}}$, $\forall \alpha_{b, k}\in [0, 1]$.

Consider Algorithm \ref{alg-sol} (\textsc{SolveSubproblem}b) and let $\{\alpha_{b, k}^*\}$ be the solution returned by the algorithm. Note that the binary search for finding $\alpha_{b, k^*}^*$ and for determining $\alpha_{b, k}^*$'s in \textsc{SolveSubproblem}b is correct from the choice of $k^*$ and because $\frac{d r}{d \alpha_{b, k}}$ is monotonically decreasing, $\forall k$. 
\begin{algorithm}

\caption{\textsc{SolveSubproblem}b}\label{alg-sol}
\begin{algorithmic}[1]
\State $k^* = \arg\max_k\big\{\frac{\gamma_{bm, k}}{1 + \alpha_{m, k}\overline{\gamma_{mm, k}}} - \gamma_{bb, k} + \frac{\gamma_{bb, k}}{1+\alpha_{m, k}\overline{\gamma_{mb, k}}}\big\}$
\State For $\alpha_{b, k^*}\in [0, 1]$, via binary search, find the maximum $\alpha_{b, k^*}$ such that $r+b \leq r_b^*$ and $\sum_k \alpha_{b, k} \leq 1$, where:
\If{$\frac{d r}{d \alpha_{b, k}}\big|_{\alpha_{b, k} = 0} < \frac{d r}{d \alpha_{b, k^*}}$}
 $\alpha_{b, k} = 0$
\Else 
\State {Via binary search over $\alpha_{b, k}\in [0, 1]$, find $\alpha_{b, k}$ such that $\frac{d r}{d \alpha_{b, k}} = \frac{d r}{d \alpha_{b, k^*}}$}
\EndIf
\end{algorithmic}
\end{algorithm}

We first show that $\{\alpha_{b, k}^*\}$ is a local maximum for $(Q_b)$. Because of the algorithm's termination conditions, it must be either $\sum_k \alpha_{b, k}^* = 1$ or $r_b = r_b^*$. If $\sum_k \alpha_{b, k}^* = 1$, then to move to any alternative solution, the total change must be $\sum_k \Delta \alpha_{b, k} \leq 0$, or, equivalently $\Delta \alpha_{b, k^*}\leq -\sum_{k\neq k^*}\Delta \alpha_{b, k}$. As $\frac{d r}{d \alpha_{b, k}} \leq \frac{d r}{d \alpha_{b, k^*}}$, it follows that $\sum_k \frac{d r}{d \alpha_{b, k}}  \Delta \alpha_{b, k}\leq 0$, which is the first-order optimality condition. Now suppose that $r_b = r_b^*$. Since $\frac{d r}{d \alpha_{b, k}} = \frac{d r_b}{d \alpha_{b,k}} + \frac{d r_m}{d \alpha_{b,k}} > 0$ and $\frac{d r_b}{d \alpha_{b,k}} > 0$, $\frac{d r_m}{d \alpha_{b,k}}<0$, to keep the solution feasible (i.e., to keep $r_b \leq r_b^*$), we must have $\sum_k \frac{d r_{b, k}}{d \alpha_{b,k}} \Delta \alpha_{b, k}\leq 0$, which implies $\sum_k \frac{d r}{d \alpha_{b, k}} \Delta \alpha_{b, k}\leq 0$. Therefore, $\{\alpha_{b, k}^*\}$ computed by \textsc{SolveSubproblem}b is a local optimum.

In fact, for any local optimum: $\frac{d r}{d \alpha_{b, k}} \leq \frac{d r}{d \alpha_{b, k^*}}$, otherwise we can construct a better solution. Suppose that $\frac{d r}{d \alpha_{b, j}} > \frac{d r}{d \alpha_{b, k^*}}$ for some $j$. Then if $\frac{d r_b}{d \alpha_{b, j}} \leq \frac{d r_b}{d \alpha_{b, k^*}}$, we can choose a sufficiently small $\Delta > 0$, so that the solution $\{\alpha_{b, k}'\}$ with $\alpha_{b, j}' = \alpha_{b, j}+\Delta$, $\alpha_{b, k^*}' = \alpha_{b, k^*}-\Delta$, and $\alpha_{b, k}' = \alpha_{b, k}$ for $k\notin\{j, k^*\}$ is feasible. For such a solution $\sum_k \frac{d r}{d \alpha_{b, k}} (\alpha_{b, k}'- \alpha_{b, k}) > 0$, and therefore, it is not a local optimum. Conversely, if $\frac{d r_b}{d \alpha_{b, j}} > \frac{d r_b}{d \alpha_{b, k^*}}$, we can choose sufficiently small $\Delta_1, \Delta_2 > 0$ such that  $\Delta_2 > \Delta_1$ and $\frac{d r}{d \alpha_{b, j}}\Delta_1 > \frac{d r}{d \alpha_{b, k^*}}\Delta_2$. Then, we can construct an $\{\alpha_{b, k}'\}$ with $\alpha_{b, j}' = \alpha_{b, j}+\Delta_1$, $\alpha_{b, k^*}' = \alpha_{b, k^*}-\Delta_2$, and $\alpha_{b, k}' = \alpha_{b, k}$ for $k\notin\{j, k^*\}$ that is feasible. Again, we have $\sum_k \frac{d r}{d \alpha_{b, k}}  (\alpha_{b, k}'- \alpha_{b, k}) > 0$, and $\{\alpha_{b, k}\}$ cannot be a local maximum.

Finally, since $\{\alpha_{b, k}^*\}$ returned by \textsc{SolveSubproblem}b satisfies $\alpha_{b, k}^*\geq \alpha_{b, k}'$ for any other local maximum $\{\alpha_{b, k}'\}$ and the objective is strictly increasing in all $\alpha_{b, k}$'s, $\{\alpha_{b, k}^*\}$ must be a global maximum. From the strict monotonicity of $\frac{d r}{d \alpha_{b, k}}$, this maximum is unique. 
The proof for $(Q_{R, m})$ when $r_b^* \geq s_b$ uses similar arguments and is omitted.
\end{proof}

\begin{figure*}[t!]
\center
\subfloat[]{\label{fig:cmph_rate_improve_conv}\includegraphics[scale = 0.22 ]{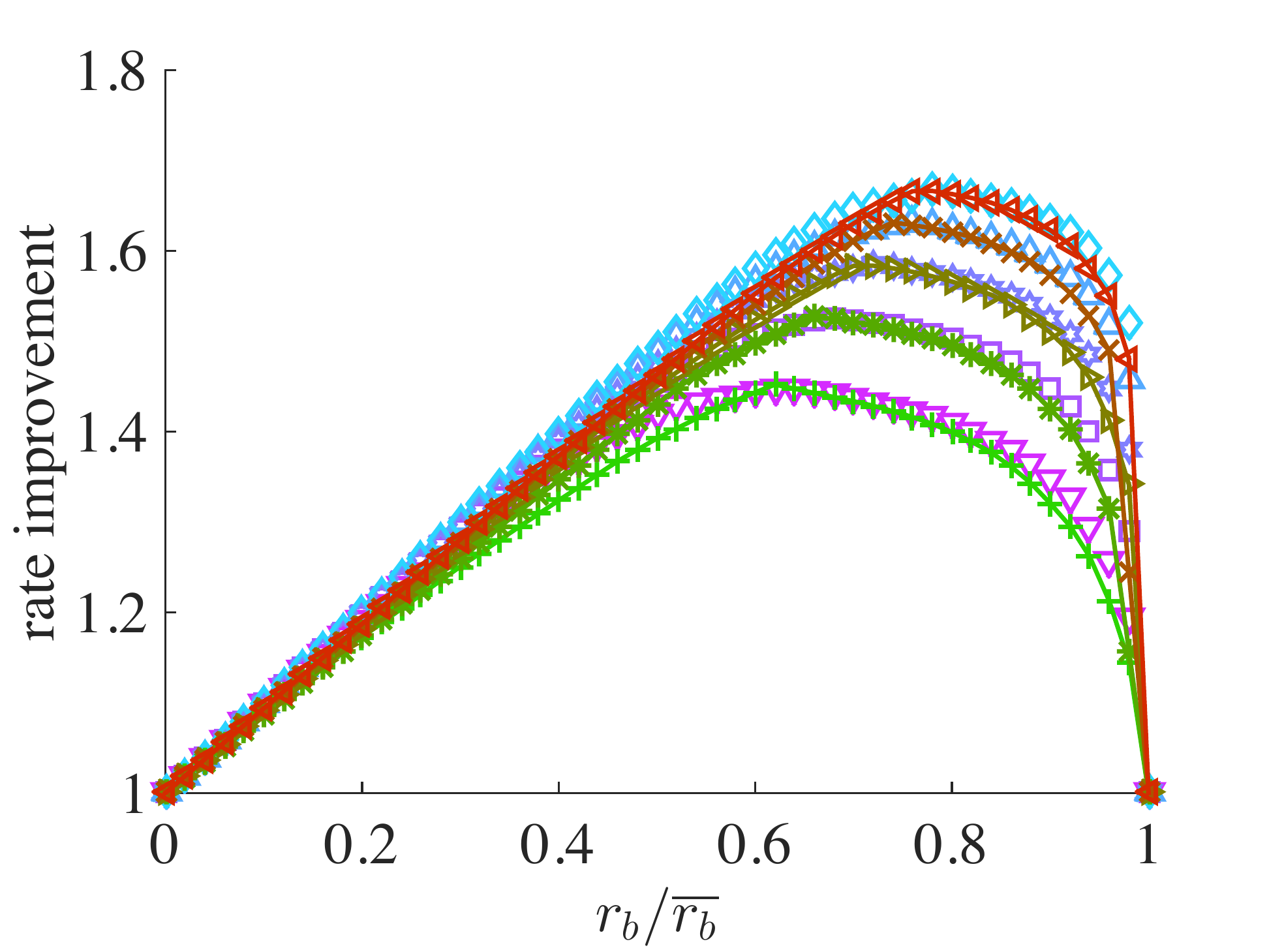}}\hspace{\fill}
\subfloat[]{\label{fig:cmph_rate_improve_FDE_1}\includegraphics[scale = 0.22]{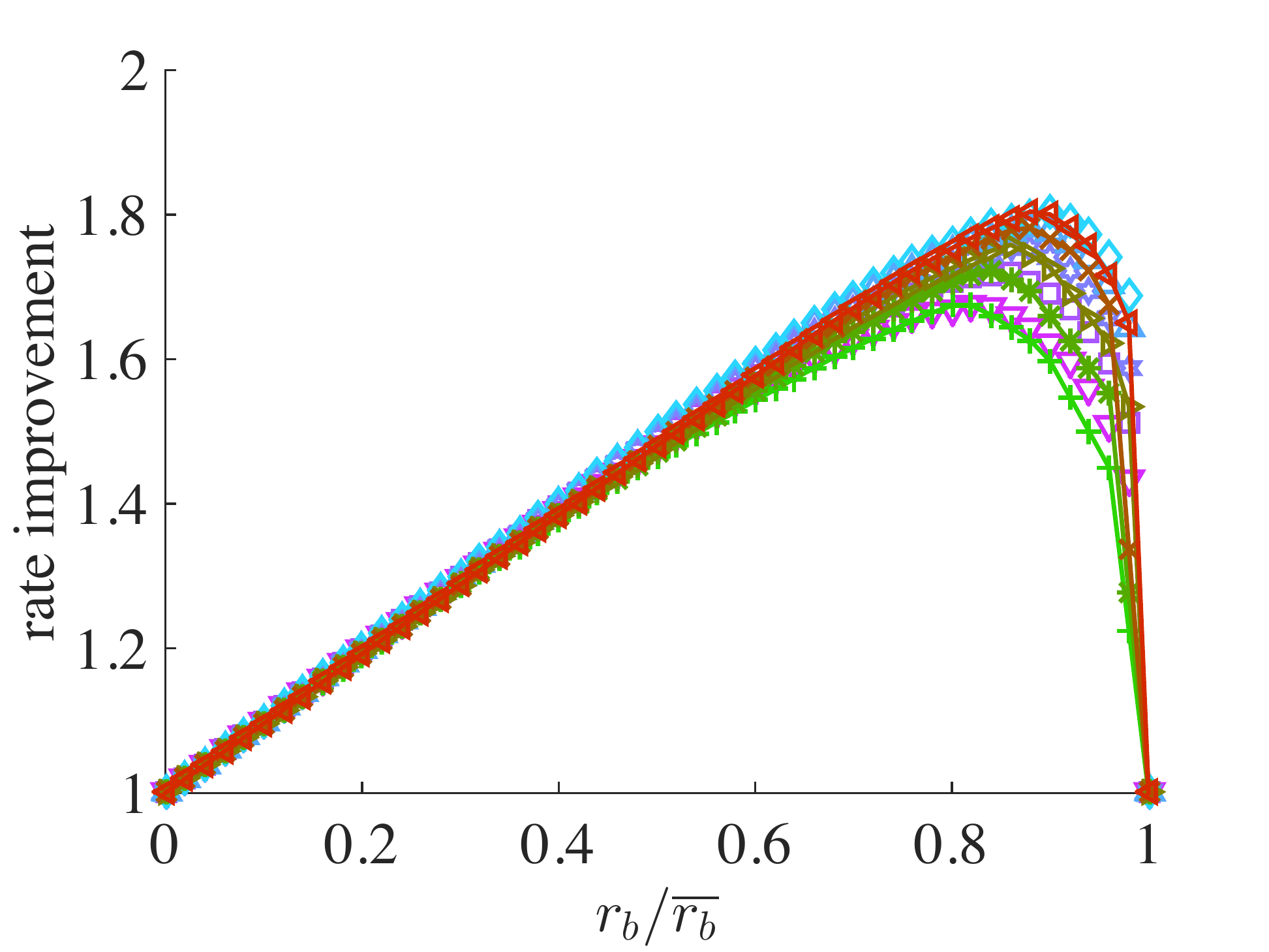}}\hspace{\fill}
\subfloat[]{\label{fig:cmph_rate_improve_FDE_2}\includegraphics[scale = 0.22]{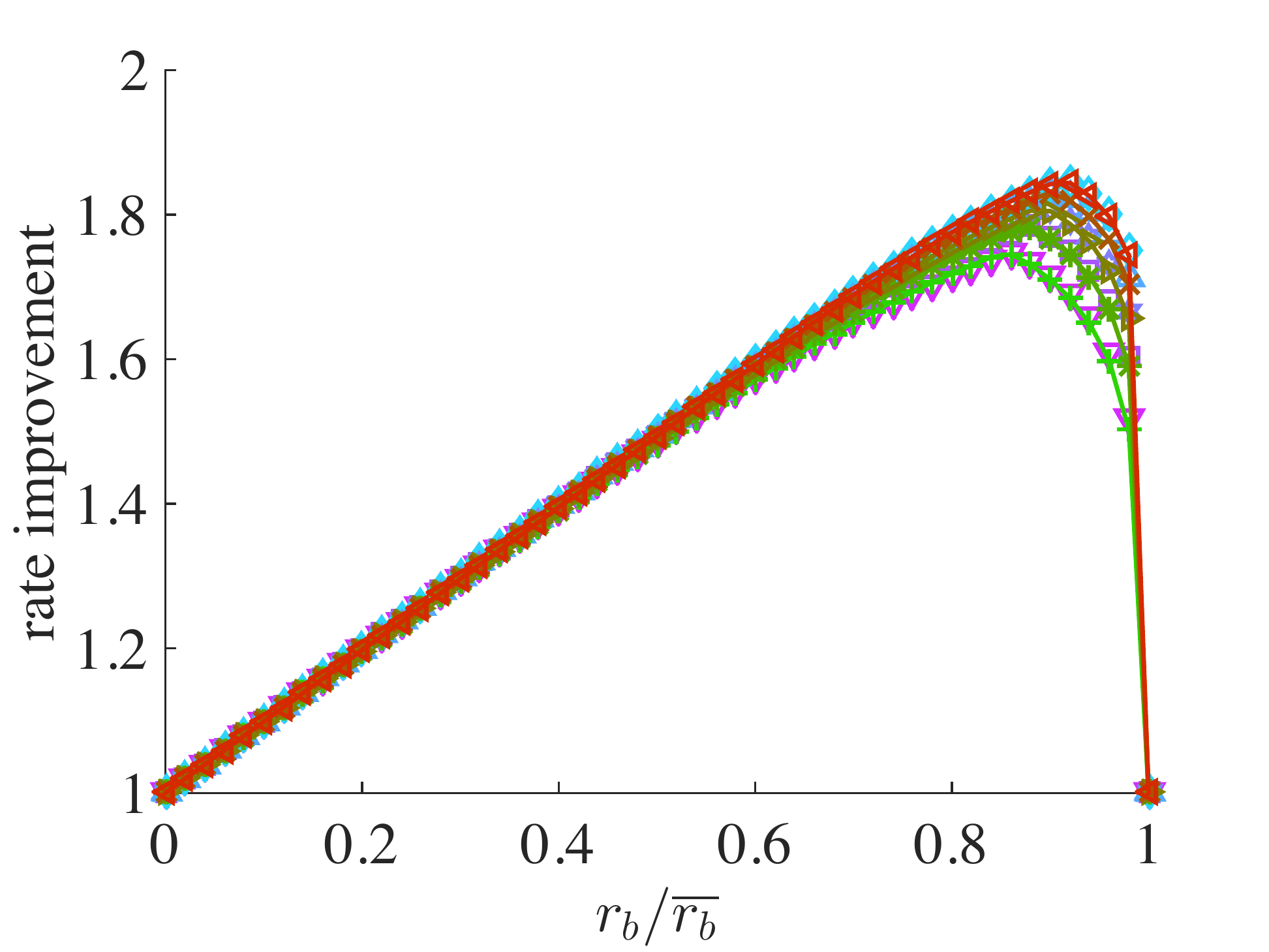}}\hspace{\fill}
\subfloat{\includegraphics[scale = 0.25]{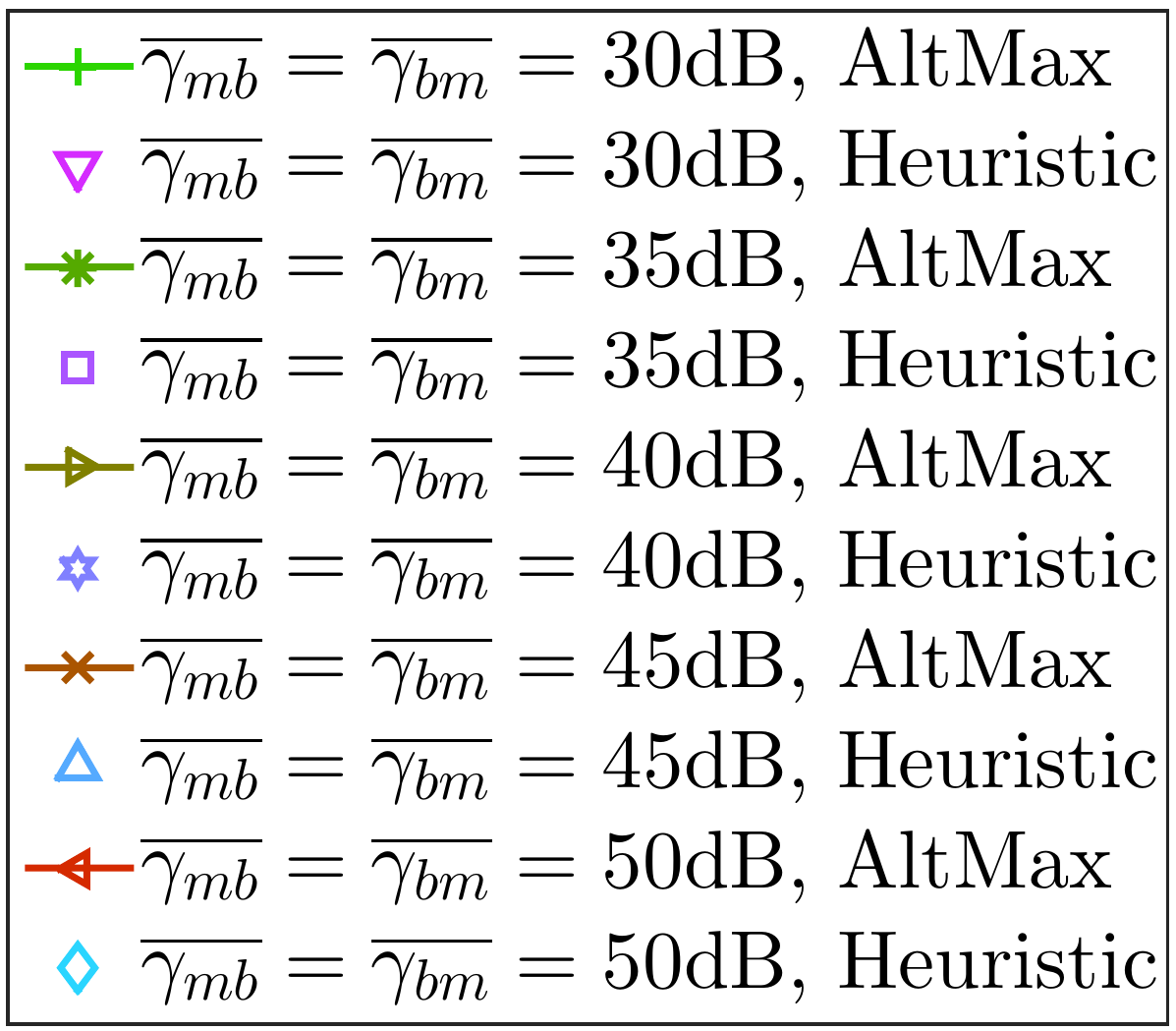}}\\\vspace{-10pt}
\setcounter{subfigure}{3}
\subfloat[]{\label{fig:cmp_rate_improve_conv}\includegraphics[scale = 0.22 ]{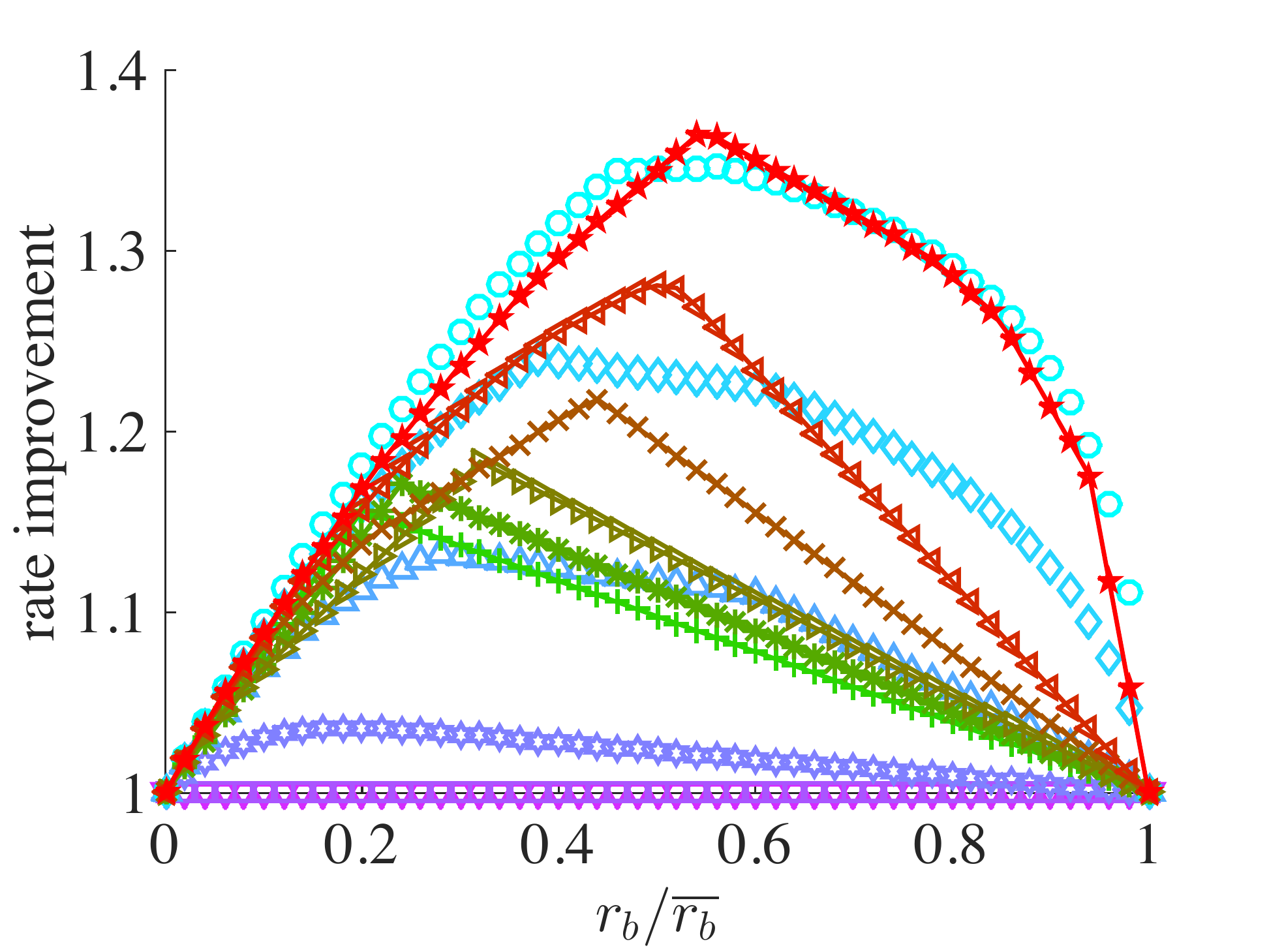}}\hspace{\fill}
\subfloat[]{\label{fig:cmp_rate_improve_FDE_1}\includegraphics[scale = 0.22]{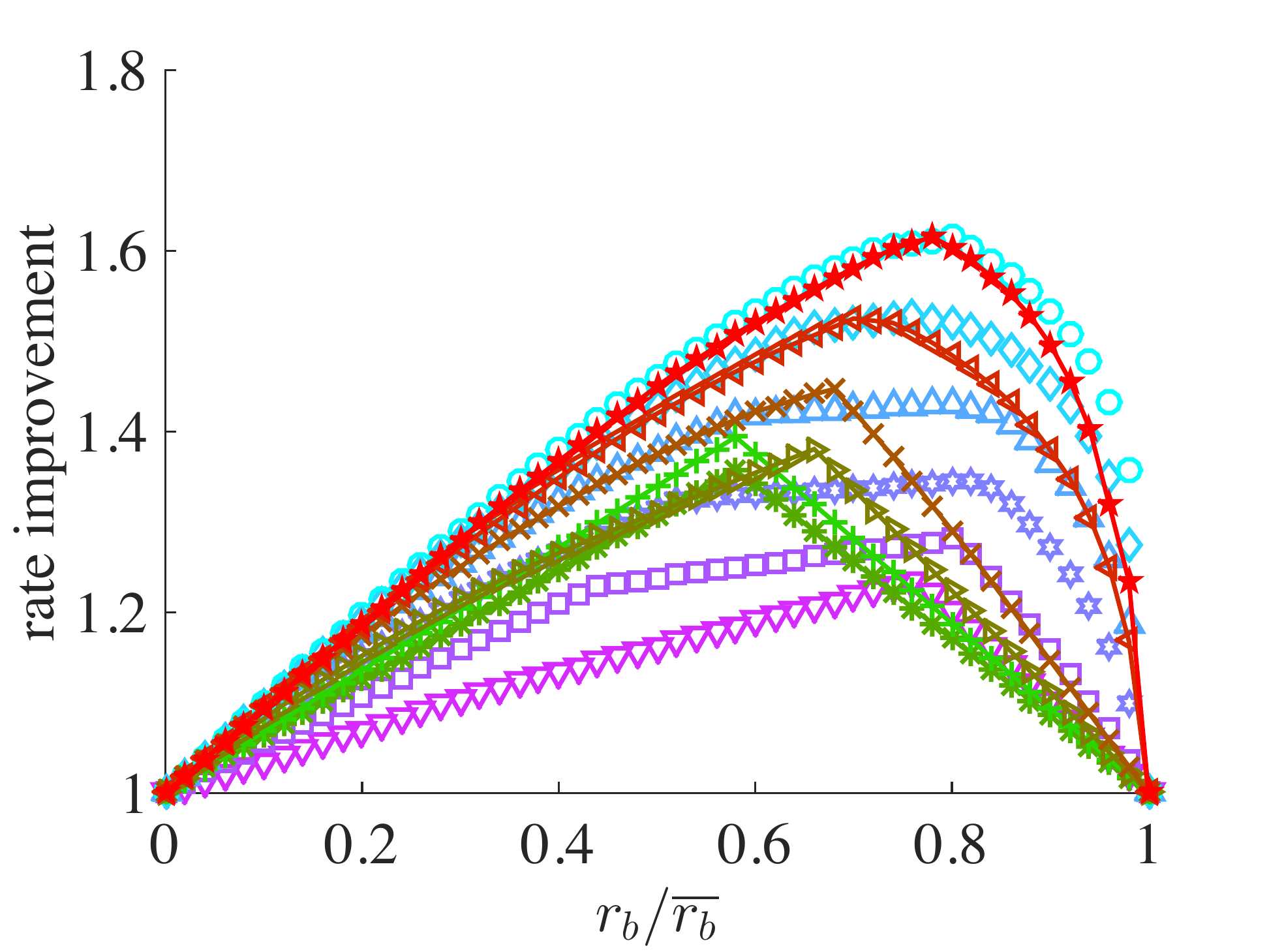}}\hspace{\fill}
\subfloat[]{\label{fig:cmp_rate_improve_FDE_2}\includegraphics[scale = 0.22]{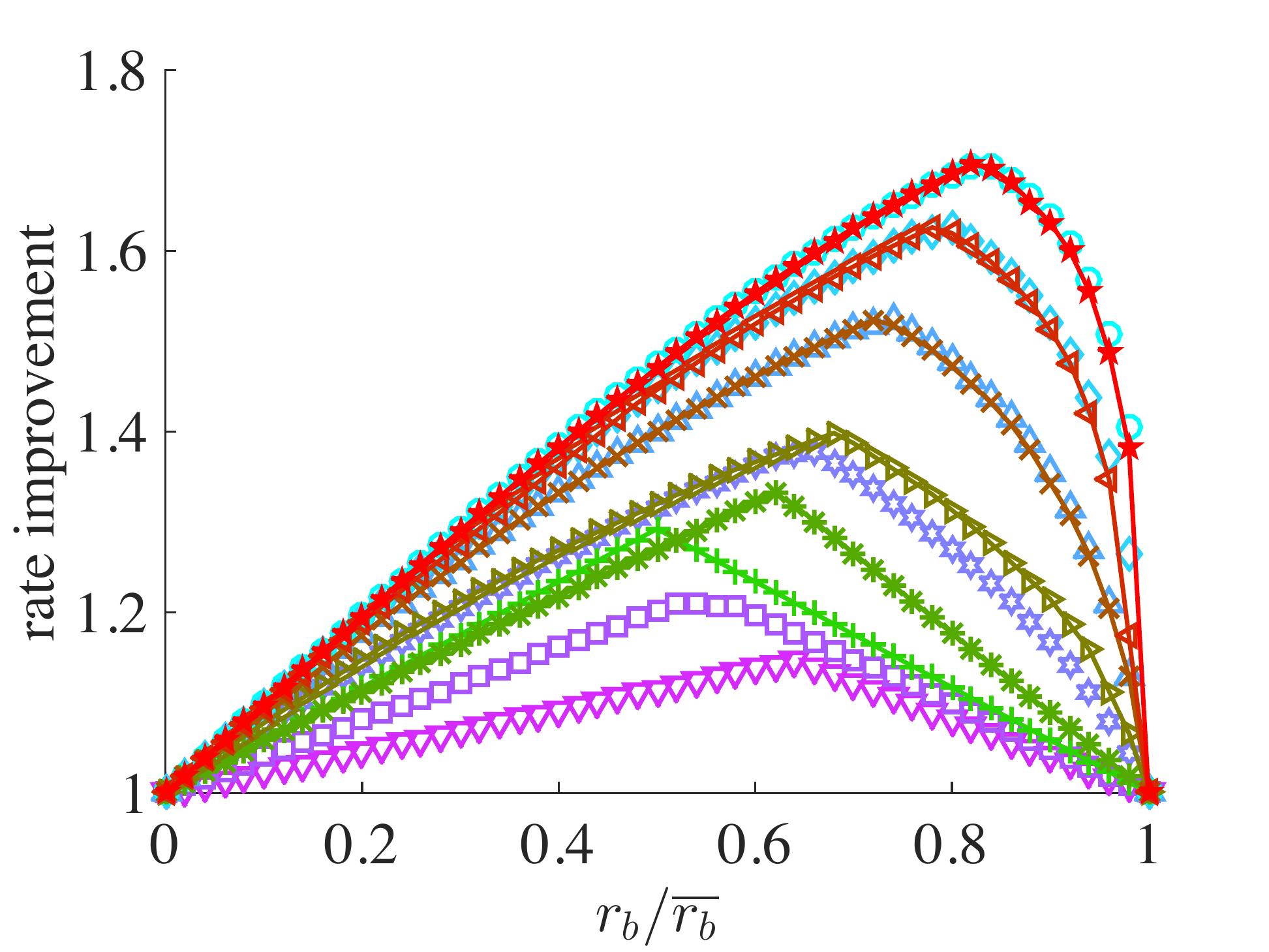}}\hspace{\fill}
\subfloat{\includegraphics[scale = 0.25]{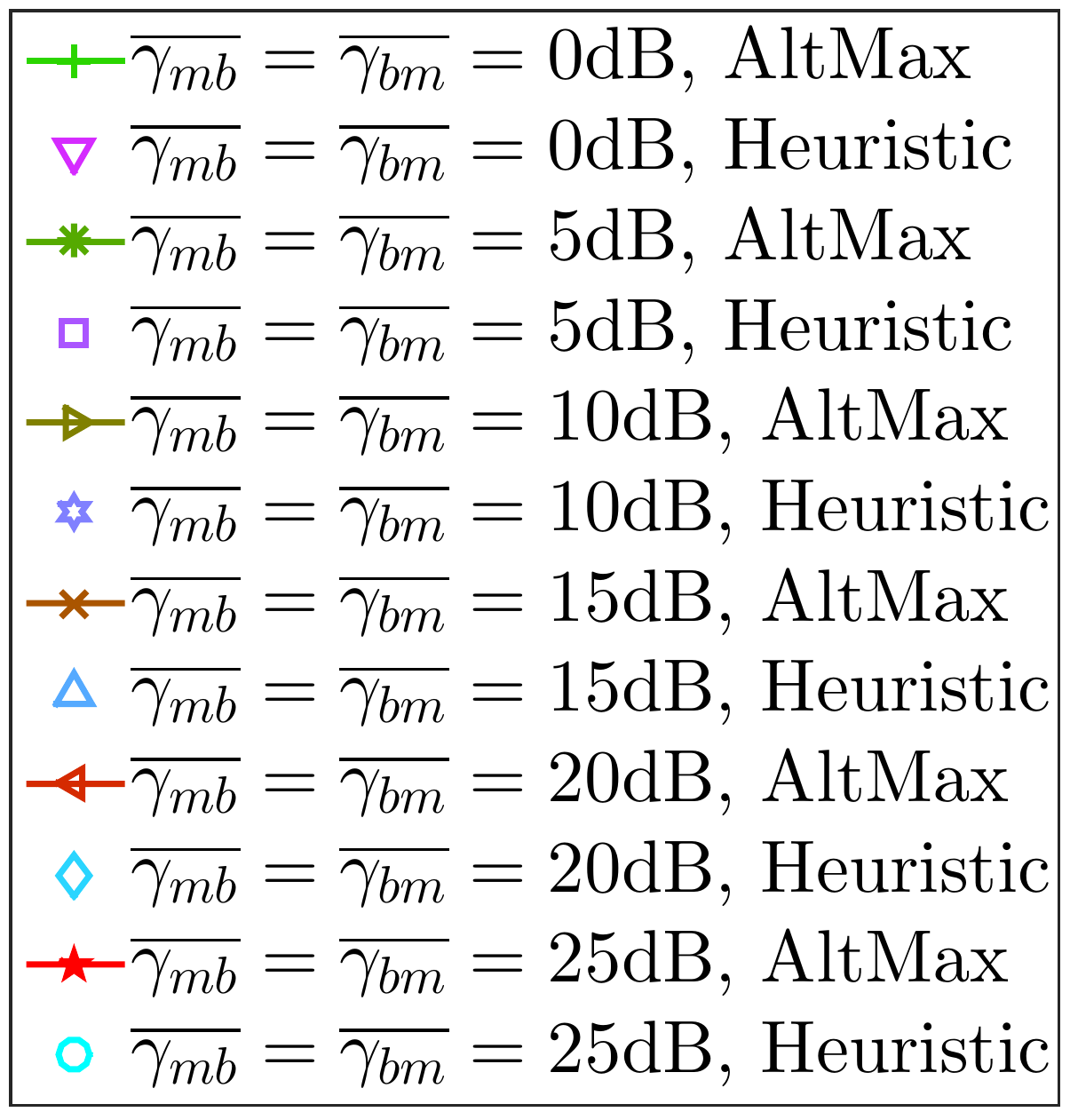}}\\\vspace{-10pt}
\setcounter{subfigure}{6}
\subfloat[]{\label{fig:cmp_asym_rate_improve_conv}\includegraphics[scale = 0.22 ]{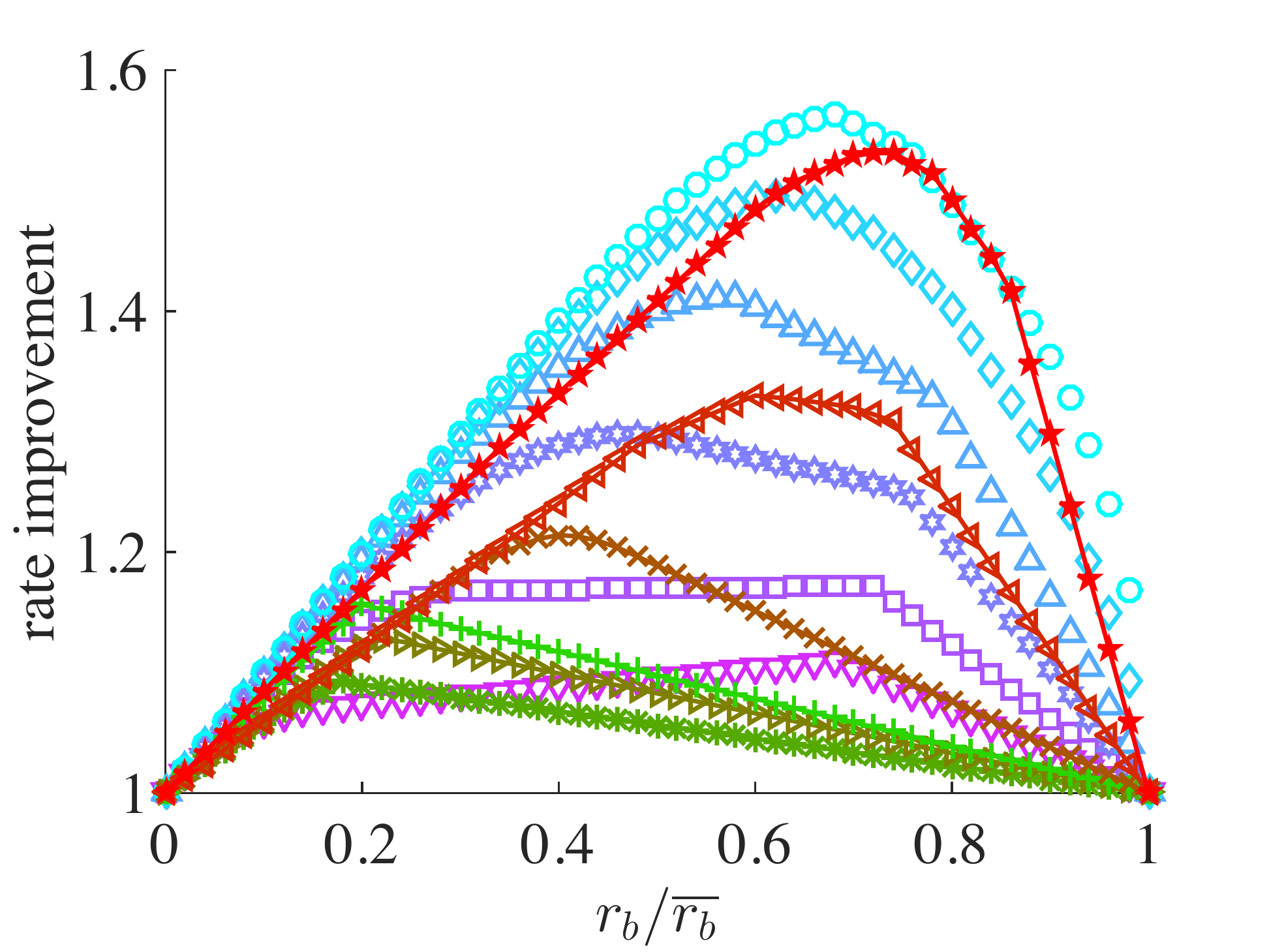}}\hspace{\fill}
\subfloat[]{\label{fig:cmp_asym_rate_improve_FDE_1}\includegraphics[scale = 0.22]{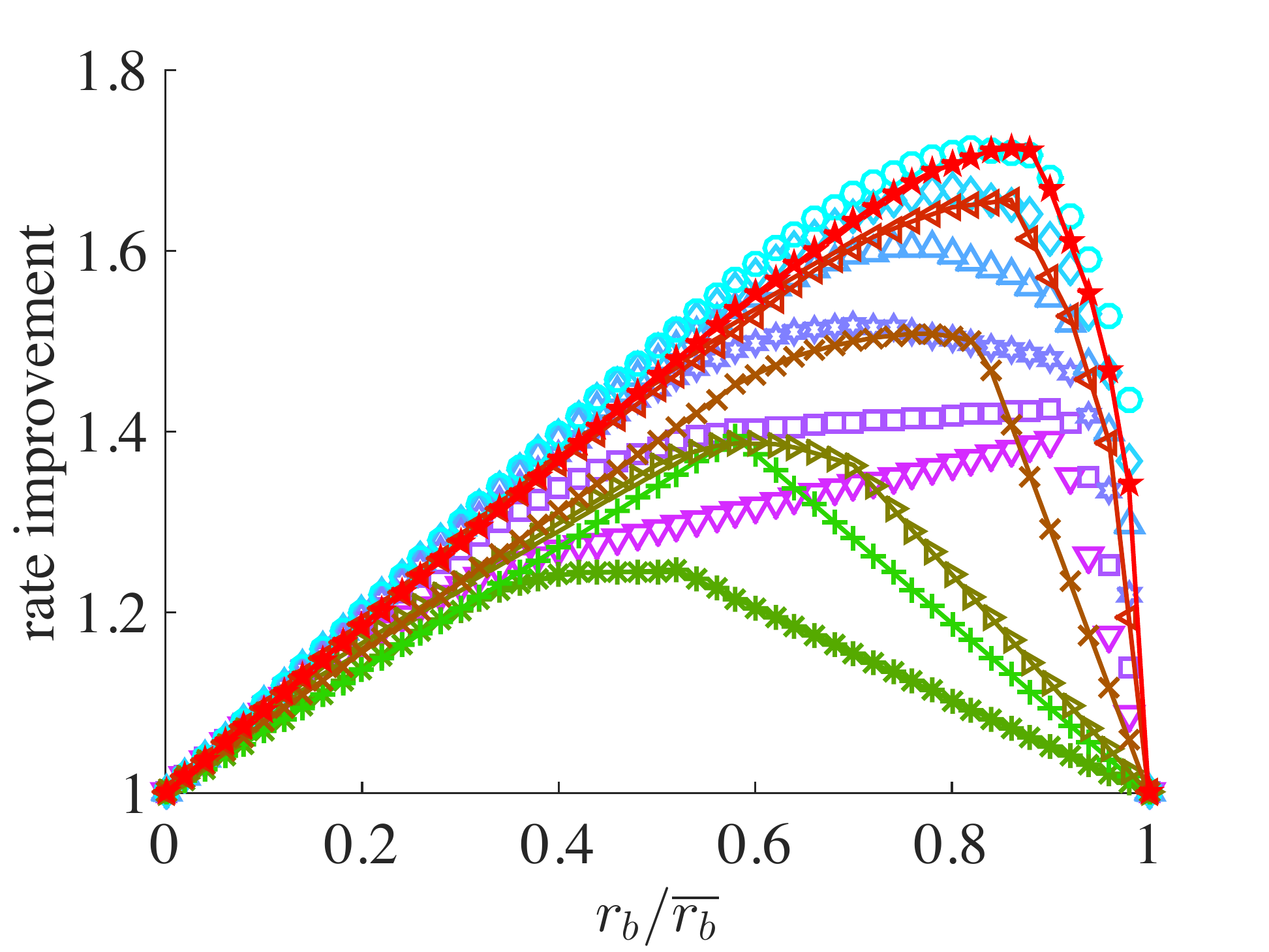}}\hspace{\fill}
\subfloat[]{\label{fig:cmp_asym_rate_improve_FDE_2}\includegraphics[scale = 0.22]{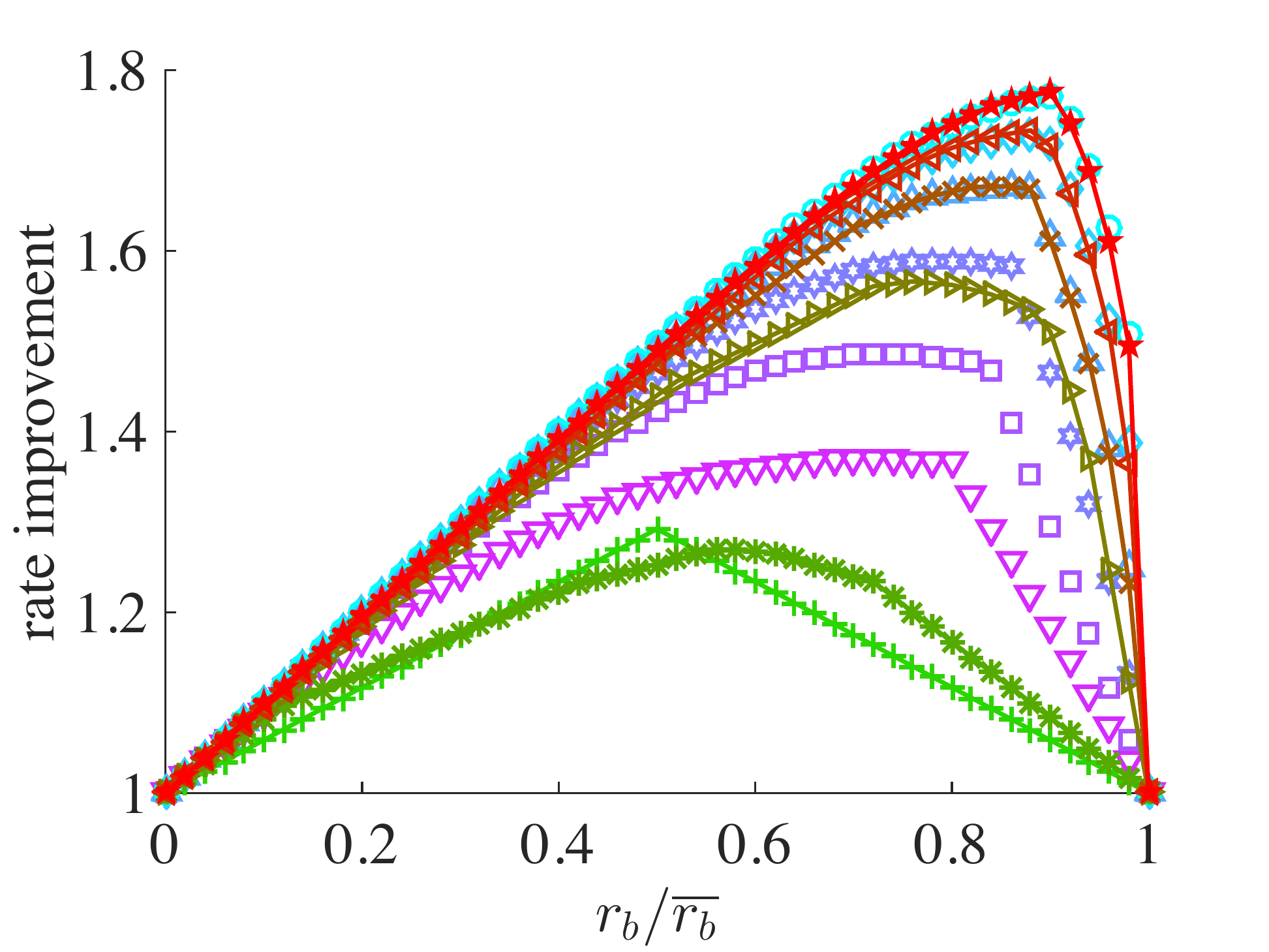}}\hspace{\fill}
\subfloat{\includegraphics[scale = 0.25]{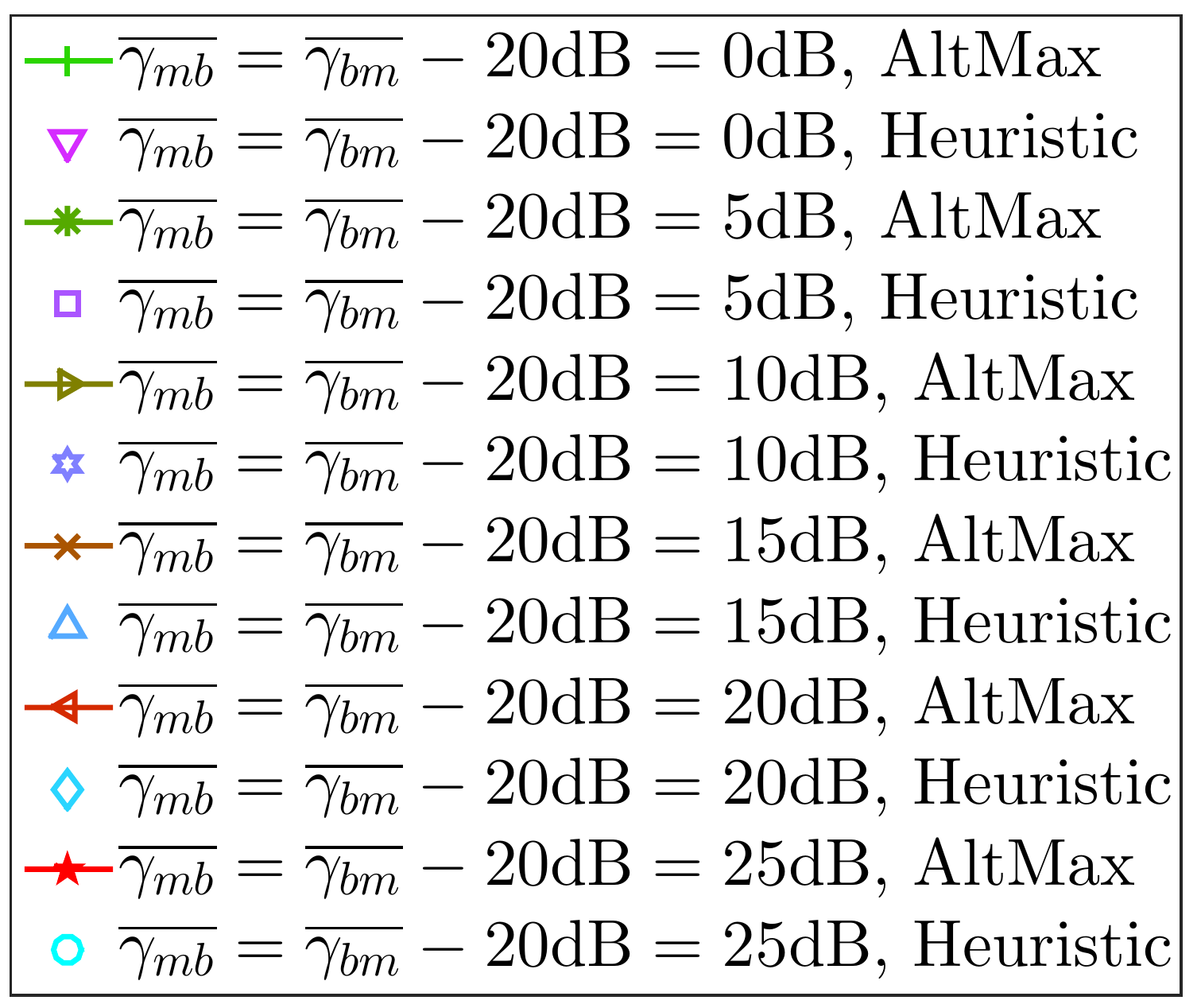}}\vspace{-10pt}
\caption{Rate improvements for $\overline{\gamma_{bb, k}}$ and $\overline{\gamma_{mm, k}}$ from Fig.~\ref{fig:cancellation-profiles}. The leftmost column of graphs corresponds to $\overline{\gamma_{mm, k}}$  from  Fig.~\ref{fig:cancellation-profiles}\protect\subref{fig:gamma_mm_conv}, the middle column corresponds to $\overline{\gamma_{mm, k}}$ from  Fig.~\ref{fig:cancellation-profiles}\protect\subref{fig:gamma_mm_FDE_1}, and the rightmost column corresponds to $\overline{\gamma_{mm, k}}$ from Fig.~\ref{fig:cancellation-profiles}\protect\subref{fig:gamma_mm_FDE_2}. $\overline{\gamma_{bb, k}}$ is selected according to Fig.~\ref{fig:cancellation-profiles}\protect\subref{fig:gamma_bb}. When rate improvements are at least $1.4\times$, the heuristic performs similar to or better than the alternating maximization.}\vspace{-10pt}
\label{fig:gpa-heuristic-comparison} 
\end{figure*}

\subsection{A Simple Power Allocation Heuristic}\label{section:heuristic}

Even though the algorithm described in the previous section will lead to the optimal or a near-optimal TDFD capacity region in many cases of interest, it may not be suitable for a real-time implementation.  This motivates us to develop a simple heuristic that performs well in most cases and is based on the observations we made while implementing the algorithms described in previous sections.

The intuition for the heuristic is that around the points $(0, \overline{r_m})$ and $(\overline{r_b}, 0)$, one of the two rates is very low, and the power allocation at the station with the high rate behaves as the optimal HD power allocation. When the SNR on each channel and at both stations is high compared to the XINR, the power allocation around the point $(s_b, s_m)$ has the shape of the power allocation in the high SINR approximation\footnote{See \cite{full-duplex-sigmetrics} for the high SINR approximation power allocation.}. When the SNR compared to the XINR is high on some channels, but not high on the other channels, then it may be better to use some of the channels with low SNR as HD. For practical implementations of compact FD transceivers, the channels with the higher XINR typically appear closer to the edges of the frequency band. The pseudocode of the heuristic for the case $r_b^* \leq s_b$ is provided in Algorithm \ref{algo:pa-heuristic} (\textsc{PA-Heuristic}) in the appendix. The pseudocode for the case $r_b^* > s_b$ is analogous to the $r_b^* \leq s_b$ case and is omitted. Here, $(s_b, s_m)$ is obtained as the rate pair that maximizes the sum rate under the high SINR approximation, as in \cite{full-duplex-sigmetrics}.

For the FD capacity region determined by the heuristic, we further run a convex hull computation algorithm \cite{cormen2009introduction} to determine the FD + TDD capacity region. The total running time is $O(NK^2 \log(\sum_k \overline{\gamma_{bb, k}}/(K\varepsilon)))$ for computing $N$ points on the FD capacity region boundary by using \textsc{PA-Heuristic}, plus additional $O(N)$ for convexifying the capacity region. Note that in practice $K$ and $N$ are at the order of 100, which makes this algorithm real-time.

The comparison of the rate improvement for FD + TDD operation determined by \textsc{PA-Heuristic} and the alternating maximization algorithm described in the previous section is shown in {Fig.~\ref{fig:gpa-heuristic-comparison}}. The results shown in {Fig.~\ref{fig:gpa-heuristic-comparison}} were obtained assuming that $\overline{\gamma_{bm, 1}} =\overline{\gamma_{bm, K}} ...\equiv K\overline{\gamma_{bm}}$, $\overline{\gamma_{mb, 1}}=...=\overline{\gamma_{mb, K}}\equiv K\overline{\gamma_{mb}}$, and $\overline{\gamma_{mm, k}}, \overline{\gamma_{bb, k}}$ from Fig.~\ref{fig:cancellation-profiles}.  
The alternating maximization algorithm can provide an optimal solution only when conditions (\ref{eq:C1}) and (\ref{eq:C2}) are non-restrictive, i.e., when $\overline{\gamma_{bm, k}}\geq \overline{\gamma_{bb, k}}(1+\overline{\gamma_{mm, k}})$ and $\overline{\gamma_{mb, k}}\geq \overline{\gamma_{mm, k}}(1+\overline{\gamma_{bb, k}})$, $\forall k$. For $\overline{\gamma_{bb, k}}$ from Fig.~\ref{fig:cancellation-profiles}\subref{fig:gamma_bb} and $\overline{\gamma_{mm, k}}$ from Fig.~\ref{fig:cancellation-profiles}\subref{fig:gamma_mm_conv}, \subref{fig:gamma_mm_FDE_1}, and \subref{fig:gamma_mm_FDE_2}, (\ref{eq:C1}) and (\ref{eq:C2}) are non-restrictive when (i) $\overline{\gamma_{bm}}\geq 39.1$dB, $\overline{\gamma_{mb}}\geq 39.2$dB, (ii) $\overline{\gamma_{bm}}\geq 32.8$dB, $\overline{\gamma_{mb}}\geq 32.3$dB, and (iii) $\overline{\gamma_{bm}}\geq 25.3$dB, $\overline{\gamma_{mb}}\geq 25.3$dB, respectively.

As {Fig.~\ref{fig:gpa-heuristic-comparison}}\subref{fig:cmph_rate_improve_conv}--\subref{fig:cmph_rate_improve_FDE_2} shows, when (\ref{eq:C1}) and (\ref{eq:C2}) are non-restrictive, the alternating maximization algorithm and the \textsc{PA-Heuristic} provide almost identical results (minor differences are mainly due to a numerical error in computation). Moreover, when the smallest upper bound on $\alpha_{b, k}$'s and $\alpha_{m, k}$'s imposed by (\ref{eq:C1}) and (\ref{eq:C2}) is no higher than $5/K$, i.e., for (i) $\overline{\gamma_{bm}}\geq 28.9$dB, $\overline{\gamma_{mb}}\geq 29.7$dB, (ii) $\overline{\gamma_{bm}}\geq 22.6$dB, $\overline{\gamma_{mb}}\geq 23.4$dB, and (iii) $\overline{\gamma_{bm}}\geq 15.2$dB, $\overline{\gamma_{mb}}\geq 15.9$dB, for $\overline{\gamma_{mm, k}}$ from Fig.~\ref{fig:cancellation-profiles}\subref{fig:gamma_mm_conv}, \subref{fig:gamma_mm_FDE_1}, and \subref{fig:gamma_mm_FDE_2}, respectively, the differences between the alternating maximization algorithm and the \textsc{PA-Heuristic} are still negligible (Fig.~\ref{fig:gpa-heuristic-comparison}\subref{fig:cmph_rate_improve_conv}--\subref{fig:cmp_asym_rate_improve_FDE_2}).

When (\ref{eq:C1}) and (\ref{eq:C2}) are restrictive (Fig.~\ref{fig:gpa-heuristic-comparison}\subref{fig:cmp_rate_improve_conv}--\subref{fig:cmp_asym_rate_improve_FDE_2}), all following cases may happen: (i) the alternating maximization outperforms the \textsc{PA-Heuristic}, (ii) the \textsc{PA-Heuristic} outperforms the alternating maximization, and (iii) both have similar performance. Case (i) typically happens when most channels are allocated as HD by the alternating maximization, with some of them allocated to the BS, and others to the MS. In this case the rate improvements predominantly come from using higher total irradiated power compared to TDD, rather than from using full-duplex. Note that the \textsc{PA-Heuristic} allows the HD channels to be assigned either only to the BS or only to the MS, but not both. Case (ii) happens when (\ref{eq:C1}) and (\ref{eq:C2}) restrict the part of the feasible region where high rate improvements are possible; namely, when either both $\overline{\gamma_{bm}}$ and $\overline{\gamma_{mb}}$ are low, or when $\overline{\gamma_{bm}}$ is much ($20$dB) higher than $\overline{\gamma_{mb}}$.

\section{Conclusion}\label{sec:conclusion}
We presented a theoretical study of the capacity region of FD in both the single and multi-channel cases. We developed algorithms that not only allow characterizing the region but can also be used for asymmetrical rate allocation. We numerically demonstrated the gains from FD. 

While significant attention has been given to resource allocation in HD OFDM networks (e.g., \cite{HSAB2009} and references therein), as we demonstrated, the special characteristics of FD pose many new challenges. In particular, the design of MAC protocols that support the co-existence of HD and FD users while providing fairness is an open problem. Moreover, there is a need for experimental evaluation of scheduling, power control, and channel allocation algorithms tailored for the special characteristics of FD. 

\section{Acknowledgements}
This work was supported in part by the NSF grant ECCS-1547406, the Qualcomm Innovation Fellowship, and the People Programme (Marie Curie Actions) of the European Union's Seventh Framework Programme (FP7/2007-2013) under REA grant agreement n${^{\text{o}}} $[PIIF-GA-2013-629740].11. We thank Jin Zhou and Harish Krishnaswamy for providing us with the cancellation data from \cite{Zhou_WBSIC_ISSCC15}.

\bibliographystyle{abbrv}
\bibliography{references_FD}
\newpage
\appendix{}
\vspace{-10pt}\begin{algorithm}
\caption{\textsc{PA-Heuristic}($K, r_b^*$)}
\begin{algorithmic}[1]
\State Input: $\{\overline{\gamma_{bm, k}}, \overline{\gamma_{mb, k}}, \overline{\gamma_{mm, k}}, \overline{\gamma_{bb, k}}\}$
\State $\{\alpha_{b, k}^L\} = \arg\{\overline{r_b}\}$, $\{\alpha_{m, k}^L\} = \arg\{\overline{r_m}\}$
\State $\{\alpha_{b, k}^H\}, \{\alpha_{b, k}^H\} = \arg\{s_b + s_m\}$
\State $f_1 = \mathrm{true}$, $f_2 = \mathrm{true}$, $k=1$
\If{$r_b^*\leq s_b$}
\State $j=0$, $\{\alpha_{b, k}^1\}=\{\alpha_{b, k}^L\}$, $\{\alpha_{b, k}^2\}=\{\alpha_{b, k}^H\}$
\While{$j \leq K/2$ \textbf{and} ($f_1$ \textbf{or} $f_2$)}
\State $\{\overline{\gamma_{bm, k}^1}, \overline{\gamma_{mb, k}^1}, \overline{\gamma_{mm, k}^1}, \overline{\gamma_{bb, k}^1}\}$ = \textsc{Scale}($\{\alpha_{b, k}^1, \alpha_{m, k}^L\}$)
\State $r_m^1$ = \textsc{MCFind-}$r_m$($r_b^*, K$) for input above
\State $\{\overline{\gamma_{bm, k}^2}, \overline{\gamma_{mb, k}^2}, \overline{\gamma_{mm, k}^2}, \overline{\gamma_{bb, k}^2}\}$ = \textsc{Scale}($\{\alpha_{b, k}^2, \alpha_{m, k}^H\}$)
\State $r_m^2$ = \textsc{MCFind-}$r_m$($r_b^*, K$) for input above
\If{$j = 0$}
\State $r_m^* = \max\{r_m^1, r_m^2\}$
\Else
\State $\{\alpha_{b, k}^t\} = \{\alpha_{b, k}^1\}/(\sum_k \alpha_{b, k}^1)$, $\alpha_{b, j}^t = 0$
\State \textbf{if }{\parbox[t]{\dimexpr\linewidth-\algorithmicindent}{$r_b(\{\alpha_{b, k}^t\}, \{\alpha_{m, k}^L\}) \geq r_b^*$ \textbf{and} \textsc{MCFind-}$r_m$($r_b^*, K$)$ > r_m^*$ with input = $\{\overline{\gamma_{bm, k}^1}, \overline{\gamma_{mb, k}^1}, \overline{\gamma_{mm, k}^1}, \overline{\gamma_{bb, k}^1}\}$}\strut}
\State \textbf{then}
\State $\>\>\>\>$ $r_m^* = $\textsc{MCFind-}$r_m$($r_b^*, K$), $\{\alpha_{b, k}^1\} = \{\alpha_{b, k}^t\}$
\State \textbf{else}
 \;$f_1 = \mathrm{false}$
\State $\{\alpha_{b, k}^t\} = \{\alpha_{b, k}^2\}/(\sum_k \alpha_{b, k}^2)$, $\alpha_{b, K-j+1}^t = 0$
\State \textbf{if }\parbox[t]{\dimexpr\linewidth-\algorithmicindent}{$r_b(\{\alpha_{b, k}^t\}, \{\alpha_{m, k}^L\}) \geq r_b^*$ \textbf{and} \textsc{MCFind-}$r_m$($r_b^*, K$)$ > r_m^*$ with input = $\{\overline{\gamma_{bm, k}^2}, \overline{\gamma_{mb, k}^2}, \overline{\gamma_{mm, k}^2}, \overline{\gamma_{bb, k}^2}\}$}
\State \textbf{then}
\State $\>\>\>\>$ $r_m^* = $\textsc{MCFind-}$r_m$($r_b^*, K$), $\{\alpha_{b, k}^2\} = \{\alpha_{b, k}^t\}$
\State \textbf{else}
 \; $f_2 = \mathrm{false}$
\EndIf
\EndWhile
\Else
\State $\dots$
\EndIf
\end{algorithmic}\label{algo:pa-heuristic}
\end{algorithm}\vspace{-18pt}
\begin{algorithm}
\caption{\textsc{Scale}($\{\alpha_{b, k}, \alpha_{m, k}\}$)}
\begin{algorithmic}[1]
\State Input: $\{\overline{\gamma_{bm, k}}, \overline{\gamma_{mb, k}}, \overline{\gamma_{mm, k}}, \overline{\gamma_{bb, k}}\}$
\For{$k = 1$ \textbf{to} $K$} 
\State $\overline{\gamma_{bm, k}}^s = K\alpha_{b, k}\overline{\gamma_{bm, k}}$, $\overline{\gamma_{mb, k}}^s = K\alpha_{m, k}\overline{\gamma_{mb, k}}$ 
\State $\overline{\gamma_{mm, k}}^s = K\alpha_{m, k}\overline{\gamma_{mm, k}}$, $\overline{\gamma_{bb, k}}^s = K\alpha_{b, k}\overline{\gamma_{bb, k}}$
\EndFor
\Return  $\{\overline{\gamma_{bm, k}}^s, \overline{\gamma_{mb, k}}^s, \overline{\gamma_{mm, k}}^s, \overline{\gamma_{bb, k}}^s\}$
\end{algorithmic}
\end{algorithm}
\end{document}